\documentclass{article}
\usepackage{amsfonts}
\usepackage{amsmath}
\usepackage{amssymb}
\usepackage{amsthm}
\usepackage{array}
\usepackage{color}
\usepackage{comment}
\usepackage{enumerate}
\usepackage[mathscr]{euscript}
\usepackage{graphicx}
\usepackage[utf8]{inputenc}
\usepackage{latexsym}
\usepackage{mathrsfs}
\usepackage{mathtools}
\usepackage{appendix}
\usepackage{multirow}
\usepackage{subfigure}
\usepackage[a4paper, top=2cm, bottom=2.5cm, left=2.5cm, right=2.5cm]{geometry}
\usepackage{tikz}
\usepackage[font=small]{caption}
\usepackage[square,sort,comma,numbers]{natbib}
\RequirePackage[colorlinks,citecolor=blue,urlcolor=blue]{hyperref}
\usepackage[american]{babel}
\newcommand{\R}{\mathbb{R}}

\newcommand{\B}{\mathcal{B}}

\newcommand{\F}{\mathcal{F}}

\newcommand{\be}{\begin{equation}}
\newcommand{\ee}{\end{equation}}
\newcommand{\bea}{\begin{eqnarray}}
\newcommand{\eea}{\end{eqnarray}}
\newcommand{\beas}{\begin{eqnarray*}}
\newcommand{\eeas}{\end{eqnarray*}}
\newcommand{\ds}{\displaystyle}

\def\theequation{\thesection.\arabic{equation}}
\newtheorem{theorem}{Theorem}

\newtheorem{corollary}[theorem]{Corollary}

\newtheorem{definition}[theorem]{Definition}
\newtheorem{example}[theorem]{Example}

\newtheorem{lemma}[theorem]{Lemma}

\newtheorem{proposition}[theorem]{Proposition}
\newtheorem{remark}[theorem]{Remark}

\DeclareMathOperator*{\esssup}{ess\,sup}
\DeclareMathOperator*{\essinf}{ess\,inf}
\DeclareMathOperator*{\essmax}{ess\,max}
\DeclareMathOperator*{\essmin}{ess\,min}
\DeclareMathOperator*{\argmin}{arg\,min}

\begin{document}

\title{\vskip -1cm Dynamic Return and Star-Shaped Risk Measures via BSDEs\thanks{We are very grateful to Fabio Bellini, Freddy Delbaen, Marco Frittelli and Martin Schweizer for useful comments and discussions.
This research was funded in part by the Netherlands Organization for Scientific Research under grant NWO Vici 2020--2027 (Laeven). Emanuela Rosazza Gianin and Marco Zullino are members of Gruppo Nazionale per l’Analisi Matematica, la Probabilità e le loro Applicazioni (GNAMPA), Italy.}}
\author{Roger J.~A.~Laeven \\
{\footnotesize Dept.~of Quantitative Economics}\\
{\footnotesize University of Amsterdam, CentER}\\
{\footnotesize and EURANDOM, The Netherlands}\\
{\footnotesize \texttt{r.j.a.laeven@uva.nl}}\\
\and Emanuela Rosazza Gianin\footnote{Corresponding author.} \\
{\footnotesize Dept.~of Statistics and Quantitative Methods}\\
{\footnotesize University of Milano Bicocca, Italy}\\
{\footnotesize \texttt{emanuela.rosazza1@unimib.it}}\\
\and Marco Zullino \\
{\footnotesize Dept.~of Mathematics and Applications}\\
{\footnotesize University of Milano Bicocca, Italy}\\
{\footnotesize \texttt{m.zullino@campus.unimib.it}}}

\date{This Version: \today }

\maketitle

\begin{abstract}
This paper establishes characterization results for dynamic return and star-shaped risk measures induced via backward stochastic differential equations (BSDEs).
We first characterize a general family of static star-shaped functionals in a locally convex Fr\'echet lattice.
Next, employing the Pasch-Hausdorff envelope, we build a suitable family of convex drivers of BSDEs inducing a corresponding family of dynamic convex risk measures of which the dynamic return and star-shaped risk measures emerge as the essential minimum.
Furthermore, we prove that if the set of star-shaped supersolutions of a BSDE is not empty, then there exists, for each terminal condition, at least one convex BSDE with a non-empty set of supersolutions, yielding the minimal star-shaped supersolution.
We illustrate our theoretical results in a few examples and demonstrate their usefulness in two applications, to capital allocation and portfolio choice.\\[3mm]
\end{abstract}

\section{Introduction}

The intimate connection between dynamic convex and, more generally, monetary risk measures and Backward Stochastic Differential Equations (BSDEs)
has been studied in a large theoretical literature (e.g., \cite{B76,PP90,P97,EPQ97,P04,P05,BEK05,RG06,J08,LS14,DPR10,DNRG23,TW23})
and has proven to be highly useful in a wide variety of applications in applied probability and stochastic control, financial and insurance mathematics, and operations research (e.g., \cite{EPQ97,CE02,LQ03,MS05,ELKR09,LS14,KLLSS18,RZ23}).

Recently, \cite{BLR18} introduced return risk measures (i.e., risk measures that are monotone and  positively homogeneous)
to stand on par with the class of monetary risk measures.
Return risk measures provide a relative (or geometric) evaluation of risk,
whereas monetary risk measures provide absolute (or arithmetic) risk assessments,
resonant with the distinction between relative and absolute risk aversion.
In \cite{LR22}, the positive homogeneity property of return risk measures was replaced by the more general notion of star-shapedness
(and log-convexity by quasi-logconvexity).

In this paper, we establish characterization results for dynamic return and star-shaped risk measures induced via BSDEs.
That is, we uncover new links between BSDEs and dynamic return and star-shaped risk measures.
These are our main results, the mathematical details of which are delicate.
To accomplish this, we first characterize a general family of static star-shaped functionals as a contribution of independent interest.

The first main challenge in devising our proof strategies is the absence of translation invariance (i.e., cash-additivity).
As a result, even in the static case, we cannot use techniques similar to those employed in the existing literature, where this property is heavily relied upon.
Furthermore, to achieve generality, we first analyze the static case in a locally convex Fr\'echet lattice, adopting the approach used for convex risk measures in \cite{BF10}.
This represents a novel approach in the literature on star-shaped risk measures.
Additionally, as a by-product, we prove a representation result for general static star-shaped functionals without requiring monotonicity.
These results may be considered of interest in their own right, not previously known in the classical theory of non-convex analysis.

The second major challenge concerns the dynamic representation of star-shaped risk measures induced via BSDEs.
In particular, we build a suitable family of convex drivers of BSDEs that generate a corresponding family of convex risk measures, of which the essential minimum yields the star-shaped risk measure of interest.
We construct such drivers by means of the Pasch-Hausdorff envelope (see \cite{RW97}).
Our approach introduces a fresh perspective in the literature, by employing the Pasch-Hausdorff envelope to regularize the drivers.
This regularization ensures the existence and uniqueness of solutions under Lipschitz assumptions while preserving the desired properties.
The work of \cite{BEK05} explores a similar concept, albeit with a different objective.
By leveraging this technique, we extend our representation results from the static case to the dynamic setting.
Furthermore, we demonstrate the effectiveness of our approach in the setting of star-shaped supersolutions to BSDEs.
In this particular setting, our main result asserts that if the collection of star-shaped supersolutions to a BSDE is not devoid of elements, then a set of BSDEs generated by convex drivers can be identified such that the essential minimum of the corresponding minimal convex supersolutions represents the minimal star-shaped supersolution to the original BSDE.
In particular, for any given terminal condition, there exists at least one BSDE for which the set of convex supersolutions is not empty, and the minimal convex supersolution of that BSDE achieves the minimum in the representation of the star-shaped risk measure.
Typically, the existence of a supersolution for a convex BSDE requires additional assumptions.
However, in our specific setting, this  is automatically guaranteed as long as there exists a non-empty set of star-shaped supersolutions.

We provide examples to demonstrate the generality of our framework and to show that our dynamic risk measures may arise naturally.
Our theoretical results are amenable to a broad range of applications.
We illustrate two such applications: capital allocation rules and optimal portfolio choice.

Perhaps most closely related to our work is \cite{BEK05,J08}, in which representation results are established for dynamic monetary convex risk measures induced by BSDEs, and \cite{C00}, in which a dual representation is derived for positively homogeneous functions.
In interesting concurrent work, \cite{MR22} studies the connection between monetary and monetary star-shaped risk measures in a static setting (see also the earlier \cite{CDFM15,CCMTW22,JXZ20}), while \cite{TW23} analyzes dynamic monetary and dynamic monetary star-shaped risk measures.
A key distinction between \cite{MR22,TW23} and our work, is that \cite{MR22,TW23} restrict attention to the case of monetary risk measures in which the property of cash-additivity is assumed.
By contrast, we do not in general impose cash-additivity, neither in the static setting (of Section~\ref{sec:star}) nor in the dynamic setting (of Sections~\ref{sec:mr}--\ref{sec:supsolution}).
The absence of cash-additivity not only substantially generalizes the scope of the representation theorems established in this paper, but also requires proof strategies that differ, both in the static and the dynamic environment, from those employed in the (large) literature on monetary risk measures, as explained above.
Furthermore, other distinctions include that in the static setting we consider a general locally convex Fr\'echet lattice and allow for non-monotonicity, whereas in the dynamic environment we explicitly characterize the BSDE driver using the Pasch-Hausdorff envelope and exploit supersolutions of BSDEs.
In addition, we derive an explicit expression for the penalty function associated with the min-max representation of star-shaped risk measures induced through BSDEs.
Remarkably, this representation was previously unknown even in the realm of dynamic cash-additive star-shaped risk measures.
In the recent \cite{RM22}, star-shapedness has also been analyzed in the context of static deviation measures, satisfying translation insensitivity.

This paper is organized as follows.
In Section~\ref{sec:prel}, we introduce the setting and recall some preliminaries for dynamic risk measures and BSDEs.
In Section~\ref{sec:con}, we provide representation results for dynamic convex return risk measures.
Section~\ref{sec:star} provides general characterization results for star-shaped risk measures.
Section~\ref{sec:mr} contains our main results.
In Section~\ref{sec:supsolution}, we provide a min-max representation of minimal supersolutions for star-shaped BSDEs.
In Section~\ref{sec:examples}, we develop three examples.
Section~\ref{sec:app} contains applications to capital allocation and portfolio choice.
All the proofs are provided in \mbox{Appendix \ref{sec:appendix}.}
\setcounter{equation}{0}

\section{Preliminaries}\label{sec:prel}

\subsection{Dynamic monetary and return risk measures}
In this subsection, we first present the definitions of the main functional spaces used in this paper, then we present the basic definitions of dynamic monetary and return risk measures, and next we provide some related relevant definitions and results.

Let $T>0$ and $(\Omega,\mathcal{F},\mathbb{P},(\mathcal{F}_t)_{t\in[0,T]})$ be a filtered probability space such that $\mathcal{F}_0$ is the trivial $\sigma$-algebra and $\mathcal{F}_T=\mathcal{F}$.
Furthermore, let $L^{p}\left(\mathcal{F}_{t}\right)=L^{p}\left(\Omega,\mathcal{F}_{t},\mathbb{P}\right)$ (with $t\in \left[ 0,T\right]$) denote the space of all real-valued, $\mathcal{F}_{t}$-measurable and $p$-integrable random variables endowed with the $L^{p}$-norm topology, for any $p\in[1,+\infty)$, and let $L^{\infty}(\mathcal{F}_t)$ be the space of essentially bounded random variables equipped with the topology $\sigma(L^1,L^{\infty})$.
The space $L^p_{+}(\mathcal{F}_t)$ (resp.\ $L^p_{++}(\mathcal{F}_t)$) is the restriction of $L^p(\mathcal{F}_t)$ to non-negative (resp.\ strictly positive) random variables.
Let $L_{\mathcal{F}}^{2}\left( T;\mathbb{R}^{n}\right)$ denote the space of all $\mathbb{R}^{n}$-valued, adapted processes $\left(X_{t}\right) _{t\in \left[ 0,T\right] }$ such that $\mathbb{E}\left[\int_{0}^{T}\left| X_{t}\right| ^{2}\text{ }dt\right] <+\infty$, where $\left| \cdot \right| $ is the Euclidean norm on $\mathbb{R}^{n}$.
For the Euclidean scalar product between $(a, b) \in \mathbb{R}^n$, we interchangeably use the notations $\langle a, b \rangle$ and $a \cdot b$.
Equalities and inequalities between random variables should be understood $d\mathbb{P}$-a.s., while for processes equalities and inequalities must be interpreted $d\mathbb{P}\times dt$-a.s.
We also define the space $$\mbox{BMO}(\mathbb{P}):=\left\{X\in L^2_{\mathcal{F}}(T;\R^n): \exists C>0 \mbox{ s.t. } \forall t\in[0,T] \mbox{ it holds that }  \mathbb{E}_{\mathbb{P}}\left[\int_t^T|X_s|^2ds\Bigg|\mathcal{F}_t\right]\leq C \right\}.$$
In the following definitions we assume $p\in[1,+\infty]$.
\begin{definition}[\cite{D02,BLR18,DS05,BLR21,D06}]
    $\rho_t:L^{p}(\mathcal{F}_T)\to L^{p}(\mathcal{F}_t)$ is a dynamic risk measure if it satisfies:
    \begin{itemize}
        \item Monotonicity: $X\geq Y \implies \rho_t(X)\geq\rho_t(Y)$, for any $X,Y\in L^{p}(\mathcal{F}_T),$ \mbox{$t\in[0,T].$}
    \end{itemize}
    A dynamic risk measure is monetary if cash-additivity holds: $\rho_t(X+m)=\rho_t(X)+m,$
    for any $X\in L^{p}(\mathcal{F}_T)$ and $m\in L^{p}(\mathcal{F}_t)$. \\
    A dynamic risk measure \mbox{$\rho_t: L^{p}_{+}(\mathcal{F}_T)\to L^{p}_{+}(\mathcal{F}_t)$} is a return risk measure if it is positively homogeneous:
    $$\rho_t(\alpha X)=\alpha\rho_t(X), \ \forall X\in L^{p}_{+}(\mathcal{F}_T),\ \forall \alpha\in L^{p}_{+}(\mathcal{F}_t).$$
\label{def:dmr}
\end{definition}
We recall that positive homogeneity can be defined analogously for a risk measure on the entire space $L^p(\mathcal{F}_T)$.
Sometimes, as discussed in \cite{BLR18}, return risk measures are normalized at $1$, meaning that $\rho_t(1) = 1$ for any $t \in [0, T]$.
However, it is important to note that this requirement is not always necessary in the context of this paper.
A static monetary or return risk measure can be defined analogously by taking $t=0$ in Definition~\ref{def:dmr}.

In the recent literature, the conditions of cash-additivity for monetary risk measures and positive homogeneity for return risk measures have been weakened to cash-subadditivity and star-shapedness, respectively (see \cite{ELKR09,LR22}).
By cash-subadditivity and star-shapedness we mean:
\begin{definition}
Let $\rho_t:L^{p}(\mathcal{F}_T)\to L^{p}(\mathcal{F}_t)$ be a dynamic risk measure.
Then $\rho_t$ is:
    \begin{itemize}
        \item Cash-subadditive: if $\rho_t(X+c)\leq\rho_t(X)+c$ for any $X\in L^p(\mathcal{F}_T),\ c\in L^{p}_{+}(\mathcal{F}_t),\ t\in[0,T]$.
    \item Star-shaped: if $\rho_t(\lambda X)\leq \lambda\rho_t(X)+(1-\lambda)\rho_t(0)$ for any $X\in L^{p}(\mathcal{F}_T),\ \lambda\in L^{p}(\mathcal{F}_t)$ with $0\leq \lambda\leq 1$.
\end{itemize}
\end{definition}
Although cash-additive and cash-subadditive static risk measures have been extensively studied, and their dual representations are well-known in the literature (e.g., \cite{CDFM15,JXZ20,HWWX22}), to the best of our knowledge, there are essentially no results on risk measures that satisfy only positive homogeneity or star-shapedness.
In \cite{CCMTW22}, representation results are provided for risk measures that are both cash-additive and star-shaped in the static case, and their proofs heavily rely on the axiom of cash-additivity.
In Sections~\ref{sec:star} and~\ref{sec:mr}, we will provide further insight into min-max representations for positively homogeneous and star-shaped risk measures, in both the static and dynamic cases.

In order to enhance our comprehension of the temporal structure of dynamic risk measures, it is frequently necessary to consider the following two properties:
\begin{itemize}
\item Time-consistency:
$\rho_s(\rho_t(X))=\rho_s(X), \ \ \forall X\in L^{p}(\mathcal{F}_T), \ 0\leq s\leq t\leq T.$
\item Regularity: $\mathbb{I}_A\rho_t(X)=\mathbb{I}_A\rho_t(\mathbb{I}_AX), \ \ \forall X\in L^{p}(\mathcal{F}_T), \ A\in\mathcal{F}_t, \ t\in [0,T],$
where $\mathbb{I}_{B}$ is the characteristic function of $B\subseteq\Omega$.
\end{itemize}
It is known that a dynamic cash-additive risk measure defined on bounded random variables is regular. As shown in Lemma \ref{REGCS}, regularity still holds for cash-subadditive risk measures.
We conclude this subsection by presenting a non-exhaustive list of well-known axioms for a dynamic risk measure $\rho_t: L^p(\mathcal{F}_T) \rightarrow L^p(\mathcal{F}_t)$:
\begin{itemize}
\item Normalization: for any $t\in[0,T]$, $\rho_t({0})=0.$
\item Positive-constancy: for any $t\in[0,T]$ and $c_t\in L^{p}_+(\mathcal{F}_t)$, $\rho_t(c_t)=c_t$.
\item Convexity:  $d\mathbb{P}\times dt$-a.s., $\rho_t(\lambda X+(1-\lambda)Y)\leq \lambda\rho_t(X)+(1-\lambda)\rho_t(Y)$ for any  $X,Y\in L^{p}(\mathcal{F}_T),$ and $\lambda\in L^{p}(\mathcal{F}_t),$
    with $\lambda\in L^p(\mathcal{F}_t)$ valued in $[0,1].$
    \item Sublinearity: $d\mathbb{P}\times dt$-a.s., $\rho_t$ satisfies positive homogeneity and subadditivity, i.e., \\
    \mbox{$\rho_t(X+Y)\leq \rho_t(X)+\rho_t(Y) \ \ \forall X,Y\in L^{p}(\mathcal{F}_T).$}
\end{itemize}
The same properties can be defined for a risk measure $\rho_t:L^p_+(\mathcal{F}_T)\to L^p_+(\mathcal{F}_t)$.
It is worth mentioning that for normalized risk measures once any two axioms among positive homogeneity, convexity, and sublinearity are satisfied, the remaining axiom is automatically implied (cf.\ \cite{FR02}).

For further details on the financial implications of the properties listed above, interested readers are referred to \cite{RG06, ELKR09, BEK05}.

\subsection{Dynamic risk measures and BSDEs}
In this subsection, we recall the relation between monetary risk measures and BSDEs, in particular, the so-called $g$-expectations introduced by \cite{P97}; see \cite{P05,BEK05,RG06,LS14}, among others.

Let $\left( W_{t}\right)_{t\in [0,T]}$ be a standard $n$-dimensional Brownian motion defined on the probability space $\left(\Omega,\mathcal{F},\mathbb{P}\right)$ and let the filtration $\left(\mathcal{F}_{t}\right)_{t\in [0,T]}$ be the augmented filtration associated to that generated by $\left(W_{t}\right) _{t\in [0,T]}$.

Consider a function
\begin{equation*}
\begin{array}{rl}
g:  \Omega \times \left[ 0,T\right] \times \mathbb{R}\times \mathbb{R}^{n}& \rightarrow \mathbb{R}\smallskip \\
 \left( \omega ,t,y,z\right) & \mapsto g\left( \omega ,t,y,z\right),
\end{array}
\end{equation*}
satisfying the following `$L^2$-standard assumptions':

\begin{itemize}
        \item $g$ is a $\mathcal{P}\times\mathcal{B}(\mathbb{R})\times\mathcal{B}(\mathbb{R}^n)$-measurable stochastic process with $\mathcal{P}$ the $\sigma$-algebra generated by predictable sets on $\Omega\times[0,T]$.
        \item There exists $k>0$ such that, $d\mathbb{P}\times dt$-a.s.,
        \begin{align*}
            |g(\omega,t,y_1,z_1)-g(\omega,t,y_2,z_2)|\leq k(|y_1-y_2|+|z_1-z_2|) \ \forall (y_1,z_1),(y_2,z_2)\in \mathbb{R}\times\mathbb{R}^n.
            \end{align*}
        \item $g(\cdot,0,0)\in L^2_{\mathcal{F}}(T;\mathbb{R})$.
\end{itemize}
Given a terminal condition $X \in L^{2}\left( \mathcal{F}_{T}\right)$, consider now the following BSDE:
\begin{equation}
\begin{cases}
-dY_{t}=g(t,Y_{t},Z_{t})dt-Z_{t}dW_{t},\text{ }\forall t\in \left[ 0,T\right];  \\
Y_{T}=X.
\end{cases}
\label{eq:bsde}
\end{equation}
We recall that for any $X\in L^{2}\left( \mathcal{F}_{T}\right)$, the BSDE \eqref{eq:bsde}
has a unique solution \mbox{$\left(Y,Z\right)\in L_{\mathcal{F}}^{2}\left( T;\mathbb{R}\right) \times L_{\mathcal{F}}^{2}\left( T;\mathbb{R}^{n}\right) $}
(see \cite{PP90}).
Under `$L^{\infty}$-standard assumptions', which additionally require $g(\cdot,0,0)\in\mbox{BMO}(\mathbb{P})$ and {$X\in L^{\infty}(\mathcal{F}_T)$}, the solution $(Y,Z)$ to the BSDE \eqref{eq:bsde} gains the further regularity $(Y,Z)\in L^{\infty}(\mathcal{F}_t)\times\mbox{BMO}(\mathbb{P})$ (see \cite{RZ23}).
A BSDE may be viewed as a representation of the dynamic programming principle in continuous-time; see \cite{PP90,EPQ97,LQ03,HJ07} and also the early \cite{B76}.
Furthermore, the $Y$-component of such a solution is called a conditional $g$-expectation (\cite{P97}).
Namely, given the solution $(Y_t^X, Z_t^X)$ associated to $X$, the $g$-expectation of $X$ is defined by
$\mathscr{E}_{g}\left[ X\right] \triangleq Y_{0}^{X}$,
while the conditional $g$-expectation of $X$ at time $t$ (for $t\in \left[ 0,T\right]$) is defined by
$\mathscr{E}_{g}\left[ \left. X\right| \mathcal{F}_{t}\right] \triangleq Y_{t}^{X}$.
When $g(t,y,z)=\mu|z|$, with $\mu>0$, $\mathscr{E}_g$ is usually denoted by $\mathscr{E}^{\mu}$.
See \cite{P97} for further details.

Our focus in the dynamic case will be on risk measures that have dynamics described in terms of BSDEs.
Therefore, let us review some properties of risk measures induced through BSDEs.
A dynamic risk measure, denoted by $\rho_t$, is defined through BSDEs as the first component of the solution to Equation \eqref{eq:bsde}.
It is important to note that the comparison theorem for BSDEs (cf.\ Theorem~2.2 in \cite{EPQ97}) guarantees that monotonicity of the first component of the solution to Equation \eqref{eq:bsde} with respect to the final condition $X$ is automatically satisfied when we assume $L^2$- or $L^{\infty}$-standard assumptions.
Consequently, solutions to BSDEs always serve as dynamic risk measures, in accordance with Definition~\ref{def:dmr}.

The following list provides some additional properties of risk measures defined through BSDEs; it makes explicit the relation between the driver $g$ and the corresponding $g$-expectation.
\begin{enumerate}
    \item Normalization: if $g(\cdot,0,0)\equiv 0$ then $\rho_t({0})=0$.
    \item Convexity: if $g$ is convex in $(y,z)$ then $\rho_t$ is convex.
    \item Positive homogeneity: if $g$ is positively homogeneous in $(y,z)$ then $\rho_t$ is positively homogeneous.
    \item Cash-additivity: if $g$ does not depend on $y$ then $\rho_t$ is cash-additive.
    \item Cash-subadditivity: if $g$ is non-increasing in $y$ then $\rho_t$ is cash-subadditive.
\end{enumerate}

Any risk measure induced via BSDEs is automatically time-consistent and regular (see \cite{P05}).
Nevertheless, dynamic risk measures not based on BSDEs might not satisfy time-consistency and/or regularity properties in general.
However, for monetary convex risk measures, time-consistency can be characterized in terms of dual representations and acceptance sets (see \cite{BN08,MGRO15}).
On the other hand, regularity is guaranteed for risk measures satisfying cash-subadditivity, as shown in Lemma \ref{REGCS}.

It is important to note that if the driver function satisfies the stronger property $g(\cdot,\cdot,0)\equiv 0$, then the converse implications of items 2-4 are also true. Specifically, under the assumption $g(\cdot,\cdot,0)\equiv 0$, convexity, positive homogeneity, and cash-additivity of $\rho_t$ imply that the driver function $g$ is convex, positively homogeneous, and independent of $y$, respectively.
These results have been established in \cite{J08}.

\subsection{Dual representation of convex risk measures}
In this subsection, we recall some relevant duality results that will be used later.
A cash-additive and convex risk measure, generated using BSDEs with a driver $g$ that satisfies $L^2$-standard assumptions and is independent of $y$, can be expressed in the following manner (see \cite{BEK05}):
    \begin{equation}
\rho_t(X)=\essmax_{\mu_t\in\mathcal{A}}\mathbb{E}_{\mathbb{Q}^{\mu}_t}\left[X-\int_t^TG(s,\mu_s)ds\Big|\mathcal{F}_t\right], \ \ \forall X\in L^{2}(\mathcal{F}_T),\ t\in[0,T],
\label{eq:dualSA}
    \end{equation}
    where $G:\R^n\to\R\cup\{+\infty\}$ defined by
    $$G(t,\mu):=\sup_{z\in\mathbb{R}^n}\left\{-\langle \mu,z\rangle-g(t,z)\right\}, \ \ \mu\in\R^n,$$ is the Fenchel transformation of the driver $g$ and $$\mathcal{A}:=\left\{(\mu_t)_{t\in[0,T]} \mbox{ predictable process}: G(t,\mu_t)<+\infty, 0\leq ||\mu_t||\leq k \right\},$$ with $k>0$ the Lipschitz constant of $g$.
    The measure $\mathbb{Q}^{\mu}_t$ is given by $$\mathbb{E}\left[\frac{d\mathbb{Q}^{\mu}_t}{d\mathbb{P}}\Big|\mathcal{F}_t\right]=\mathfrak{E}(\mu\cdot W)_t:=\exp\left(-\frac{1}{2}\int_0^t|\mu_s|^2ds-\int_0^t\mu_sdW_s\right)\in L^2(\mathcal{F}_t).$$
    Now we proceed to the dual representation of a cash-subadditive risk measure induced via BSDEs with a driver $g$ satisfying $L^2$-standard assumptions and that is non-increasing w.r.t.~$y$ (see \cite{ELKR09}):
    \begin{equation}
\rho_t(X)=\essmax_{(\beta_t,\mu_t)\in\mathcal{A}'}\mathbb{E}_{\mathbb{Q}^{\mu}_t}\left[D^{\beta}_{t}X-\int_t^TD_{t,s}^{\beta}G(s,\beta_s,\mu_s)ds\Big|\mathcal{F}_t\right],
\label{eq:dualCSA}
    \end{equation}
for any $X\in L^{2}(\mathcal{F}_T)$ and $ t\in[0,T]$, where $G:\mathbb{R}\times\R^n\to\R\cup\{+\infty\}$ defined by
$$G(t,\beta,\mu):=\sup_{(y,z)\in\mathbb{R}\times\mathbb{R}^n}\left\{-\beta y-\langle \mu,z\rangle-g(t,y,z)\right\}, \ \  (\beta,\mu)\in\mathbb{R}\times\mathbb{R}^n,$$ is the Fenchel transformation of the driver $g$ and  $$\mathcal{A}':=\left\{(\beta_t,\mu_t)_{t\in[0,T]} \ \mbox{predictable processes}: G(t,\beta_t,\mu_t)<+\infty, 0\leq \beta_t+|\mu_t|\leq k\right\},$$ with $k$ the Lipschitz constant of $g$ and $D^{\beta}_{t,s}:=e^{-\int_t^s\beta_udu}$ with the convention $D^{\beta}_{t,T}:=D^{\beta}_t$.
If we assume the stronger $L^{\infty}$-standard assumptions, the previous representations hold with $\rho_t(X)\in L^{\infty}(\mathcal{F}_t)$ for any $t\in[0,T]$.
In addition, in this case we also have that the density process is bounded: $\mathbb{Q}^{\mu}_t\in L^{\infty}(\mathcal{F}_t)$.
Moreover, when we deal with convex drivers satisfying $L^2$-standard assumptions (no cash-additivity or cash-subadditivity), we have the following representation valid for any $t\in[0,T]$ and $X\in L^2(\mathcal{F}_T)$, as shown in Proposition~3.1 of \cite{DKRT14}:\footnote{In our current setting, we have slightly different assumptions compared to those in Proposition~3.1 of \cite{DKRT14}.
However, it is worth noting that sublinearity is immediately verified under standard assumptions on the driver $g$, in particular it is implied by Lipschitz continuity.
Furthermore, it is sufficient for $g(\cdot,0,0)\in\mbox{BMO}(\mathbb{P})$ to obtain the same result as Proposition~3.1 of \cite{DKRT14}.
We also emphasize that the positivity of $g$ is not a necessary assumption, as it is not involved in the proof of Proposition~3.1 in \cite{DKRT14}.}
\begin{equation}
\rho_t(X)=\esssup_{(\beta,q)\in\mathcal{G}\times\mathcal{Q}^2}\mathbb{E}_{\mathbb{Q}^{q}_t}\left[D^{\beta}_{t}X-\int_t^TD^{\beta}_{t,s}G(s,\beta_s,q_s)ds\Big|\mathcal{F}_t\right],
\label{eq:conrep}
\end{equation}
where
\begin{align*}
    &\mathcal{Q}^p:=\left\{q\in L^{2}_{\mathcal{F}}(T;\R^n):  \exp\left(\int_0^Tq_sdW_s-\frac{1}{2}\int_0^T|q_s|^2ds\right)\in L^{p}(\mathcal{F}_T) \right\}, \\
    &\mathcal{G}:= \Bigg\{(\beta_t)_{t\in[0,T]}\mbox{ predictable process}:  \int_0^T\beta_s^-ds\in L^{\infty}(\mathcal{F}_T),\int_0^T\beta^+_sds<+\infty\Bigg\},
\end{align*}
with $p\in[1,+\infty]$.
Furthermore, Lipschitzianity of $g$ ensures that $\beta$ and $q$ are bounded processes in the representation \eqref{eq:conrep}.
If the driver $g$ satisfies $L^{\infty}$-standard assumptions and $X\in L^{\infty}(\mathcal{F}_T)$, then Equation \eqref{eq:conrep} holds with $\esssup$ taken over $(\beta,q)\in \mathcal{G}\times\mathcal{Q}^{\infty}$. Furthermore, the $\esssup$ is always attained for some $(\bar{\beta},\bar{q})\in \mathcal{G}\times\mathcal{Q}^{\infty}$.

\setcounter{equation}{0}

\section{Dynamic Convex Return Risk Measures}\label{sec:con}

This section presents results on dynamic convex return risk measures.
We analyze in particular whether there exist non-trivial dynamic return risk measures induced by BSDEs that are convex but not cash-additive.
We answer this question positively.
We also demonstrate that convexity is both a powerful and a restrictive axiom in the context of return risk measures:
every convex return risk measure is shown to be both cash-subadditive and sublinear,
somewhat similar to the well-known convexity implications for normalized monetary risk measures.
This will later serve as a motivation to relax the convexity assumption and develop a more general and flexible theory.

We first present some results that apply to general dynamic convex return risk measures.
Subsequently, we narrow down our focus to the specific case of return risk measures induced via BSDEs.

    \begin{proposition}
            Let $\rho_t:L^{\infty}(\mathcal{F}_T)\to L^{\infty}(\mathcal{F}_t)$ be a monotone, convex, continuous from below and positively homogeneous functional such that $\rho_t(1)=1 \ \forall t\in[0,T]$.
            Then
            \begin{enumerate}
            \item $\rho_t(0)=0 \ \ \forall t\in[0,T]$ a.s.,
            \item $\rho_t$ admits the dual representation:
        \begin{equation*}
            \rho_t(X)=\esssup_{(D,\mathbb{Q})\in \mathcal{D}\times\mathcal{Q}}\left\{D_t\mathbb{E}_{\mathbb{Q}_t}[X|\mathcal{F}_t]\right\}, \ \forall X\in L^{\infty}(\mathcal{F}_T), \ t\in[0,T],
        \end{equation*}
         where
         \begin{align*}
         &\mathcal{D}:=\{D:[0,T]\times\Omega\to[0,1], \ \ \mathcal{F}_t\mbox{-adapted stochastic processes}\}, \\
         &\mathcal{Q}:=\left\{\mathbb{Q}_t \mbox{ probability measures on } (\Omega,\mathcal{F}_T): \mathbb{Q}_t\ll\mathbb{P} \mbox{ and } \mathbb{Q}_t|_{\mathcal{F}_t}\equiv\mathbb{P}\right\},
         \end{align*}
        \item $\rho_t$ is cash-subadditive.
         \end{enumerate}
         \label{DynamicReturn}
        \end{proposition}

        \begin{corollary}
            If $\rho_t:L^{\infty}_{+}(\mathcal{F}_T)\to L^{\infty}_{+}(\mathcal{F}_t)$ is a monotone, convex, continuous from below and positively homogeneous functional such that $\rho({1})=1$, then it is a sublinear and cash-subadditive return risk measure normalized at $1$.
            \label{DynamicReturn1}
        \end{corollary}

Now let us shift our focus to the setting of dynamic convex risk measures induced through BSDEs.
We assume $L^{\infty}$-standard assumptions for the driver $g$, however it is important to note that all the theses of the next proposition remain valid in the case of $L^2$-standard assumptions.
If we want the return risk measure $\rho_t$ to have a convexity property, we can choose a driver $g$ that is convex in $(y,z)$.
    It is well-known (cf.\ \cite{FR02,LS13}) that convexity, continuity from below, which is satisfied by any convex solution of a BSDE, and constancy of $\rho_t$ imply the cash-additivity property.
    However, in the context of return risk measures, as shown in Remark~\ref{Rconstancy}, we only have {positive}-constancy and we wonder whether such property still implies a corresponding `positive cash-additive' property.
    The following proposition answers this question negatively.
       We can represent convex return risk measures using a BSDE with a positively homogeneous and convex driver $g$, which satisfies $g(\cdot,1,0)\equiv0$.
       Unlike in the case of cash-additivity, the driver $g$ may also depend on $y$.
       Thus, the next proposition shows that we are able to build dynamic convex return risk measures induced via BSDEs.

            \begin{proposition}
            Suppose $g:\Omega\times [0,T]\times\mathbb{R}\times\mathbb{R}^n\to\mathbb{R}_+$ satisfies the $L^{\infty}$-standard assumptions, as well as positive homogeneity, convexity in $(y,z)$ and $g(\cdot,1,0)\equiv0$.
            Then the first component of the solution to Equation~\eqref{eq:bsde}, $\rho_t:L^{\infty}_+(\mathcal{F}_T)\to L^{\infty}_+(\mathcal{F}_t)$,  defines a sublinear, positive-constant, and cash-subadditive return risk measure for any \mbox{$X\in L^{\infty}_+(\mathcal{F}_T)$.}
            \label{PRP}
            \end{proposition}

It is important to emphasize that, throughout this section, we have not imposed the stronger condition $g(t, y, 0) \equiv 0$ for all $y \in \mathbb{R}$.
As a result, convexity (sublinearity) can co-exist with a driver that depends on $y$.
Hence, cash-additivity is not valid in general, and instead, it is replaced by cash-subadditivity.

Furthermore, it is worth noting that in the dynamic setting, the intersection between convex risk measures (restricted to $L^{\infty}_{+}(\mathcal{F}_T)$) and convex return risk measures normalized at 1 occurs when considering sublinear and cash-subadditive risk measures.
This is in line with the fact that convexity cannot be assumed independently of sublinearity in the context of return risk measures, due to the positive homogeneity property.
This extends the findings from the static case (see Figure~1 in \cite{LR22}), where the intersection between convex, normalized, cash-additive risk measures and return risk measures normalized at 1 is given by coherent risk measures (i.e., convex, positively homogeneous and cash-additive risk measures).

\setcounter{equation}{0}

\section{Star-Shaped Risk Measures}\label{sec:star}

As we have demonstrated in Section~\ref{sec:con}, convex return risk measures are always cash-subadditive and sublinear.
Although convexity and sublinearity are interesting properties for return risk measures, it is important to note that the only assumptions in their definition are monotonicity and positive homogeneity.
In some cases, such as when the market includes illiquid assets, only monotonicity and star-shapedness may be relevant.
More information on this can be found in \cite{LR22}.
Due to these reasons and motivated by the findings in Section~\ref{sec:con}, we now develop a general theory that includes convex return risk measures only as a particular case.

In light of these considerations, we aim to provide a characterization of risk measures that are star-shaped (or positively homogeneous).
Specifically, we show that every functional satisfying star-shapedness can be represented as the minimum of convex functionals, and we also provide an equivalent min-max representation.
For the sake of generality, our framework includes both positively and negatively valued random variables.
When working with return risk measures defined only on positive random variables, one can simply restrict the domain without altering the general structure.
Additionally, our results hold for spaces more general than $L^{\infty}$, which is the space usually employed in the related literature.

To achieve our goal, we introduce some definitions that will be used in the following.
Henceforth, let $\mathcal{X}\subseteq \bar{L}^0(\Omega,\mathcal{F},\mathbb{P})$ where $\bar{L}^0(\Omega,\mathcal{F},\mathbb{P})$ is the space of extended real-valued random variables.
We always assume that $\mathcal{X}$ includes the constants.
We need some topological assumptions: $\mathcal{X}$ will be a locally convex\footnote{By a locally convex space we mean a separated topological vector space that admits a neighborhood basis consisting of convex sets.
In the following, we refer to this topology by $\tau$.} Fr\'echet lattice w.r.t.~the usual order relation between random variables, which is also order complete and order separable (see \cite{BF10} and references therein for a thorough discussion about these properties).
We say that $X_n\xrightarrow{o}X$ if $|X_n-X|\leq Y_n$ with $Y_n\leq Y_{n+1}$ for any $n\in\mathbb{N}$ and s.t.\ the infimum of $Y_n$ is $0$, where $|X|=X\vee -X$ ($\vee$ is the usual lattice operation, i.e., the maximum between random variables).
We denote with $\mathcal{X}^*$ the order continuous dual of $\mathcal{X}$ and with $\sigma(\mathcal{X},\mathcal{X}^*)$ the corresponding weak-topology.
We indicate the duality between $\mathcal{X}$ and $\mathcal{X}^*$ with $\langle \cdot , \cdot \rangle$.
We say that $f:\mathcal{X}\to\mathbb{R}\cup\{+\infty\}$ is order lower semicontinuous if for any $X_n\xrightarrow{o}X$ it holds that $\liminf_n f(X_n)\geq f(X)$.
We always suppose that $\sigma(\mathcal{X},\mathcal{X}^*)$ has the C-property (see Definition~3 in \cite{BF10}).
We recall that $L^p$ spaces with $p\in[1,+\infty]$, Orlicz spaces and the Morse space (i.e., Orlicz heart) have all these properties.
We define the proper domain of a functional $f:\mathcal{X}\to\mathbb{R}\cup\{+\infty\}$ as the set $dom(f):=\{X\in\mathcal{X} \mid f(X)<+\infty\}$.
We present the following definition of star-shapedness, which simplifies to the concept of star-shapedness introduced in Section~\ref{sec:prel} when the functional is normalized:

\begin{definition}
    A functional $f:\mathcal{X}\to\mathbb{R}\cup\{+\infty\}$ is said to be star-shaped with respect to its value at $0$ (star-shaped for short in the following) if, for any $\lambda\in[0,1]$, the inequality
    $$f(\lambda X)\leq \lambda f(X)+(1-\lambda)f(0) \ \forall X\in\mathcal{X}$$
    holds and $f(0)<+\infty$.
\end{definition}
We begin with a simple lemma that demonstrates the preservation of star-shapedness under the operation of the infimum.
This implies, in particular, that the infimum of convex functionals is star-shaped.

\begin{lemma}
Let $\Gamma$ be a set of indexes and $(f_{\gamma})_{\gamma\in\Gamma}$ a family of star-shaped functionals $f_{\gamma}:\mathcal{X}\to\mathbb{R}\cup\{+\infty\}$ such that there exists $c\in\R$ verifying $f_{\gamma}(0)=c$ for any $\gamma\in\Gamma$. We define the functional
    $$f(X):=\inf_{\gamma\in\Gamma}f_{\gamma}(X) \ \forall X\in\mathcal{X}.$$
    Then, $f$ is star-shaped and $f(0)=c$.
    \label{Lemma:infSS}
\end{lemma}
In the following, we prove that the converse is also true: any star-shaped functional can be expressed as the minimum of convex functionals.

    \begin{proposition}[Dual representation of star-shaped functions]
       Let $f:\mathcal{X}\to\mathbb{R}\cup\{+\infty\}$ with $f(0)<+\infty$ be a star-shaped functional.
       Then $f$ can be represented as a pointwise minimum of convex and (order) lower semicontinuous functionals:
        \begin{equation}
            f(X)=\min_{Z\in dom(f)}f_Z(X) \ \ \ \forall X\in\mathcal{X},
            \label{minstarshaped}
        \end{equation}
        where $f_Z$ is defined as in Equation \eqref{convexf}.
        Moreover, it admits the min-max representation:
        \begin{equation}
            \label{minmaxstatic}
            f(X)=\min_{Z\in dom(f)}\sup_{q\in \mathcal{X}^*}\{\langle q, X \rangle-f_Z^{*}(q)\},
        \end{equation}
        where $f_Z^{*}$ is the convex conjugate of $f_Z$.

        Furthermore, if $f$ is monotone (thus it is a static star-shaped risk measure\footnote{In this setting, a functional $\rho:\mathcal{X}\to\R$ is said to be a static risk measure if it satisfies monotonicity, i.e., for any $X,Y\in\mathcal{X}$ such that $X\geq Y$ it holds that $\rho(X)\geq\rho(Y)$.}) we can find a family $(\tilde{f}_Z)_{Z\in dom(f)}$ of proper, convex and order lower semicontinuous risk measures such that the representation \eqref{minstarshaped} holds.
        Moreover, we obtain the min-max representation:
         \begin{equation}
            \label{minmaxstaticmon}
f(X)=\min_{Z\in dom(f)}\sup_{q\in\mathcal{X}^*_+}\{\langle q, X \rangle -\tilde{f}_Z^{*}(q)\},
        \end{equation}
        where $\mathcal{X}^*_+$ is the convex cone in the order continuous dual $\mathcal{X}^*$ generated by positive order continuous functionals (i.e., $\langle q,X\rangle \geq 0 \ \forall X\in\mathcal{X}$).
        \label{staticSS}
    \end{proposition}
While the formal proofs of these statements are provided in Appendix~\ref{sec:appendix}, we would like to explicate the intuitive strategies adopted.
The approach involves constructing a family of functionals, denoted as $f_Z$, where $Z\in dom(f)$, such that for any $X\in\mathcal{X}$ lying on the line segment connecting $0$ and $Z$, we have $f_Z(X) = \alpha f(Z)+(1-\alpha)f(0)$, with $\alpha\in[0,1]$ satisfying $X=\alpha Z$.
Additionally, $f_Z(X)=+\infty$ when $X$ does not lie on the segment joining the origin and $Z$.
By design, $f_Z$ is convex, and it is evident that $f_X(X)=f(X)$ for any $X$ in the proper domain of $f$.
This methodology captures the notion that a function $f$ is star-shaped with respect to $f(0)$ if and only if the line segment joining the points $(0,f(0))$ and $(Z,f(Z))$ on the graph of $f$ lies entirely above the graph for any $Z\in dom(f)$.
This observation intuitively justifies the inequality $f_Z \geq f$, for any element $f_Z$ of the family. Furthermore, due to the simple form of the functionals $f_Z$ we construct, we are able to establish lower semicontinuity and lower order semicontinuity through direct inspection.
        \begin{remark}
            It is not difficult to prove that similar results hold when we restrict our attention to positively homogeneous functionals (which are clearly star-shaped) defined on $\mathcal{X}$.
            In this case, we can show that their representation involves only sublinear functionals.
            Indeed, we already know from Remark~1 in \cite{C00} that $$f(X)=\min_{Z\in dom(f)}f_Z(X) \ \ \forall X\in\mathcal{X},$$ where $\mathcal{X}$ is a locally convex space and $f_Z:\mathcal{X}\to\mathbb{R}\cup\{+\infty\}$ is defined by:
            \begin{equation*}
                f_Z(X)=\begin{cases}
                \alpha f(Z) \ &\mbox{ if there exists } \alpha\in[0,+\infty) \mbox{ s.t. } X=\alpha Z, \\
                +\infty &\mbox{ otherwise.}
                 \end{cases}
            \end{equation*}
            In detail, for any $Z\in dom(f)$, $f_Z$ is a proper, lower semicontinuous and sublinear functional.
            As in the proof of Proposition~\ref{staticSS}, we can prove that $f_Z$ is also order lower semicontinuous, assuming that $\mathcal{X}$ is a locally convex order complete and order separable Fr\'echet lattice with the C-property.
            Hence, we obtain the min-max representation:
            $$f(X)=\min_{Z\in dom(f)}\sup_{q\in\mathcal{C}_Z}\left\{\langle q,X\rangle\right\},$$
            where $\mathcal{C}_Z$ is a convex subset of $\mathcal{X}^*$ depending on $Z\in dom(f)$.
            In addition, if $f$ is also monotone we can consider only dual elements $q\in\mathcal{C}'_Z$, where $\mathcal{C}'_Z$ is a convex subset of $\mathcal{X}^*_+$. In other words, we are able to represent every positively homogeneous risk measure in a (very) general space as a min-max of linear and positive functionals.
            \end{remark}
               It is worth noting that the representation of a star-shaped functional as the minimum of convex functionals is not necessarily unique.
                We can obtain uniqueness by choosing a suitable relaxation, using arguments analogous to those in Proposition~4 of \cite{CCMTW22}.

           The general representation results established in Proposition~\ref{staticSS} encompass various existing results as special cases.
           For example, by imposing the additional constraint of cash-additivity, we recover the conclusions of Theorem 2 and Proposition 5 in \cite{CCMTW22}.
           We emphasize, however, that all our representations hold
           in (very) general spaces, which include $L^{p}$ for any $p\in[0,+\infty]$, Orlicz spaces generated by a proper Young function and Morse spaces, and we allow for non-monotonicity.
           In the following corollary, we formally report these results; the proof is a consequence of the proofs of Lemma~\ref{Lconvexf} and Proposition~\ref{staticSS}.
           It is noteworthy that the following result still holds in the case of non-normalized star-shaped functionals.
           In this setting, the family of convex functionals that represents the star-shaped functional will consist of non-normalized functionals, wherein the value at $0$ of each functional will be equal to the value at $0$ of the star-shaped functional.
            \begin{corollary}
                Let $\rho:\mathcal{X}\to\mathbb{R}\cup\{+\infty\}$ be a normalized, star-shaped and cash-additive functional.
                Then there exist a set of indexes $\Gamma$ and a family $(\rho_{\gamma})_{\gamma\in\Gamma}$ of normalized, convex and cash-additive functionals $\rho_{\gamma}:\mathcal{X}\to\mathbb{R}\cup\{+\infty\}$ such that:
                \begin{equation*}
                    \rho(X)=\min_{\gamma\in\Gamma}\rho_{\gamma}(X) \ \ \forall X\in\mathcal{X},
                \end{equation*}
                or, equivalently, in terms of a min-max representation:
                \begin{equation*}
                    \rho(X)=\min_{\gamma\in\Gamma}\sup_{q\in\mathcal{C}}\left\{\langle q,X\rangle-\rho^*_{\gamma}(q)\right\} \ \ \forall X\in\mathcal{X},
                \end{equation*}
                where $\mathcal{C}:=\left\{q\in\mathcal{X}^*: \langle q, 1 \rangle =1\right\}$.
               Finally, if $\rho$ is also monotone, then every element of the family $(\rho_{\gamma})_{\gamma\in\Gamma}$ can be chosen to be monotone and the min-max representation holds if we restrict the admissible scenarios to $\mathcal{C}\cap\mathcal{X}^*_+$.
                \label{corCA}
                \end{corollary}

\subsection{Star-shaped risk measures and acceptance sets}

In this subsection, we investigate the properties of the acceptance sets of a star-shaped risk measure.
We define the concept of families of star-shaped acceptance sets and we prove the one-to-one correspondence with star-shaped risk measures.
Our definition of acceptance sets differs from the usual definition given, for instance, in \cite{DK12}.
Indeed, we do not necessitate the convexity of acceptance sets. Instead, we only require the weaker property of \textit{star-shapedness}, drawing an analogy with the star-shapedness of risk measures.
For brevity, we restrict our attention to normalized risk measures.
\begin{definition}
$\mathcal{A}:=(A^m)_{m\in\R}\subseteq \mathcal{X}$ is called a star-shaped family of acceptance sets if
\begin{itemize}
\item $\mathcal{A}$ is increasing in $m$, i.e., $A^{m_1}\subseteq A^{m_2}$ for any $m_1\leq m_2$;
\item $\mathcal{A}$ is monotone, i.e., if $Y\geq X$ and $X\in A^{m}$ for some $m\in\R$, then $Y\in A^m$;
\item $\mathcal{A}$ is star-shaped, i.e., $\lambda A^{m}\subseteq A^{\lambda m}$ for any $\lambda\in[0,1]$ and $m\in\R$;
\item $\mathcal{A}$ is right-continuous, i.e., $A^{m}=\bigcap_{\bar{m}>m}A^{\bar{m}}$ for any $m\in\R$.
\end{itemize}
\end{definition}

We recall that a family of acceptance sets is said to be a \textit{convex} family of acceptance sets if it is monotone, increasing, right-continuous and  $\lambda A^m+(1-\lambda)A^{m'}\subseteq A^{\lambda m+(1-\lambda)m'}$ for all $\lambda\in[0,1]$ and $m,m'\in\R$.
We stress that if $0\in A^0$, then convexity of $(A^m)_{m\in\mathbb{N}}$ implies star-shapedness of $(A^{m})_{m\in\mathbb{N}}$.
We have the following characterization of star-shaped risk measures:
\begin{lemma}
               If $\rho:\mathcal{X}\to\R\cup\{+\infty\}$ is a normalized star-shaped risk measure, then the family $\mathcal{A}_{\rho}$ \mbox{defined by}
                \begin{equation*}
                A^m_{\rho}:=\left\{X\in\mathcal{X}:\rho(X)\leq m\right\}\ \mbox{ for any } m\in\R,
                \end{equation*}
                is a star-shaped family of acceptance sets such that $\inf\{m\in\R: 0\in A^{m}_{\rho}\}=0$.
                Conversely, if $\mathcal{A}=(A^m)_{m\in\R}$ is a star-shaped family of acceptance sets such that $\inf\{m\in\R: 0\in A^{m}\}=0$, then \mbox{$\rho_{\mathcal{A}}:\mathcal{X}\to\R\cup\{+\infty\}$} defined by
                \begin{equation*}
                \rho_{\mathcal{A}}(X)=\inf\left\{m\in\R:X\in A^m\right\},
                \end{equation*}
                is a normalized star-shaped risk measure.
                In addition, $\rho=\rho_{\mathcal{A}_{\rho}}$ and $\mathcal{A}=\mathcal{A}_{\rho_{\mathcal{A}}}$.
                \label{asSS}
\end{lemma}

We now state the main result of this subsection.
\begin{theorem}
Let $\rho:\mathcal{X}\to\R\cup\{+\infty\}$.
Then the following statements are equivalent:
\begin{enumerate}[i)]
\item $\rho$ is a normalized star-shaped risk measure;
\item There exists a family of normalized and convex risk measures $(\rho_{\gamma})_{\gamma\in\Gamma}$ parameterized by $\gamma\in\Gamma$ such that
$$\rho(X)=\min_{\gamma\in\Gamma}\rho_{\gamma}(X) \ \ \forall X\in\mathcal{X};$$
\item There exists a convex family of acceptance sets $(\mathcal{A}_{\gamma})_{\gamma\in\Gamma}$ parameterized by $\gamma\in\Gamma$ verifying \begin{equation}
\inf\{m\in\R:0\in A^m_{\gamma}\}=0 \ \mbox{ for any } \gamma\in\Gamma,
\label{eq: normas}
\end{equation}
such that
$$\rho(X)=\inf\left\{m\in\R:X\in A^m_{\gamma} \ \mbox{ for some } \gamma\in\Gamma\right\}.$$
\end{enumerate}
\label{th:asSS}
\end{theorem}
The last theorem can also be seen as a generalization of Theorem~2 in \cite{CCMTW22}, where the authors request the further property of cash-additivity on $\rho$, which can be understood at the level of acceptance sets as $A^{m}=A^{0}+m$ for any $m\in\R$ (cf.~\cite{DK12} for a thorough discussion on this topic).

\setcounter{equation}{0}

\section{Main Results}\label{sec:mr}

In this section, we present the main results of our paper.
First, we prove that any dynamic risk measure that is star-shaped and induced via BSDEs can be represented as the pointwise minimum of a family of dynamic convex risk measures, whose dynamics are governed by BSDEs.
Next, we investigate the circumstances under which a dynamic star-shaped risk measure can be created from a set of convex risk measures.
When the convex risk measures are also governed by BSDEs, we identify the conditions under which they can be used to form a dynamic star-shaped risk measure that can be represented via a BSDE.

\subsection{Min-max representation of star-shaped risk measures induced via BSDEs}

\begin{proposition}
Let $g:\Omega\times [0,T]\times\mathbb{R}\times\mathbb{R}^n\to\mathbb{R}$ satisfy $L^2$-standard assumptions.
Then for any $X\in L^2(\mathcal{F}_T)$ the unique first component of the solution, $(\rho_t(X))_{t\in[0,T]}$, to the equation:
\begin{equation}
    \rho_t(X)=X+\int_t^Tg(s,\rho_s,Z_s)ds-\int_t^TZ_sdW_s
    \label{eq: star-shapedBSDE}
\end{equation}
is star-shaped (resp.\ positively homogeneous) if and only if the driver $g$ is star-shaped (resp.\ positively homogeneous) w.r.t.\ $(y,z)$.
    \label{POSS}
\end{proposition}

We are ready to state the main results of this subsection.
          \begin{theorem}
              Let $g:\Omega\times [0,T]\times\mathbb{R}\times\mathbb{R}^n\to\mathbb{R}$ satisfy $L^2$-standard assumptions and star-shapedness w.r.t.\ $(y,z)$ (resp.\ positive homogeneity w.r.t.\ $(y,z)$).
              Then $\rho_t$ induced by the BSDE \eqref{eq: star-shapedBSDE} is a star-shaped (resp.\ positively homogeneous and normalized) risk measure that can be represented as
                  \begin{equation}
                  \rho_t(X)=\essmin_{\gamma\in\Gamma}\rho^{\gamma}_t(X) \ \ \forall t\in[0,T], \forall X\in L^2(\mathcal{F}_T),                  \label{minsub}
                  \end{equation} where $\Gamma$ is a set of indices and $\rho^{\gamma}_t$ are convex (resp.\ sublinear and normalized) risk measures induced via BSDEs that dominate $\rho_t$.
                  For each $\gamma\in\Gamma$, the driver $g^{\gamma}$ is convex (resp.\ sublinear) and the family $(g^{\gamma})_{\gamma\in\Gamma}$ is $k$-equi-Lipschitz, where $k$ is the Lipschitz constant of $g$.\footnote{A family of functions $(g^{\gamma})_{\gamma\in\Gamma}$ is said to be $k$-equi-Lipschitz if there exists a constant $k>0$ such that, $d\mathbb{P}\times dt$-a.s., for all $\gamma\in\Gamma$ and $(y_1,z_1),(y_2,z_2)\in\mathbb{R}\times\mathbb{R}^n$, the following inequality holds:
        $$|g^{\gamma}(t,y_1,z_1)-g^{\gamma}(t,y_2,z_2)|\leq k(|y_1-y_2|+|z_1-z_2|).$$}
                  We also have the min-max representation:
                  \begin{equation}
                  \rho_t(X)=\essmin_{\gamma\in\Gamma}\esssup_{(\beta,q)\in\mathcal{G}\times\mathcal{Q}^2}\mathbb{E}_{\mathbb{Q}^{q}_t}\left[D^{\beta}_{t}X-\int_t^TD^{\beta}_{t,s}G^{\gamma}(s,\beta_s,q_s)ds\Big|\mathcal{F}_t\right],
                  \label{minmax}
                  \end{equation}
                  for any $ t\in[0,T]$ and $ X\in L^2(\mathcal{F}_T),$ with $\beta$ and $q$ bounded processes.
                  In particular, $G^{\gamma}$ vanishes on its proper domain when the driver $g$ is positively homogeneous.
                  \label{SuCO}
          \end{theorem}

The proof of the previous results (see Appendix~\ref{sec:appendix}) relies on a regularization technique. Intuitively, we achieve this by performing an infimal convolution between the elements of a family $g_{\bar{y},\bar{z}}$ with $(\bar{y},\bar{z})\in\R\times\R^n$, which serve as analogs of $(f_Z)_{Z\in dom(f)}$ for $f$ in Proposition~\ref{staticSS}, and a regularizing Lipschitz function.
For any fixed $(t, \omega) \in [0, T] \times \Omega$ and $\gamma=(\alpha,\delta)$, where $\alpha$ and $\delta$ are square integrable and predictable processes valued in $\R\times\R^n$, the function $(y, z) \mapsto g^\gamma(\omega, t, y, z)$ represents the infimal convolution of the line segment connecting $g(\omega, t, 0, 0)$ to $g(\omega, t, \alpha_t(\omega), \delta_t(\omega))$, where $(y, z)$ lies on the segment joining $(0, 0)$ to $(\alpha_t(\omega), \delta_t(\omega))$, with the regularization function $h(y, z) = k(|y| + |z|)$.
This approach allows us to construct a family of drivers that satisfy $L^\infty$- or $L^2$-standard assumptions due to the regularization through infimal convolution.
Importantly, this construction preserves the properties of $g_{\bar{y},\bar{z}}$, meaning that $g^\gamma \geq g$ holds for any $\gamma \in \Gamma$, and equality is achieved for $\gamma=(\rho,Z)$ where $(\rho,Z)$ is the solution to Equation~\eqref{eq: star-shapedBSDE}.
Finally, by utilizing a comparison theorem for BSDEs, we obtain the desired result.

Different from \cite{TW23} in which the driver is assumed to be independent of $y$ (and to satisfy a quadratic growth condition), we consider here the non-cash-additive case in which the driver is allowed to depend on $y$.
Compared to Corollary~4.1 of \cite{TW23}, our Theorem~\ref{SuCO} establishes representation results for a wider family of risk measures that are not in general monetary/cash-additive.
As a corollary, our findings not only recover the results of \cite{TW23} (in the setting of Lipschitz conditions) when the driver is independent of $y$, but they also establish a characterization of $\rho_{t}^{\gamma}$ as the solution to a BSDE with a specific driver, based on the Pasch-Hausdorff envelope of $g$.
Additionally, unlike Corollary~4.1 in \cite{TW23}, by the flow property of BSDEs (see e.g., \cite{EPQ97}), we can also ensure the time-consistency of $\rho^{\gamma}_t$ for each $\gamma\in\Gamma$.
Indeed, further insight can be gained into the family of drivers $(g^{\gamma})_{\gamma\in\Gamma}$ by imposing additional conditions on the dynamics of $\rho_t$.
Specifically, assuming that $\rho_t$ is cash-additive (resp.\ cash-subadditive) under certain conditions on $g$, we can express $\rho_t$ in terms of a family of cash-additive (resp.\ cash-subadditive) and convex (resp.\ sublinear) risk measures $(\rho_t^{\gamma})_{\gamma\in\Gamma}$.
Analogous to the representations provided in Equations~\eqref{minsub}--\eqref{minmax}, the resulting expressions offer valuable insights into the behavior of $\rho_t$ under these additional constraints.

The following corollary provides the precise statements of these results.
Notably, this corollary extends the findings in Corollary~\ref{corCA} to the dynamic setting, where $\mathcal{X}=L^2(\mathcal{F}_T)$ or $\mathcal{X}=L^{\infty}(\mathcal{F}_T)$.

\begin{corollary}
    Let us assume the same hypotheses as in Theorem~\ref{SuCO}.
    If the driver $g$ does not depend on $y$ and $g(t,\cdot,\cdot)\equiv0$, then there exists a family of drivers $(g^{\gamma})_{\gamma\in\Gamma}$ with $g^{\gamma}$ convex, $g^{\gamma}(t,\cdot,\cdot)\equiv0$ and independent from $y$ for any $\gamma\in\Gamma$, of which the corresponding family $(\rho^{\gamma}_t)_{\gamma\in\Gamma}$ of normalized, cash-additive and convex (resp.\ sublinear) risk measures verifies:
   \begin{equation}
        \rho_t(X)=\essmin_{\gamma\in\Gamma}\rho_t^{\gamma}(X)= \essmin_{\gamma\in\Gamma}\essmax_{\mu_t\in\mathcal{A}}\mathbb{E}_{\mathbb{Q}^{\mu}_t}\left[X-\int_t^TG^{\gamma}(s,\mu_s)ds\Big|\mathcal{F}_t\right],
        \label{minmax1}
    \end{equation}
 for any $X\in L^{\infty}(\mathcal{F}_T)$ and $t\in[0,T]$, with the same notation as in Equation~\eqref{eq:dualSA}.
   When $g$ is positively homogeneous then $g^{\gamma}$ is sublinear for any $\gamma\in\Gamma$ and the Fenchel transformation $G^{\gamma}$ vanishes on its proper domain. \\
    Furthermore, when $g$ is decreasing in $y$ then there exists a family of drivers $(g^{\gamma})_{\gamma\in\Gamma}$ decreasing in $y$ for any $\gamma\in\Gamma$.
    In detail, the corresponding  normalized, convex and cash-subadditive risk measures $(\rho^{\gamma})_{\gamma\in\Gamma}$ verify:
\begin{equation}
\rho_t(X)=\essmin_{\gamma\in\Gamma}\essmax_{(\beta_t,\mu_t)\in\mathcal{A}'}\mathbb{E}_{\mathbb{Q}^{\mu}_t}\left[D_t^{\beta}X-\int_t^TD_{t,s}^{\beta}G^{\gamma}(s,\beta_s,\mu_s)ds\Big|\mathcal{F}_t\right],
        \label{minmax2}
    \end{equation}
 for any $X\in L^{\infty}(\mathcal{F}_T)$ and $t\in[0,T]$, with the same notation as in Equation~\eqref{eq:dualCSA}.
   When $g$ is positively homogeneous then $g^{\gamma}$ is sublinear for any $\gamma\in\Gamma$ and the Fenchel transformation $G^{\gamma}$ vanishes on its proper domain.
   \label{corCSA}
\end{corollary}

\subsection{When can dynamic star-shaped risk measures be induced via BSDEs?}
We now provide a converse to our previous findings. Specifically, under appropriate conditions, we establish that the pointwise minimum of convex (resp.\ sublinear) risk measures, induced by a family of standard parameters, constitutes a star-shaped (resp.\ positively homogeneous) risk measure with dynamics governed by a BSDE.

Before proceeding, we present a general statement concerning the star-shaped property in the dynamic setting.
We demonstrate that star-shapedness is preserved when taking the minimum operation between star-shaped risk measures.
Additionally, the pointwise minimum risk measure inherits other significant properties.
\begin{proposition}
    Let $\Gamma$ be a set of indexes and let $(\rho^{\gamma})_{\gamma\in\Gamma}$ with $\rho^{\gamma}_t: L^{\infty}(\mathcal{F}_T)\to L^{\infty}(\mathcal{F}_t)$ be a family of cash-subadditive (resp.\ cash-additive), continuous from above and star-shaped (resp.\ positively homogeneous) dynamic risk measures such that there exists $c\in\R$ verifying $\rho^{\gamma}_t(0)=c$ for any $\gamma\in\Gamma$ and $t\in[0,T]$.
    If for each $X\in L^{\infty}(\mathcal{F}_T)$ there exists $\bar{\gamma}\in\Gamma$ such that \begin{equation}
    \rho_t(X):=\essinf_{\gamma\in\Gamma}\rho^{\gamma}_t(X)=\rho_t^{\bar{\gamma}}(X), \ \ \forall t\in[0,T],
    \label{minC}
    \end{equation}
    then $\rho_t$ is a regular, cash-subadditive (resp.\ cash-additive), continuous from above and star-shaped (resp.\ positively homogeneous) dynamic risk measure.
    In addition, if $\rho^{\gamma}_t$ is time-consistent for each $\gamma\in\Gamma$ and for each $X\in L^{\infty}(\mathcal{F}_T)$, and there exists $\tilde{\gamma}\in\Gamma$ such that
    \begin{equation}
    \rho_s(\rho_t(X))=\rho_s^{\tilde{\gamma}}(\rho_t^{\tilde{\gamma}}(X)),
    \label{TC}
    \end{equation}
    $0\leq s \leq t \leq T$,
    then $\rho_t$ is a time-consistent risk measure.
    \label{PM}
\end{proposition}
\begin{remark}
    As far as the (strong) time-consistency of $\rho_t$ is concerned, we underline that if $(\rho_t^{\gamma})_{\gamma\in\Gamma}$ is a family of time-consistent risk measures, then without assuming condition \eqref{TC} $\rho_t$ could fail to be time-consistent.
    The \textit{weaker} version of recursion given by $$\rho_s(\rho_t(X))\leq \rho_s(X) \ \ \forall X\in L^{\infty}(\mathcal{F}_T), 0\leq s\leq t\leq T,$$ is still valid, as is clear from Equation~\eqref{WTC2}.

   Furthermore, it is worth noting that condition \eqref{TC} holds for the family of risk measures given in Equation~\eqref{minsub}.
   Clearly, in this case time-consistency of $\rho_t$ is ensured given that $\rho_t$ is the first component of
   the solution of a BSDE.
    However, it is interesting that condition \eqref{TC} is automatically verified in this circumstance.
    According to the proof of Theorem~\ref{SuCO}, for any $X\in L^{\infty}(\mathcal{F}_T)$ we have:
    $$\rho_s(\rho_t(X))=\rho_s^{\rho(\rho_t(X)),Z^1}(\rho_t^{\rho(X),Z^2}),$$
    $0\leq s\leq t\leq T$,
    where $(\rho(\rho_t(X)),Z^1)$ and $(\rho(X),Z^2)$ are solutions to Equation~\eqref{eq: star-shapedBSDE} with terminal conditions $\rho_t(X)$ and $X$, respectively.
    By the flow property of BSDEs (cf.~Proposition~2.5 in \cite{EPQ97}) we know that $(\rho(\rho_t(X)),Z^1)\equiv(\rho(X),Z^2)\ d\mathbb{P}$-a.s.\ on $[0,t]$.
    This implies that for each $0\leq s \leq t \leq T$ we can choose the control $\bar{\gamma}=(\rho(X),Z^2)$ to represent $\rho_s(\rho_t(X))$ because it generates the same solution of $\tilde{\gamma}=(\rho(\rho_t(X)),Z^1)$ on the entire interval $[0,t]$, thus:
    $$\rho_s(\rho_t(X))=\rho_s^{\bar{\gamma}}(\rho_t^{\tilde{\gamma}})=\rho_s^{\tilde{\gamma}}(\rho_t^{\tilde{\gamma}}(X)) \ \ \forall s\in[0,t],$$
    which means that the index $\tilde{\gamma}$ verifies condition \eqref{TC}.
\end{remark}

\begin{proposition}
     Let us consider a set of indexes $\Gamma$ and a family of parameters $(g^{\gamma})_{\gamma\in\Gamma}$ with {$g^{\gamma}:\Omega\times [0,T]\times\mathbb{R}\times\mathbb{R}^n\to\mathbb{R}$} such that
    \begin{itemize}
            \item $g^{\gamma}$ is an $\mathcal{P}\times\mathcal{B}(\mathbb{R})\times\mathcal{B}(\mathbb{R}^n)$-measurable process for any $\gamma\in\Gamma$, where $\mathcal{P}$ is the $\sigma$-algebra of predictable sets on $\Omega\times[0,T]$.
        \item $(g^{\gamma})_{\gamma\in\Gamma}$ is $k$-equi-Lipschitz, for some $k>0$.
        \item $g^{\gamma}$ is convex (resp.\ sublinear) w.r.t.\ $(y,z)$ and $g^{\gamma}(\cdot,0,0)=0 \ \forall t\in[0,T]$ and for each $\gamma\in\Gamma$.
    \end{itemize}
     Let us assume that for each $X\in L^{2}(\mathcal{F}_T)$ there exists $\bar{\gamma}\in\Gamma$ such that, for all $t\in[0,T]$,
     \begin{equation}
     \essinf_{\gamma\in\Gamma}\rho_t^{\gamma}(X)=\rho_t^{\bar{\gamma}}(X) \ d\mathbb{P}\text{-a.s.},
     \label{minC1}
     \end{equation}
     where $(\rho^{\gamma}_t)_{\gamma\in\Gamma}$ is the family of convex risk measures generated via BSDEs with drivers $(g^{\gamma})_{\gamma\in\Gamma}$.
    Moreover, we require the family $(\rho^{\gamma}_t)_{\gamma\in\Gamma}$ to satisfy condition \eqref{TC}.

    Under these hypotheses, the pointwise minimum in Equation \eqref{minC1} defines a star-shaped (resp.\ positively homogeneous) risk measure $\rho_t$ induced via a BSDE.
    Moreover, the driver $g:\Omega\times[0,T]\times\mathbb{R}\times\mathbb{R}^n\to\mathbb{R}$ associated with the BSDE representing the star-shaped (resp.\ positively homogeneous) risk measure $\rho_t$ satisfies $g(\cdot,0,0)\equiv 0$ and is star-shaped (resp.\ positively homogeneous) with respect to $(y,z)$.
    \label{POBSDEs}
\end{proposition}

        We underline that if the drivers $(g^{\gamma})_{\gamma\in\Gamma}$ satisfy some further suitable conditions, then also $\rho_t$ inherits the corresponding properties.
        For instance, if the drivers $(g^{\gamma})_{\gamma\in\Gamma}$ are decreasing in $y$ then $(\rho^{\gamma})_{\gamma\in\Gamma}$ is a family of cash-subadditive risk measures (cf.\ \cite{ELKR09}).
        By Proposition~\ref{PM}, $\rho_t$ is then a cash-subadditive and star-shaped (resp.\ positively homogeneous) risk measure.
        Furthermore, if the drivers ($g^{\gamma})_{\gamma\in\Gamma}$ do not depend on $y$ then the family $(\rho^{\gamma})_{\gamma\in\Gamma}$ is cash-additive, thus $\rho_t$ is a cash-additive and star-shaped (resp.\ positively homogeneous) risk measure.
        Moreover, in these cases we can say something more about the properties of the driver $g$ driving the dynamics of $\rho_t$, as the following corollary states.
        We start with a proposition in the spirit of Proposition~\ref{POSS}.
    \begin{proposition}
    Let $g:\Omega\times [0,T]\times\mathbb{R}\times\mathbb{R}^n\to\mathbb{R}$ satisfy $L^2$-standard assumptions.
    Consider the BSDE:
\begin{equation*}
    \rho_t(X)=X+\int_t^Tg(s,\rho_s,Z_s)ds-\int_t^TZ_sdW_s,
\end{equation*}
where $X\in L^2(\mathcal{F}_T)$.
Then, $\rho_t(X)$ is cash-subadditive if and only if $g$ is decreasing in $y$.
    \label{IFFCS}
    \end{proposition}
    \begin{corollary}
        Under the same hypotheses and with the same notation as in Proposition~\ref{POBSDEs}, if the family of drivers $(g^{\gamma})_{\gamma\in\Gamma}$ does not depend on $y$, then the driver $g$ involved in the BSDEs representation of $\rho_t$ does not depend on $y$ and $g(\cdot,0)\equiv0$. Moreover, if the family of drivers $(g^{\gamma})_{\gamma\in\Gamma}$ is decreasing w.r.t.\ $y$ then the driver $g$ is decreasing w.r.t.\ $y$.
    \label{cor:CSA}
    \end{corollary}

\setcounter{equation}{0}

\section{Minimal Supersolutions of Star-Shaped BSDEs}\label{sec:supsolution}

In this section, we present some results on the min-max representation of minimal supersolutions of BSDEs with a star-shaped driver.
Existing results on minimal supersolutions of convex BSDEs and their dual representations are in \cite{DKRT14}.
Existence and uniqueness of minimal supersolutions are established in \cite{DHK13} and \cite{HKM12}.
We recall some relevant definitions and theorems from the aforementioned references.
Let us define:
    \begin{align*}
        &\mathcal{S}:=\left\{Y:\Omega\times[0,T]\to\mathbb{R}; \ Y \mbox{ is adapted and càdlàg}\right\}, \\
        &\mathcal{L}:=\left\{Z:\Omega\times[0,T]\to\R^n; \ Z \mbox{ is predictable and } \int_0^T|Z_s|^2ds<+\infty\right\}. \\
    \end{align*}
    Consider a BSDE of the form
    \begin{equation}
        Y_t=X+\int_t^Tg(s,Y_s,Z_s)ds-\int_t^TZ_sdW_s,
    \label{eq:bsdess}\end{equation}
     where $X\in L^0(\mathcal{F}_T)$ and the driver $g:\Omega\times[0,T]\times\R\times\R^n \to \R$ is a $\mathcal{P}\otimes\mathcal{B}(\R)\otimes\B(\R^n)$-measurable function satisfying the following assumptions:
    \begin{itemize}
   \item  (LSC) Lower semicontinuity w.r.t.\ $(y,z)$,
  \item  (POS) Positivity, i.e., $g\geq 0$,
    \item (NORM) Normalization, i.e., $g(\cdot,\cdot,0)\equiv 0$.
        \end{itemize}

     Henceforth, we write `SA' to refer to these assumptions.
     We say that a pair $(Y,Z)\in\mathcal{S}\times\mathcal{L}$ is a supersolution of \eqref{eq:bsdess} if it satisfies:
    \begin{equation}
        \begin{cases}
            Y_s-\int_s^tg(u,Y_u,Z_u)du+\int_s^tZ_udW_u\geq Y_t \ \ \mbox{ for every } 0\leq s \leq t \leq T, \\
            Y_T\geq X,
        \end{cases}
        \label{Defsupersolution}
    \end{equation}
    and $Z$ is admissible, i.e., $\int_{0}^{\cdot}Z_udW_s$ is a supermartingale.
    We define the set:
    \begin{equation*}
        \mathcal{A}(X,g):=\left\{(Y,Z)\in\mathcal{S}\times\mathcal{L}: \ Z \mbox{ is admissible and \eqref{Defsupersolution} holds}\right\}.
    \end{equation*}
   A supersolution $(\bar{Y},\bar{Z})$ is said to be a \textit{minimal} supersolution of the BSDE \eqref{Defsupersolution} if $(\bar{Y},\bar{Z})\in\mathcal{A}(X,g)$ and for any other $(Y,Z)\in\mathcal{A}(X,g)$ it holds that $\bar{Y}\leq Y$.
   \begin{theorem}[\cite{HKM12}]
   If the driver $g$ satisfies SA, $X\in L^0(\mathcal{F}_T)$, $X^-\in L^1(\mathcal{F}_T)$ and $\mathcal{A}(X,g)\neq\emptyset$, then there exists a unique minimal supersolution to Equation \eqref{Defsupersolution}.
   In particular, the value process $\bar{Y}$ is given by: $$\bar{Y}_t:=\essinf\left\{Y_t\in L^{0}(\mathcal{F}_T): (Y,Z)\in\mathcal{A}(X,g)\right\}.$$
   \label{EUsuper}
   \end{theorem}
   For completeness, we recall that is also possible to weaken the hypotheses of positivity and normalization, as done in Theorem~3.10 of \cite{HKM12}.
   In the following we will also make use of some further properties of the driver $g$:
   \begin{itemize}
       \item (CONV) $g$ is convex in $(y,z)$,
       \item (SS) $g$ is star-shaped in $(y,z)$,
       \item (DEC) $g$ is not-increasing in $y$.
   \end{itemize}
   Given a driver $g$ satisfying SA, the operator $\mathcal{E}^g:L^{\infty}(\mathcal{F}_T)\to\mathcal{S}\cup\{+\infty\}$ is defined as:
   \begin{equation}
   \mathcal{E}^g(X):=\begin{cases}
       \bar{Y} \ \ &\mbox{ if } \mathcal{A}(X,g)\neq\emptyset \\
       +\infty \ \ &\mbox{ otherwise}
   \end{cases},
   \label{solsuper}
   \end{equation}
   where $\bar{Y}$ is the unique minimal supersolution to Equation \eqref{Defsupersolution} with driver $g$ and final condition $X\in L^{\infty}(\mathcal{F}_T)$. Now we can state the main theorem of this section.
\begin{theorem}
Let $g$ be a driver satisfying SA and star-shapedness.
Then there exist a set of indexes $\Gamma$ and a family of drivers $(g^{\gamma})_{\gamma\in\Gamma}$ satisfying SA and convexity such that the corresponding operators defined by Equation \eqref{solsuper} verify:
\begin{equation}
    \mathcal{E}^g_t(X)=\essmin_{\gamma\in\Gamma}\mathcal{E}^{g^\gamma}_t(X) \ \ \forall t\in[0,T],\ \forall X\in L^{\infty}(\mathcal{F}_T).
    \label{minsuper}
\end{equation}
Furthermore, we have the min-max representation:
\begin{equation}
\mathcal{E}^g_t(X)=\essmin_{\gamma\in\Gamma}\esssup_{(\beta,q)\in\mathcal{G}\times\mathcal{Q}^{\infty}}\mathbb{E}_{\mathbb{Q}^{q}_t}\left[D^{\beta}_{t,T}X-\int_t^TD^{\beta}_{t,u}G^{\gamma}(u,\beta_u,q_u)du\Bigg|\mathcal{F}_t\right],
\label{minmaxsuper}
\end{equation}
for any $t\in[0,T]$ and $X\in L^{\infty}(\mathcal{F}_T)$.
In addition, if the driver $g$ is positively homogeneous, there exists a family of sublinear operators $\mathcal{E}_t^{g^{\gamma}}$ such that Equation~\eqref{minsuper} holds and the min-max representation reduces to:
\begin{equation*}
\mathcal{E}^g_t(X)=\essmin_{\gamma\in\Gamma}\esssup_{(\beta,q)\in\mathcal{G}_{\gamma}\times\mathcal{Q}_{\gamma}^{\infty}}\mathbb{E}_{\mathbb{Q}^{q}_t}\left[D^{\beta}_{t,T}X\right],
\end{equation*}
for any $t\in[0,T]$ and $X\in L^{\infty}(\mathcal{F}_T)$.
\label{superapresentation}
\end{theorem}
              We present the following corollary;  its proof  follows immediately from the proofs of Corollary~\ref{corCSA} and Theorem~\ref{superapresentation}.
              \begin{corollary}
                  Under the same hypotheses as in Theorem~\ref{superapresentation}:
                  \begin{itemize}
                      \item If $g$ is decreasing in $y$, then the family of drivers $(g^{\gamma})_{\gamma\in\Gamma}$ satisfies (DEC) and the theses of Theorem~\ref{superapresentation} hold considering $\mathcal{D}_+:=\left\{\beta\in\mathcal{D}:\beta\geq0\right\}$.
                      In particular, any element of the family $(\mathcal{E}_t^{g^{\gamma}}(X))_{\gamma\in\Gamma}$ is cash-subadditive.
                      \item If $g$ does not depend on $y$, then also the family of drivers $(g^{\gamma})_{\gamma\in\Gamma}$ does not depend on $y$ and the min-max representation becomes:
                      $$\mathcal{E}_t^{g^{\gamma}}(X)=\essmin_{\gamma\in\Gamma}\esssup_{q\in\mathcal{Q}^{\infty}}\mathbb{E}_{\mathbb{Q}^q_t}\left[X-\int_t^TG^{\gamma}(u,q_u)du\Bigg|\mathcal{F}_t\right],$$
                    for any $X\in L^{\infty}(\mathcal{F}_T)$ and $t\in[0,T]$. In particular, any element of the family $(\mathcal{E}_t^{g^{\gamma}}(X))_{\gamma\in\Gamma}$ is cash-additive.
                  \end{itemize}
                  \label{corsuper1}
              \end{corollary}
\setcounter{equation}{0}

\section{Examples}\label{sec:examples}

In this and the next section, we provide illustrations and applications within the field of finance of the theoretical results established in the previous sections.
We consider a general dynamic setting in which risk measures are induced by BSDEs and satisfy two fundamental properties: monotonicity and star-shapedness.
Both assumptions are standard and have been adopted, implicitly\footnote{Each convex risk measure is star-shaped (\cite{FR02}).} or explicitly, in a vast literature.
Monotonicity has an obvious financial meaning; star-shapedness reflects financial liquidity and size considerations (\cite{FS02,FR02,LR22,CCMTW22}).
We demonstrate the generality, flexibility and applicability of this family of dynamic risk measures and analyze their implications for capital allocation and portfolio choice.
Whereas the implications of specific classes of risk measures for major problems, such as asset pricing, portfolio choice, capital allocation, and risk management, have been analyzed under more restrictive assumptions --- imposing monotonicity, star-shapedness and other conditions such as cash-additivity and convexity --- we consider a general setting.
This allows us to characterize the primitive implications in their most general form.

Adhering to the conventional interpretation of the driver $g$ as an infinitesimal risk measure that aligns with an individual's local preferences, we underscore the relevance of the process $Z_t$ (the second component of the solution to Equation~\eqref{eq:bsde}) as the local volatility.
Specifically, for a fixed $X\in L^{\infty}(\mathcal{F}_T)$, within a time interval $[t, t+dt]$, $d\rho_t(X)$ has an expected value given by $\mathbb{E}[d\rho_t|\mathcal{F}_{t}]=g(t, \rho_t, Z_t)dt$.
This implies that, on a local scale (i.e., within the infinitesimal time interval), the expected value of the risk associated with $X$ is proportional to the driver $g$.
Furthermore, the infinitesimal quadratic variation $\mathbb{V}[d\rho_t|\mathcal{F}_t]$ of the BSDE's solution is given by $|Z_t|^2dt$, allowing $Z_t$ to be interpreted as the local conditional volatility of the risk measure.
For a more comprehensive understanding of this concept, interested readers may refer to Section~6.2 in \cite{BEK05}, Example~7.2 in \cite{ELKR09} and the references therein.

We now present three examples, two of which feature a non cash-additive/cash-subadditive risk measure, illustrating how this setting can reflect an investor's or regulator's preferences in a real financial market.
\begin{example}
    An example of a star-shaped driver that is neither convex nor concave is provided by the function $g:\Omega\times[0,T]\times\R\times\R\to\R$ defined as follows:
    $$g(t,y,z):=-\gamma_t|y|e^{-|y|}+\delta_t(z^2\mathbb{I}_{\{|z|\leq1\}}+z\mathbb{I}_{\{|z|> 1\}}),$$
    which may represent the preferences of an investor facing a stochastic interest rate $\gamma_t$ (a bounded positive process) and a risk aversion coefficient $\delta_t$ (also a bounded positive process).
    It is worth noting that the induced risk measure is neither convex nor concave, nor is it cash-additive or subadditive (as $g$ depends non-monotonically on $y$).
    However, it is star-shaped due to the star-shapedness of the driver $g$ (see Proposition \ref{POSS}). Moreover, the function $g$ is Lipschitz in $(y,z)$.
    Thus, all the assumptions of Theorem~\ref{SuCO} are fulfilled, and as a result, $\rho_t$ can be represented in the form of Equation~\eqref{minmax}.\end{example}
{\small
\begin{center}
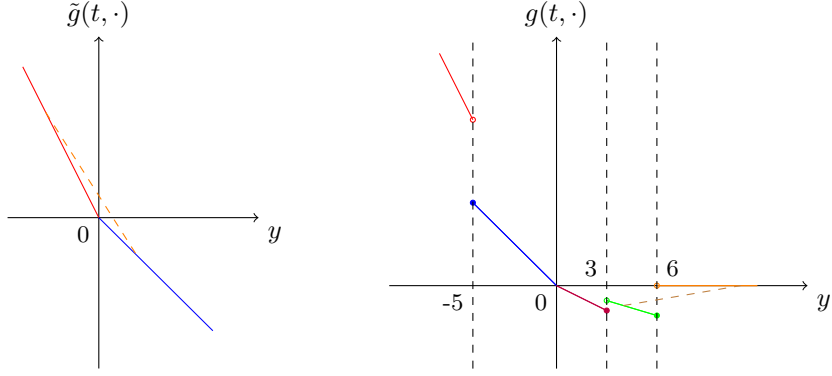
\begin{figure}[t]
\centering
\begin{tikzpicture}[scale=0.10]
\draw[->] (-12,0) -- (21,0) node[below right] {$y$};
\draw[->] (0,-20) -- (0,24) node[above ] {$\tilde{g}(t,\cdot)$};
\draw[color=red, domain=-10:0] plot(\x,{-2*\x}) ;
\draw[color=blue, domain=0:15] plot(\x,{-1*\x});
\draw (0,0) node[below left] {\small 0};
\draw[color=orange, dashed] (-7,14) -- (5,-5);
\end{tikzpicture}
\ \ \ \ \ \ \ \ \ \
\begin{tikzpicture}[scale=0.22]
\draw[->] (-10,0) -- (15,0) node[below right] {$y$};
\draw[->] (0,-5) -- (0,15) node[above] {$g(t,\cdot)$};
\draw[color=red, domain=-7:-5] plot(\x,{-2*\x}) ;
\draw[color=blue, domain=-5:0] plot(\x,{-1*\x});
\draw[color=purple, domain=0:3] plot(\x,{-0.5*\x});
\draw[color=green, domain=3:6] plot(\x,{-0.3*\x});
\draw[color=orange, domain=6:12] plot(\x,0);
\draw[color=blue, domain=-5:0] plot(\x,{-1*\x});
\filldraw[color=blue] (-5,5) circle (0.15);
\draw[color=purple, domain=0:3] plot(\x,{-0.5*\x});
\filldraw[color=purple] (3,-1.5) circle (0.15);
\draw[color=red] (-5,10) circle (0.15);
\draw[color=green, domain=3:6] plot(\x,{-0.3*\x});
\filldraw[color=green] (6,-1.8) circle (0.15);
\draw[color=green] (3,-0.9) circle (0.15);
\draw[color=orange, domain=6:12] plot(\x,0);
\draw[color=orange] (6,0) circle (0.15);
\draw[dashed] (-5,-5) -- (-5,15);
\draw (-5,0) node[below left] {\small -5};
\draw[dashed] (3,-5) -- (3,15);
\draw[dashed] (6,-5) -- (6,15);
\draw[color=brown, dashed] (4,-1.2) -- (11,0);
\draw (3,0) node[above left] {\small 3};
\draw (0,0) node[below left] {\small 0};
\draw (6,0) node[above right] {\small 6};
\end{tikzpicture}
\caption{\small Panel~(a): This figure illustrates an example of a function $y\mapsto \tilde{g}(t,y)$ for a fixed time $t\in[0,T]$ and scenario $\omega\in\Omega$ according to Example~7.2 in \cite{ELKR09}.
The function must be convex and made up of two half-lines joining at the origin with different angular coefficients.
These restrictions capture that the regulator's preferences only take into account ambiguity on the interest rate, without considering the dependence of the interest rate on the level of risk $y$.
Panel~(b): This figure illustrates a possible example of $y\mapsto g(t,y)$ for a fixed time $t\in[0,T]$ and scenario $\omega\in\Omega$, with thresholds $y_1=-5$, $y_2=3$ and $y_3=6$.
Any segment connecting the origin to a point on the graph of $g(t,\cdot)$ lies entirely within the epigraph of the function, thus $g(t,\cdot)$ is star-shaped. The dashed brown segment shows that $g(t,\cdot)$ is not convex.}
\label{fig:function-tikz}
\end{figure}
\end{center}}
\begin{example}
\vskip -0.7cm We analyze a financial situation in which star-shaped drivers play a crucial role, considering the interpretation of the driver $g$ as an infinitesimal risk measure locally compatible with the preferences of an investor or regulator.
To elaborate, we study a simplified setting involving a family of ambiguous interest rates $(\gamma^i_t)_{i=1}^n$ that range between positive bounded processes $[r^i_t,R^i_t]$, $i=1,\dots,n$.
It is reasonable to assume that the level of risk can influence the beliefs on the interest rate under certain circumstances.
Therefore, regulators may utilize `thresholds' to describe their view of the interest rate: if the driver's variable $y$ exceeds these thresholds, then the (random) interest rate will reach a different level.
Specifically, the regulator determines the thresholds $(y_i)_{i=1}^n$ such that $y_i\leq y_{i+1}$ for any $i=1,\dots,n$.
In this setting, the local preferences of regulators that are adverse to ambiguity and model the interest rate according to different risk levels may be represented by an infinitesimal generator of the form (compare with Example~7.2 in \cite{ELKR09} and refer to Figure~\ref{fig:function-tikz}, Panels~(a)--(b), for graphical illustrations):
 \begin{equation*}
    g(t,y)=\begin{cases}
        \ds\sup_{r^1_t\leq\gamma^1_t\leq R_t^1}\{-\gamma^1_ty\}=R^1_ty^--r_t^1y^+ &\mbox{ if } y< y_1\leq0, \\
        \ds\sup_{r^2_t\leq\gamma^2_t\leq R_t^2}\{-\gamma^2_ty\}=R^2_ty^--r_t^2y^+&\mbox{ if } y_1\leq y<y_2\leq0, \\
        \vdots \\
         \ds\sup_{r^{i+1}_t\leq\gamma^{i+1}_t\leq R_t^{i+1}}\{-\gamma^{i+1}_ty\}=R^{i+1}_ty^--r_t^{i+1}y^+&\mbox{ if } 0\geq y_i\leq  y \leq y_{i+1}\geq0, \\
         \vdots \\
        \ds\sup_{r^{n}_t\leq\gamma^{n}_t\leq R_t^{n}}\{-\gamma^{n}_ty\}=R^{n}_ty^--r_t^{n}y^+&\mbox{ if } y_{n-1}< y \leq
        y_{n}, \\
        0&\mbox{ if } y> y_n.
    \end{cases}
\end{equation*}

\noindent Here, $y^+$ and $y^-$ represent the positive and negative parts of $y\in\R$, respectively. The corresponding risk measure, denoted as $\rho_t(X)$, corresponds to the first component of the solution to the BSDE:
\begin{equation*}
    \rho_t(X)=X+\int_t^Tg(s,\rho_s(X))ds-\int_t^TZ_sdW_s, \ \ \mbox{ for any } X\in L^{\infty}(\mathcal{F}_T).
\end{equation*}
If the value process $\rho_t(X)$ (i.e., the risk at time $t$ of $X\in L^{\infty}(\mathcal{F}_T)$) exceeds the level $y_i$, then the investor will `locally' face a discontinuity, given by a different (worst case) interest/discount rate.
The meaning of $\gamma\equiv0$ when $y$ exceeds the threshold $y_n$, can be understood as a `barrier': when the risk is too high, the discount rate is $0$ and there is a null discount for future positions $X\in L^{\infty}(\mathcal{F}_T)$.
Let us note that, in general, by construction we have star-shapedness of $g$, whereas convexity is not necessarily verified (see Figure~\ref{fig:function-tikz}, Panel~(b), for an immediate visual check of this fact).
Thus the corresponding risk measure $\rho_t$ is star-shaped, according to Proposition~\ref{POSS}.
Moreover, although the positivity of the driver may not hold, it is important to note that $g$ is bounded from below by a constant.
As a result, the established results regarding the existence and uniqueness of supersolutions remain applicable (cf.\ \cite{HKM12}).
Hence, the theory presented in Section~\ref{sec:supsolution} can still be applied.
We would like to note that the function $g$ is neither increasing nor decreasing with respect to the variable $y$ (see Figure~\ref{fig:function-tikz}, Panel~(b)).
As a result, the risk measure $\rho_t$ may fail to satisfy both cash-subadditivity and cash-superadditivity, even though $\rho_t$ describes a possible financial framework, as explained above.

Therefore, there exists a family of convex supersolutions $(\mathcal{E}_t^{g^\gamma})_{\gamma\in\Gamma}$ such that the representation in Equation~\eqref{minsuper} holds.
\end{example}

\begin{example}
Given a driver of the form $g(t,z):=\frac{1}{\lambda_t}z^2$ (inducing an entropic risk measure), the coefficient $\lambda_t$ may be interpreted as the risk aversion of the investor to the market volatility.
Therefore, it is reasonable to suppose that an investor changes her/his preferences (and her/his risk aversion) when the volatility varies.
In the following, we assume that high volatility can increase the risk tolerance of an investor.
In particular, if an investor is seeking to buy or sell securities, high volatility can create opportunities for profit, if able to correctly predict market movements.
Similarly, traders who specialize in short-term investments may benefit from high volatility as it can create opportunities for quick profits.
Additionally, high volatility can also be beneficial for companies that engage in hedging activities as it allows them to purchase options or other instruments at lower prices (see also \cite{BCM18} for a detailed discussion).

From a mathematical perspective, we fix $n\geq 2$ $(\lambda^i_t)_{i=1}^{n}$ bounded, strictly positive stochastic processes such that $\lambda^i_t\geq\lambda^{i+1}_t$ for any $i=1,\dots,n-1$, and $(z_i)_{i=1}^{n}$ thresholds such that $|z_i|< |z_{i+1}|$ for any $i=1,\dots,n-1$.
An investor may choose the following driver (see Figure~\ref{fig:function-tikz2} for a graphical representation):
\begin{equation}
g(t,z)=
\begin{cases}
\frac{1}{\lambda^1_t}|z|^2 &\mbox{ if } |z|\leq |z_1|, \\
\frac{1}{\lambda_t^2}|z|^2 &\mbox{ if } |z_1|< |z| \leq |z_2|, \\
\vdots \\
\frac{1}{\lambda_t^{n}}|z|^2 &\mbox{ if } |z_{n-1}|< |z| \leq |z_{n}|, \\
+\infty &\mbox{ if } |z|\geq |z_n|.
\end{cases}
\label{eq:ex2}
\end{equation}

This driver has a clear interpretation: when the volatility of the investment exceeds a certain threshold, the investor's risk aversion experiences a sudden variation.
If the volatility is very low (less than $|z_1|$), then the investor has the highest possible risk aversion coefficient.
On the other hand, if the volatility exceeds $|z_n|$, then the investor's risk tolerance is modeled by a value $+\infty$ in the driver.
This situation may occur when the investment is always deemed profitable.
Note that while $g$ is star-shaped, it may not necessarily be convex.
Furthermore, the corresponding risk measure is cash-additive, as $g$ does not depend on $y$.
Finally, since $g$ satisfies the SA condition, we can apply Corollary~\ref{corsuper1} to obtain the min-max representation in Equation~\eqref{minmaxsuper}.

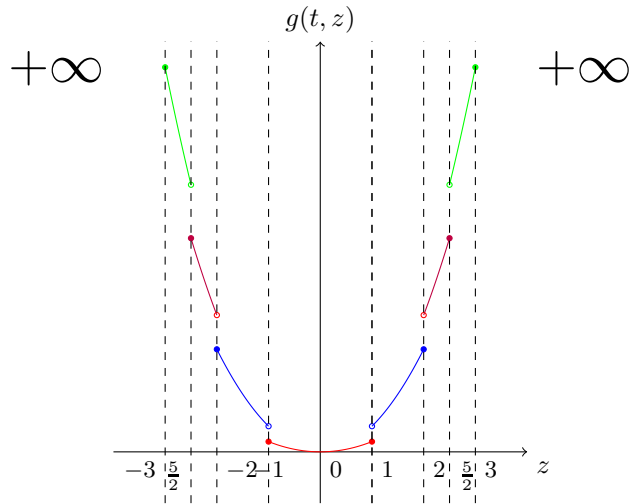
\begin{figure}[t]
    \centering
    \begin{tikzpicture}[scale=0.68]
        \draw[->] (-4,0) -- (4,0) node[below right] {$z$};
        \draw[->] (0,-1) -- (0,8) node[above] {$g(t,z)$};
        \draw[color=red, domain=-1:1] plot(\x,{\x*\x/5}) ;
        \draw[color=blue, domain=-2:-1] plot(\x,{\x*\x/2});
        \draw[color=blue, domain=1:2] plot(\x,{\x*\x/2});
        \draw[color=purple, domain=-2.5:-2] plot(\x,{\x*\x/1.5});
        \draw[color=purple, domain=2:2.5] plot(\x,{\x*\x/1.5});
        \draw[color=green, domain=-3:-2.5] plot(\x,{\x*\x/1.2});
        \draw[color=green, domain=2.5:3] plot(\x,{\x*\x/1.2});

    \filldraw[color=red] (-1,1/5) circle (0.05);
    \filldraw[color=red] (1,1/5) circle (0.05);
    \filldraw[color=blue] (-2,2) circle (0.05);
    \filldraw[color=blue] (2,2) circle (0.05);
    \filldraw[color=purple] (-5/2,25/6) circle (0.05);
    \filldraw[color=purple] (5/2,25/6) circle (0.05);
    \filldraw[color=green] (-3,7.5) circle (0.05);
    \filldraw[color=green] (3,7.5) circle (0.05);
    \draw[color=green] (5/2,5.2083333) circle (0.05);
    \draw[color=green] (-5/2,5.2083333) circle (0.05);
    \draw[color=red] (-2,2.666666) circle (0.05);
    \draw[color=red] (2,2.666666) circle (0.05);
    \draw[color=blue] (1,0.5) circle (0.05);
    \draw[color=blue] (-1,0.5) circle (0.05);

         \draw[dashed] (-1,-1) -- (-1,8);
         \draw[dashed] (1,-1) -- (1,8);
         \draw[dashed] (1,-1) -- (1,8);
         \draw[dashed] (-2,-1) -- (-2,8);
         \draw[dashed] (-2.5,-1) -- (-2.5,8);
         \draw[dashed] (2,-1) -- (2,8);
         \draw[dashed] (2.5,-1) -- (2.5,8);
         \draw[dashed] (-3,-1) -- (-3,8);
         \draw[dashed] (3,-1) -- (3,8);

         \draw (-1,0) node[below] {\small $-1$};
         \draw (1,0) node[below right] {\small $1$};
        \draw (-2,0) node[below right] {\small $-2$};
         \draw (2,0) node[below right] {\small $2$};
         \draw (-2.5,0) node[below left] {\small $\frac{5}{2}$};
         \draw (2.5,0) node[below right] {\small $\frac{5}{2}$};
        \draw (0,0) node[below right] {\small $0$};
        \draw (-3,0) node[below left] {\small $-3$};
        \draw (3,0) node[below right] {\small $3$};
        \draw (4,8) node[below right] {\huge $+\infty$};
        \draw (-4,8) node[below left] {\huge $+\infty$};

    \end{tikzpicture}
    \caption{\small This figure illustrates a potential star-shaped driver of the form given in Equation~\eqref{eq:ex2}, when $z\in\R$.
    Here, the thresholds are selected as $z_1=1,z_2=2,z_3=\frac{5}{2},z_4=3.$}
    \label{fig:function-tikz2}
\end{figure}
\end{example}
\setcounter{equation}{0}

\section{Applications}
\label{sec:app}
In this section, we explore two potential applications of star-shaped risk measures: capital allocation rules and portfolio choice.
In both applications, we highlight the possibility of reducing the star-shaped problem to a family of convex problems that can in principle be dealt with using standard techniques.

\subsection{Capital allocation rules for dynamic star-shaped risk measures}

In this subsection, we examine capital allocation rules (CARs) for risk measures that exhibit star-shapedness and of which the dynamics are governed by BSDEs.
In intuitive terms, a CAR describes how to equitably subdivide the capital requirement (or margin) of an aggregate risky position $X$ among the sub-portfolios or business lines that make up $X$, based on some financially sound criteria.
More explicitly, given a risk measure $\rho$, an aggregate position $X$ and its business lines $X_1, \ldots, X_n$, a CAR specifies how to apportion the risk associated with $X$ (according to $\rho$) across $X_1, \ldots, X_n$ by allocating capital $k_i$ to $X_i$ such that $\rho(X)=\sum_{i=1}^{n}k_i$, when the so-called \textit{full allocation} of the CAR is achieved.

Motivated by these qualitative observations, we recall the formal definition of CARs and discuss some properties that they may satisfy (see, e.g.,  \cite{D01,K05,CR18} for an axiomatic approach and see, for the dynamic setting, \cite{KO14,MGRO22,RZ23} and the references therein).

\begin{definition}
    Let $\rho_t:L^{\infty}(\mathcal{F}_T)\to L^{\infty}(\mathcal{F}_t)$ be a dynamic risk measure.
    We say that
             $\Lambda_t:L^{\infty}(\mathcal{F}_T)\times L^{\infty}(\mathcal{F}_T)\to L^{\infty}(\mathcal{F}_t)$ is a CAR with underlying risk measure $\rho_t$ if the following `consistency condition' holds: $$\Lambda_t(X,X)=\rho_t(X) \mbox{ for any } X\in L^{\infty}(\mathcal{F}_T).$$
             Furthermore, $\Lambda_t:L^{\infty}(\mathcal{F}_T)\times L^{\infty}(\mathcal{F}_T)\to L^{\infty}(\mathcal{F}_t)$ is called an \textit{audacious} CAR if only  $\Lambda_t(X,X)\leq\rho_t(X)$ for any $X\in L^{\infty}(\mathcal{F}_T)$.
             \label{DefCARs}
\end{definition}
Below we present a non-exhaustive list of axioms that CARs may satisfy; see, among many others, \cite{D01,K05,MGRO22} for further details and for a  discussion of the corresponding financial interpretation. \\ \\
\noindent - \textit{Normalization:} for any $t\in[0,T]$ and $Y\in L^{\infty}(\mathcal{F}_T)$, $\Lambda_t(0,Y)=0$.

\noindent - \textit{Monotonicity:} if $X\geq Y$ ($X,Y\in L^{\infty}(\mathcal{F}_T)$), then $\Lambda_t(X,Z)\geq \Lambda_t(Y,Z)$ for any $t\in [0,T]$, $Z \in L^{\infty}(\mathcal{F}_T)$.

\noindent - \textit{No-undercut}: $\Lambda_t(X,Y)\leq \rho_t(X)$ for any $t\in[0,T]$, $X,Y\in L^{\infty}(\mathcal{F}_T)$.

\noindent - \textit{1-cash-additivity}: $\Lambda_t(X+c_t,Y)=\Lambda_t(X,Y)+c_t$ for any $t\in[0,T]$, \mbox{$c_t\in L^{\infty}(\mathcal{F}_t),$} $X,Y \in L^{\infty}(\mathcal{F}_T)$.

\noindent - \textit{1-cash-subadditivity}: $\Lambda_t(X+c_t,Y)\leq\Lambda_t(X,Y)+c_t$ for any $t\in[0,T]$, \mbox{$c_t\in L^{\infty}_+(\mathcal{F}_t),$} $X,Y \in L^{\infty}(\mathcal{F}_T)$.

\noindent - \textit{Sub-allocation}: for any $t \in [0,T]$,
$$\Lambda_t \left(X,X \right) \geq \sum_{i=1}^{n} \Lambda_t \left(X_i,X \right)$$
holds for any $X_i,X\in L^{\infty}(\mathcal{F}_T)$, $i=1,\dots,n$, with $\sum_{i=1}^n X_i=X$.

\noindent - \textit{Weak-convexity}: for any $t \in [0,T]$,
$$\Lambda_t \left(\ds\sum_{i=1}^n a_iX_i,X \right)\leq \ds\sum_{i=1}^n a_i\Lambda_t(X_i,X)$$
holds for any $X, X_i\in L^{\infty}(\mathcal{F}_T),a_i\in[0,1]$, $i=1,\dots,n$, with $\sum_{i=1}^na_i=1$ and $\sum_{i=1}^na_iX_i=X$.\medskip \\
We emphasize that weak-convexity plays a similar role for capital allocation as star-shapedness in the context of risk measures.

\subsubsection{Star-shaped CARs: The subdifferential case}

CARs with convex risk measures and their relation with BSDEs have been widely studied in the literature (e.g., \cite{D01,CR18,MGRO22,RZ23}).
Henceforth, we consider a star-shaped risk measure $\rho_t$ induced via Equation~\eqref{eq: star-shapedBSDE}, whose dual representation is provided in Equation~\eqref{minmax}, assuming $g(\cdot,0,0)\equiv 0$ (normalization) and $X\in L^{\infty}(\mathcal{F}_T)$.
More explicitly, for each $Y\in L^{\infty}(\mathcal{F}_T)$ there exist $\gamma^Y\in\Gamma$ and $(\beta^Y,q^Y)\in\mathcal{G}\times\mathcal{Q}^{\infty}$ bounded processes such that:
\begin{equation}
\rho_t(Y)=\mathbb{E}_{\mathbb{Q}^{q^Y}_t}\left[D^{\beta^Y}_tY-\int_t^TD^{\beta^Y}_{t,s}G^{\gamma^Y}(s,\beta^Y_s,q^Y_s)ds\Bigg|\mathcal{F}_t\right].
\label{eq:subSS}
\end{equation}
It is important to note that the optimal parameters $\gamma^Y,\beta^Y,$ and $q^Y$ in the representation of $\rho_t(Y)$ depend on $Y$.
This means that there may not be a unique $(\gamma,\beta,q)\in\Gamma\times\mathcal{G}\times\mathcal{Q}^{\infty}$ that satisfies the min-max condition in Equation~\eqref{minmax}.
This situation is similar to the one encountered when dealing with CARs based on subdifferentiable convex risk measures. Therefore, we need to select one of the possible optimal scenarios in the min-max representation of $\rho_t$ going forward.
\begin{definition}
Let us define the subdifferential CAR $\Lambda^{SS}_t:L^{\infty}(\mathcal{F}_T)\times L^{\infty}(\mathcal{F}_T)\to L^{\infty}(\mathcal{F}_t)$ as:
\begin{equation*}
\Lambda^{SS}_t(X,Y):=\mathbb{E}_{\mathbb{Q}^{q^Y}_t}\left[D^{\beta^Y}_tX-\int_t^TD^{\beta^Y}_{t,s}G^{\gamma^Y}(s,\beta^Y_s,q^Y_s)ds\Bigg|\mathcal{F}_t\right],
\end{equation*}
where $(\gamma^Y,\beta^Y,q^Y)\in\Gamma\times\mathcal{G}\times\mathcal{Q}^{\infty}$ is an optimal scenario in the representation~\eqref{eq:subSS} of the underlying dynamic star-shaped risk measure $\rho_t(Y)$.
\end{definition}
\noindent We observe that the consistency condition $\Lambda_t(Y,Y)=\rho_t(Y)$ for any $Y\in L^{\infty}(\mathcal{F}_T)$ is verified, thus $\Lambda^{SS}_t$ is indeed a CAR in the sense of Definition~\ref{DefCARs}.

     The name `subdifferential' star-shaped CARs can be motivated similarly as for convex risk measures, where
the subdifferential CAR
$$\Lambda^{conv}_t(X,Y)=\mathbb{E}_{\mathbb{Q}^{q^Y}_t}\left[\left. D^{\beta^Y}_tX-\int_t^TG(s,\beta^Y_s,q_s^Y) ds \right| \mathcal{F}_t\right],$$
is defined by means of a possible optimal scenario $(D^{\beta^Y}_t,\mathbb{Q}^{q^Y}_t)$ in the dual representation of $\rho^{conv}_t(Y)$ and where all optimal $(D^{\beta^Y}_t,\mathbb{Q}^{q^Y}_t)$ belong to the subdifferential of $\rho_t(Y)$ (cf. \cite{RZ23}).

    In the more general setting of Equation~\eqref{minmax}, where the risk measure $\rho_t$ may not necessarily be convex, the subdifferential CAR $\Lambda^{SS}_t$ is given by the subdifferential CAR associated to the convex risk measure $\rho^{\gamma^Y}_t (Y)$ corresponding to some $\gamma^Y\in\Gamma$ such that $\rho_t(Y)=\rho_t^{\gamma^Y}(Y)$.
    That is, for any fixed portfolio $Y$, $\Lambda^{SS}_t(\cdot,Y)$ is the subdifferential CAR associated with the underlying risk measure $\rho^{\gamma^{Y}}_t$, and its representation does not rely on the elements of $\partial\rho_t (Y)$ (which can be an empty set) but on the elements of the family of subgradients $(\partial\rho_t^{\gamma} (Y))_{\gamma\in\Gamma}$, which are non-empty sets due to the convexity of $\rho_t^{\gamma}$ for any $\gamma\in\Gamma$.

Let us analyze the properties of the subdifferential CAR $\Lambda^{SS}_t$.
\begin{proposition}
    $\Lambda^{SS}_t:L^{\infty}(\mathcal{F}_T)\times L^{\infty}(\mathcal{F}_T)\to L^{\infty}(\mathcal{F}_t)$ is a monotone and weakly-convex CAR, which also satisfies sub-allocation.
    Moreover, when $\rho_t$ is cash-additive  (resp.\ cash-subadditive), then $\Lambda^{SS}_t$ is 1-cash-additive (resp.\ 1-cash-subadditive).
\end{proposition}

We omit the proof, which can be given using arguments similar to those in \cite{MGRO22,RZ23}.
\begin{remark}
    A well-known property of CARs for convex risk measures is the no-undercut property.
    In the case of star-shaped risk measures we can establish a related property, which is weaker but linked to no-undercut. Indeed, fixing an optimal scenario $(\gamma^{Y},\beta^Y,q^Y)\in\Gamma\times\mathcal{G}\times\mathcal{Q}^{\infty}$ in the representation of $\rho_t(Y)$ we have that:
    \begin{align*}\Lambda^{SS}_t(X,Y)&=\mathbb{E}_{\mathbb{Q}^{q^Y}_t}\left[\left.D^{\beta^Y}_tX \right| \mathcal{F}_t\right]-c^{\gamma^Y}_t(D^{\beta^Y}_t\mathbb{Q}^{q^Y}_t) \\
    &\leq\essmax_{(\beta,q)\in \mathcal{G}\times\mathcal{Q}^{\infty}} \mathbb{E}_{\mathbb{Q}^q_t}\left[ \left.D^{\beta}_tX-c^{\gamma^Y}_t(D^{\beta}_t\mathbb{Q}^{q}_t) \right| \mathcal{F}_t\right] \\
    &=\rho^{\gamma^Y}_t(X), \ \forall X\in L^{\infty}(\mathcal{F}_T),
    \end{align*}
    where
    \begin{equation}
c^{\gamma^Y}_t(D^{\beta^Y}_t\mathbb{Q}^{q^Y}_t):=\mathbb{E}_{\mathbb{Q}^{q^Y}_t}\left[\int_t^TD^{\beta^Y}_{t,s}G^{\gamma^Y}(s,\beta^Y_s,q^Y_s)ds|\mathcal{F}_t\right].
        \label{penfun}
    \end{equation}

    The inequality $\Lambda^{SS}_t(X,Y)\leq \rho^{\gamma^Y}_t(X)$ instead of the usual no-undercut property $\Lambda^{SS}_t(X,Y)\leq \rho_t(X)$ is due to the fact that an investor with preferences represented by $g^{\gamma^Y}$ has no incentive to separate the sub-portfolio $X$ from $Y$, as the capital needed to cover the risk of $X$ as a stand-alone portfolio measured through $\rho^{\gamma^Y}_t$ is greater than or equal to the capital required when $X$ is considered as a sub-portfolio of $Y$.
    However, another investor with different preferences might still prefer to separate $X$ from $Y$, if $\Lambda^{SS}_t(X,Y)>\rho^{\gamma}_t(X)$ for some $\gamma\in\Gamma$.
    This is a reasonable interpretation in the context of star-shaped risk measures, where the risk evaluation of a portfolio $X$ depends upon the specific element chosen in the set $\Gamma$ of possible `convex' risk measures available to the investor.
 \end{remark}

 \subsubsection{Star-shaped CARs: Aumann-Shapley capital allocation}
 We now focus on studying Aumann-Shapley capital allocation in the context of a star-shaped underlying dynamic risk measure.
 It is important to note that Aumann-Shapley CARs have a close connection with Aumann-Shapley values and their relevance in cooperative game theory (see, e.g., \cite{D01,CR18} and the  references therein).

Let us consider a star-shaped risk measure $\rho_t$ induced via BSDE \eqref{eq: star-shapedBSDE}.
For each $Y\in L^{\infty}(\mathcal{F}_T)$ we label by $\rho^{\gamma^Y}_t$ the convex risk measure such that $$\rho_t(Y)=\min_{\gamma\in\Gamma}\rho_t^{\gamma} (Y)=\rho_t^{\gamma^Y}(Y),$$ with $\gamma^Y$ one (of the possible) suitable $\gamma\in\Gamma$ attaining the minimum (cf.\ also Equation \eqref{minsub}).
We define `$\omega\times\omega$' the functions:
\begin{align}
    \Lambda^{AS}_t(X,Y)&=\int_0^1\mathbb{E}_{\mathbb{Q}_t^{q^{mY}}}\left[\left.D^{\beta^{mY}}_tX \right|\mathcal{F}_t \right]dm,
    \label{ASCAR} \\
     \Lambda^{p-AS}_t(X,Y)&=\int_0^1\Lambda_t^{SS}(X,mY)dm,
     \label{penCAR}
\end{align}
where $\beta^{\cdot}$ and $q^{\cdot}$ are optimal scenarios in the dual representation of the convex risk measure $\rho_t^{\gamma^{Y}}(\cdot).$
Let us stress that equations \eqref{ASCAR} and \eqref{penCAR} coincide when the penalty function $c_t$ defined in Equation \eqref{penfun} is null for any $t\in[0,T]$ (e.g., when the driver $g$ is positively homogeneous).
 We call $\Lambda_t^{AS}$ the Aumann-Shapley CAR, whereas we refer to $\Lambda^{p-AS}_t$ as the \textit{penalized} Aumann-Shapley CAR.
 These CARs have counterparts that are well-known in the context of static or dynamic convex risk measures satisfying cash-additivity/subadditivity (cf.\ \cite{MGRO22,RZ23}).

 The following result extends those of \cite{MGRO22,KO14,RZ23} to the current setting.

\begin{proposition}
Let $\rho_t:L^{\infty}(\mathcal{F}_T)\to L^{\infty}(\mathcal{F}_t)$ be a dynamic star-shaped risk measure induced by Equation \eqref{eq: star-shapedBSDE}, with $g(\cdot,0,0)\equiv0$.
The corresponding Aumann-Shapley CAR $\Lambda_t^{AS}:L^{\infty}(\mathcal{F}_T)\times L^{\infty}(\mathcal{F}_T)\to L^{\infty}(\mathcal{F}_t)$ is a monotone and normalized CAR satisfying full allocation.
Furthermore, for any $X,Y \in L^{\infty}(\mathcal{F}_T)$,
\begin{equation*}
    \Lambda_t^{AS}(X,Y)=\mathbb{E}_{\mathbb{P}}\left[\hat{L}^Y(T,t)X\Big|\mathcal{F}_t\right],
\end{equation*}
with
\begin{equation*}
    \hat{L}^Y(T,t)=\int_0^1L^{mY}(T,t)dm\triangleq \int_0^1\exp\left(-\frac{1}{2}\int_t^T |q_s^{mY}|^2ds-\int_t^T\beta_s^{mY} ds -\int_t^T\langle q_s^{mY},dW_s\rangle\right)dm,
\end{equation*}
where $(\beta_t^{\cdot},q_t^{\cdot})\in\partial\rho^{\gamma^Y}(\cdot)$ is an optimal scenario in the dual representation of $\rho^{\gamma^Y}_t(\cdot)$. In addition, when $\rho_t$ is cash-additive (resp.\ cash-subadditive) then $\Lambda^{AS}_t$ is 1-cash-additive (resp.\ 1-cash-subadditive).
\label{prop:AS}
Under the same hypotheses, $\Lambda_t^{p-AS}:L^{\infty}(\mathcal{F}_T)\times L^{\infty}(\mathcal{F}_T)\to L^{\infty}(\mathcal{F}_t)$
is a monotone, \textit{audacious} CAR satisfying the `modified' no-undercut property \mbox{$\Lambda_t^{p-AS}(X,Y)\leq\rho_t^{\gamma^{Y}}(X)$} and sub-allocation.
\end{proposition}

\subsection{Portfolio choice with dynamic star-shaped risk measures}
\label{sec:portchoice}
In this subsection, we demonstrate the feasibility of solving the canonical portfolio optimization problem in a Brownian setting even when we only impose the assumption of star-shapedness on the respective dynamic measure of risk.
The existing literature has extensively studied portfolio optimization with convex or quasi-convex risk measures, both in a static setting (see e.g., \cite{RS06a,MGRO15a}) and in a dynamic one (see e.g., \cite{LS14,RK00,HIM05}).
However, without convexity, such optimization problems tend to be highly complex and, in general, not easily solvable.
We will establish that star-shaped risk measures exhibit appealing properties for the portfolio choice problem.
Specifically, the portfolio choice problem with star-shaped risk measures is reduced to a collection of convex problems, some of which are well-documented in the literature.

In the subsequent discussion, we will employ a suitable approach to tackle the dynamic portfolio optimization problem in a not necessarily complete financial market.
We start by describing the main features of the underlying framework.
We consider a risk-free bond evolving with a null interest rate.
Let $(W_t)_{t\in[0,T]}$ be an $n$-dimensional Brownian motion and consider $d\leq n$ stocks.
The dynamics of the $i$-th stock is given by the SDE:
\begin{equation*}
\frac{dS^i_t}{S^i_t}=b^i_tdt+\sigma^i_tdW_t \ \ \forall i=1,\dots,d.
\end{equation*}
We assume that $b^i,\sigma^i$ are predictable, $\R\times\R^{d}$-valued and bounded processes.
Furthermore, the matrix $\sigma_t$ is supposed to be full-rank and $\sigma_t\sigma_t^T$ is uniform elliptic, i.e., $k_1\mathcal{I}_d\leq \sigma_t\sigma_t^T\leq k_2 \mathcal{I}_d$ with $\mathcal{I}_d$ the identity function on $\R^d$ and $0<k_1<k_2$ positive constants.
For each $i=1,\dots,d$, denote with $\pi^i_t$ the process representing the capital invested in the $i$-th stock at time $t$.
The corresponding number of shares is then given by $\pi_t^i/S_t^i$. The wealth process $X^\pi$ of a predictable trading strategy $\pi$ with an initial capital of $x_0$ is defined as follows:
\begin{equation*}
    X^{\pi}_t=x_0+\sum_i^d\int_0^t\frac{\pi_s^i}{S_s^i}dS^i_s=x_0+\int_0^t\pi_sb_sds+\int_0^t\pi_s\sigma_sdW_s.
\end{equation*}
Fixing a (not necessarily convex) compact set $\Pi\subset\R^d$, we define the set of admissible strategies as $$\mathcal{K}:=\left\{\pi:\Omega\times[0,T]\to\R^d \mbox{ predictable processes valued in } \Pi \ d\mathbb{P}\times dt\mbox{-a.s.}\right\}.$$ Under these assumptions our market is free of arbitrage (see \cite{LS14,HIM05}).
Given $\pi \in \mathcal{K}$, the investor, who possesses a payoff $F$, will adhere to the trading strategy $\pi$ until maturity, resulting in a total wealth of $X^{\pi}_T + F$.

To assess the portfolio, the investor selects a star-shaped risk measure $\rho_t$ induced via BSDE \eqref{eq: star-shapedBSDE}.
It is worth noting that we use $\rho_t$ for consistency of exposition, whereas the typical sign convention would be to employ $-\rho_t$.
Adopting this convention, $-\rho_t$ arises as the \textit{maximum} of \textit{concave} risk measures.
Under these preferences, the dynamic portfolio choice problem can be formulated as follows:
\begin{equation}
    V_t(F):=\essmin_{\pi\in\mathcal{K}}\rho_t(X^{\pi}_T+F).
    \label{eq:portfoliochoice1}
\end{equation}
According to Theorem~\ref{SuCO}, this problem can be expressed as:
\begin{align}
V_t(F)&=\essmin_{\pi\in\mathcal{K}}\essmin_{\gamma\in\Gamma}\rho^{\gamma}_t(X^{\pi}_T+F) \notag \\ &=\essmin_{\gamma\in\Gamma}\essmin_{\pi\in\mathcal{K}}\rho^{\gamma}_t(X^{\pi}_T+F).
\label{eq:portfoliochoice}
\end{align}

To approach problem~\eqref{eq:portfoliochoice}, we can take the following steps:
for each fixed $\gamma\in\Gamma$, we solve the corresponding convex optimization problem
 \begin{equation}
     V_t^{\gamma}(F):=\essmin_{\pi\in\mathcal{K}}\rho^{\gamma}_t(X^{\pi}_T+F),
     \label{eq:convopt}
 \end{equation}
 finding the optimal strategy $\bar{\pi}^{\gamma}$ associated with the chosen $\gamma\in\Gamma$. Then, we take the minimum over $\gamma\in\Gamma$, i.e.,
 \begin{align*}
 V_t(F)&=\essmin_{\gamma\in\Gamma}V_t^{\gamma}(F)=\essmin_{\gamma\in\Gamma}\rho^{\gamma}_t(X^{\bar{\pi}^{\gamma}}_T+F) \\
 &=\rho^{\bar{\gamma}}_t(X^{{\bar{\pi}}^{\bar{\gamma}}}_T+F)=V_t^{\bar{\gamma}}(F),
 \end{align*}
 for some $\bar{\gamma}\in\Gamma$, determining both the optimal strategy $\bar{\pi}^{\bar{\gamma}}$ and the optimal evaluation $V_t(F)$ for the star-shaped optimization problem.
Consequently, we discover that the optimal value function for the star-shaped risk measure $\rho_t$ is indeed the minimum of the optimal value functions corresponding to the convex risk measures $\rho_t^{\gamma}$ appearing in the representation of $\rho_t$ in the sense of Equation \eqref{minsub}.
In particular, $(\bar{\gamma},\bar{\pi})\in\Gamma\times\mathcal{K}$ is a minimizer for the star-shaped problem~\eqref{eq:portfoliochoice} if and only if $\bar{\pi}\in\mathcal{K}$ is a minimizer for the problem~\eqref{eq:portfoliochoice1}.

The convex optimization problem described in Equation~\eqref{eq:convopt} is highly versatile.
It encompasses various convex risk measures, ranging from monetary risk measures to risk measures derived from power, logarithmic, or negative exponential utilities.
To the best of our knowledge, problem~\eqref{eq:convopt} has not been comprehensively addressed in the dynamic case, making it a potential subject for future research.
However, when we assume that the star-shaped risk measure $\rho_t$ satisfies cash-additivity, Corollary \ref{corCSA} establishes that the family $(\rho_t^{\gamma})_{\gamma\in\Gamma}$ is also cash-additive. This allows us to approach the convex problem \eqref{eq:convopt} in a similar manner as demonstrated in \cite{LS14}.
We can formalize these observations in the following proposition, focusing on the case of star-shaped and cash-additive risk measures.
\begin{proposition}
In the present context, consider a monetary star-shaped risk measure $\rho_t: L^{\infty}(\mathcal{F}_T) \to L^{\infty}(\mathcal{F}_t)$ induced by the BSDE \eqref{eq: star-shapedBSDE}.
Suppose the driver $g$ of the BSDE satisfies  $L^{\infty}$-standard assumptions, $g(t,0)\equiv0$, $g \geq 0$, and does not depend on $y$. Let $(\rho_t^{\gamma})_{\gamma\in\Gamma}$ be the family of convex risk measures representing $\rho_t$ in the sense of Corollary \ref{corCSA}.
In this setting, the star-shaped optimization problem \eqref{eq:portfoliochoice} has a minimizer $(\bar{\gamma}, \bar{\pi}) \in \Gamma \times \mathcal{K}$, and the corresponding optimal evaluation is given by:
$$V_t(F)=X^{\bar{\pi}}_t+Y^{\bar{\gamma}}_t,$$
where $Y_t^{\bar{\gamma}}$ represents the first component of the solution to the following BSDE:
$$Y^{\bar{\gamma}}_t=F+\int_t^T\tilde{g}(s,Z^{\bar{\gamma}}_s)ds-\int_t^TZ_s^{\bar{\gamma}}dW_{s},$$
with $$\tilde{g}(t,Z^{\bar{\gamma}}_t):=\bar{\pi}_tb_t+g^{\bar{\gamma}}(t,Z^{\bar{\gamma}}_t+\bar{\pi}_s\sigma_s),$$ where $g^{\bar{\gamma}}$ is the driver that induces the dynamics of the risk measure $\rho_t^{\bar{\gamma}}$ in the family $(\rho_t^{\gamma})_{\gamma \in \Gamma}$.
\label{prop:linear}
\end{proposition}

\begin{remark}
We can also consider a risk-averse investor whose preferences are characterized via a star-shaped and cash-additive risk measure $\rho_t$, as defined in Equation~\eqref{eq: star-shapedBSDE}:
 $$\tilde{\rho}_t(X):=\rho_t(u(X))=\essmin_{\gamma\in\Gamma}\rho^{\gamma}_t(u(X)),$$
where $u$ represents an exponential, logarithmic, or power  (in the latter case, we assume that the risk measure $\rho_t$ is positively homogeneous) (dis)utility function.
Under these assumptions, by Equations \eqref{eq:dualSA} and \eqref{eq:portfoliochoice}, we can express the star-shaped optimization problem as:
    \begin{equation}
    V_t(F)=\essmin_{\gamma\in\Gamma}\essmin_{\pi\in\mathcal{K}}\essmax_{\mu_t\in\mathcal{A}}\mathbb{E}_{\mathbb{Q}^{\mu}_t}\left[\left.u(X^{\pi}_T+F)-\int_t^TG^{\gamma}(s,\mu_s)ds\right|\mathcal{F}_t\right].
\label{eq:SSU}
    \end{equation}
    The corresponding convex optimization problem becomes:
 $$V^{\gamma}_t(F)=\essmin_{\pi\in\mathcal{K}}\essmax_{\mu_t\in\mathcal{A}}\mathbb{E}_{\mathbb{Q}^{\mu}_t}\left[\left.u(X^{\pi}_T+F)-\int_t^TG^{\gamma}(s,\mu_s)ds\right|\mathcal{F}_t\right],$$
    for each $\gamma\in\Gamma$.
These convex optimization problems can be approached through BSDEs using Theorems 4.5, 4.8, and 4.11 in \cite{LS14} for exponential, logarithmic, or power utilities, respectively.
The minimum over $\gamma\in\Gamma$ then yields the solution to the star-shaped portfolio choice problem provided in Equation~\eqref{eq:SSU}, in the spirit of Proposition~\ref{prop:linear}.

It is worth noting that the risk measure $\tilde{\rho}_t$ is time-consistent, as it is induced by BSDEs.
Indeed, the composition $\tilde{\rho}_t=\rho_t\circ u$ follows the dynamics given by:
$$\tilde{\rho}_t(X)=\rho_t(u(X))=u(X)+\int_t^Tg(s,\rho_s(u(X)),Z_s)ds-\int_t^TZ_sdW_s.$$
Therefore, the flow property of BSDEs (see Proposition~2.5 in \cite{EPQ97}) ensures the time-consistency of $\tilde{\rho}_t$.
This demonstrates the time-consistency property, indicating that the risk measure remains consistent with the investor's preferences over time.
\begin{example}
In this example, we consider a risk measure that exhibits robustness with respect to the probabilistic model as well as the utility function.
Imagine the investor considers a set of (dis)utility functions $(u^{\gamma})_{\gamma\in\Gamma}$ rather than a single function to assess the risk associated with the same asset.
A conservative, preference robust approach would be to employ the minimum risk measure generated by the set of utility functions.
Then, the investor may consider a risk measure of the following form:
$$\rho_t(X)=\essmin_{\gamma\in\Gamma}\rho^{\gamma}_t(X),\qquad\mathrm{where}\qquad
\rho_t^{\gamma}(X)=\esssup_{\mathbb{Q}\in\mathcal{Q}}\mathbb{E}_{\mathbb{Q}}[u^{\gamma}(X)-c^{\gamma}_t(\mathbb{Q})|\mathcal{F}_t].$$
Here, $(u^{\gamma})_{\gamma\in\Gamma}$ represents a family of convex and increasing (dis)utility functions,
$\mathcal{Q}$ denotes a suitable set of probability measures that are absolutely continuous with respect to the reference probability measure $\mathbb{P}$, and $c^{\gamma}_t$ represents a penalty function that quantifies the plausibility level of a model $\mathbb{Q}\in\mathcal{Q}$.
Proposition~\ref{PM} ensures that $\rho_{t}$ is a dynamic star-shaped risk measure.
However, it is important to note that this risk measure may not be time-consistent.
The portfolio choice problem can be formulated as follows:
$$V_t(F)=\essmin_{\gamma\in\Gamma}\essmin_{\pi\in\mathcal{K}}\esssup_{\mathbb{Q}\in\mathcal{Q}}\mathbb{E}_{\mathbb{Q}}[u^{\gamma}(X_T^{\pi}+F)-c^{\gamma}_t(\mathbb{Q})|\mathcal{F}_t].$$
When the family $(u^{\gamma})_{\gamma\in\Gamma}$ consists of exponential, logarithmic and power utility functions, and $c^{\gamma}$ is an appropriate penalty function (as defined in \cite{LS14}), we can solve the convex optimization problem $\gamma$-by-$\gamma$ employing the techniques in \cite{LS14}.
It is important to realize that, in this setting, the resulting risk measure $\rho_t$ is cash-subadditive, but it may not satisfy cash-additivity.
By minimizing over $\gamma\in\Gamma$, we obtain the solution to the original star-shaped optimization problem.
This highlights the generality of our approach, as in this case we can effectively handle star-shaped risk measures that are not necessarily induced via BSDEs.
\end{example}
\end{remark}

\newpage

\appendix
\renewcommand\thesection{\Alph{section}}

\section{Appendix: Proofs}
\label{sec:appendix}
\renewcommand{\thesubsection}{\thesection.\Roman{subsection}}
\renewcommand{\theequation}{\thesection.\Roman{equation}}
\renewcommand{\thetheorem}{\thesubsection.\arabic{theorem}}
\renewcommand{\thetheorem}{\thesection.\Roman{theorem}}
\setcounter{equation}{0}
\renewcommand{\thesection}{\Alph{section}}
\renewcommand{\thesubsection}{\thesection.\arabic{subsection}}
\renewcommand{\thetheorem}{\thesubsection.\arabic{theorem}}
\setcounter{theorem}{0}
\subsection{Proofs of Section~\ref{sec:prel}}
\begin{lemma}
     Let $\rho_t: L^{\infty}(\mathcal{F}_T)\to L^{\infty}(\mathcal{F}_t)$ be a cash-subadditive risk measure.
     Then $\rho_t$ is regular.
    \label{REGCS}
\end{lemma}
\begin{proof}

    Let us fix $X\in L^{\infty}(\mathcal{F}_T)$.
   \begin{itemize}
   \item
   If $0\leq\essinf{X}\leq\esssup{X}$ then, for any $A\in\mathcal{F}_t$,
    \begin{align*}
       \mathbb{I}_{A}\rho_t(X)&=\mathbb{I}_{A}\rho_t(\mathbb{I}_{A}X+\mathbb{I}_{A^c}X) \leq \mathbb{I}_{A}\rho_t(\mathbb{I}_{A}X+\mathbb{I}_{A^c}\esssup{X}) \leq \mathbb{I}_{A}\rho_t(\mathbb{I}_{A}X)\leq\mathbb{I}_{A}\rho_t(X),
\end{align*}
where the first and last inequalities are due to monotonicity and positivity of $X$ (indeed if $X\geq0$ then $\mathbb{I}_{A}X\leq X$).
The second inequality follows from cash-subadditivity for positive $\mathcal{F}_t$-measurable random variables (given that $\esssup{X}\geq0$).
\item If $\essinf{X}\leq\esssup{X}\leq0$ the proof can be established similarly as above.
\item If $\essinf{X}\leq0\leq\esssup{X}$ then, for any $A\in\mathcal{F}_t$,
\begin{align*}
\mathbb{I}_{A}\rho_t(X)&=\mathbb{I}_{A}\rho_t(\mathbb{I}_{A}X+\mathbb{I}_{A^c}X)\leq\mathbb{I}_{A}\rho_t(\mathbb{I}_{A}X+\mathbb{I}_{A^c}\esssup{X}) \\
&\leq \mathbb{I}_{A}\rho_t(\mathbb{I}_{A}X)\leq \mathbb{I}_{A}\rho_t(\mathbb{I}_{A}X+\mathbb{I}_{A^c}\essinf{X})\leq \mathbb{I}_{A}\rho_t(X),
\end{align*}
where the first inequality is due to monotonicity, the second inequality follows from cash-subadditivity for positive $\mathcal{F}_t$-measurable random variables, while the third inequality is due to cash-subadditivity for negative $\mathcal{F}_t$-measurable random variables.
\end{itemize}
This completes the proof.
\end{proof}

\setcounter{theorem}{0}

\subsection{Proofs of Section~\ref{sec:con}}
\begin{proof}[Proof of Proposition \ref{DynamicReturn}]
        \textit{1.}
        Obvious by positive homogeneity. \\
         \textit{2.}
         We recall that regularity of $\rho_t$ follows from convexity and normalization (cf.\ Proposition 2 in \cite{DS05}), thus its dual representation can be obtained by following the proof of Proposition~6 in \cite{MGRO15}, with some necessary modifications in the sign convention.
         It can be verified that for any $X\in L^{\infty}(\mathcal{F}_T)$ and $t\in[0,T]$, the static return risk measure  $$\rho_{0,t}(X):=\mathbb{E}_{\mathbb{P}}[\rho_t(X)],$$ is convex, positively homogeneous, continuous from below and normalized at $1$, allowing its dual representation as in \textit{Step 1} of the proof of Proposition~6 in \cite{MGRO15}:
         $$\rho_{0,t}(X)=\sup_{{\mu\in\mathcal{S}}}\mathbb{E}_{\mu}[X], \ \ \forall X\in L^{\infty}(\mathcal{F}_T),$$
         where the penalty term vanishes due to positive homogeneity and $$\mathcal{S}:=\{\mu:(\Omega,\mathcal{F}_T)\to[0,1] \mbox{ measures s.t. } \mu\ll\mathbb{P} \mbox{ with } \mu(\Omega)\leq 1\}.$$
         Indeed, by Corollary~7 of \cite{FR02}, items (i) and (ii), and by observing that positive-constancy yields (using the same notation as in the aforementioned corollary and defining $\mathcal{X}\equiv L^{\infty}(\Omega,\mathcal{F},\mathbb{P})$):
            $$\rho_{0,t}(\eta)=\sup_{X'\in \mathcal{X}'_+}\left\{\eta X'(1)-\rho_{0,t}^*(X')\right\}=\eta \ \ \forall \eta\geq 0,$$ which implies:
            $$\eta(X'(1)-1)\leq \rho_{0,t}^*(X') \ \ \forall X'\in X',  \ \forall \eta\geq0, \ \rho_{0,t}^*(X')<+\infty,$$
            and the last inequality yields $X'(1)\leq 1$.
            We have that $\mathcal{X}'_+\cap \{X'\in \mathcal{X}'|X'(1)\leq 1\}=Ba_+(\Omega,\mathcal{F},\mathbb{P})\cap \{X'\in \mathcal{X}'|X'(1)\leq 1\}$, which is the space of all finite, finitely additive, positive and absolutely continuous measures w.r.t.\ $\mathbb{P}$ such that $\mu(\Omega)\leq 1$.
            The continuity from below of $\rho_t$ guarantees that $\mu$ also satisfies $\sigma$-additivity (cf.\ \cite{ELKR09} for more details).
            Moreover the penalty term vanishes according to item~(vi) of Corollary~7 in \cite{FR02}.
         The last part of the proof follows verbatim from \textit{Step 2} and \textit{Step 3} in the proof of Proposition~6 in \cite{MGRO15}.\\
         \textit{3.} Cash-subadditivity is clear from the dual representation of $\rho_t$.
        \end{proof}
         \begin{proof}[Proof of Corollary \ref{DynamicReturn1}]
 We only have to check that the restriction of $\rho_t$ in Proposition~\ref{DynamicReturn} maps positive $\mathcal{F}_T$-measurable random variables into positive $\mathcal{F}_t$-measurable random variables.
 This is clear from the dual representation in Proposition~\ref{DynamicReturn} and the positivity of the conditional expectation operator.
 \end{proof}

\begin{lemma}
Suppose $g : \Omega \times [0,T] \times \mathbb{R} \times \mathbb{R}^n \to \mathbb{R}$ satisfies the $L^{\infty}$-standard assumptions, is positively homogeneous with respect to $(y,z)$, and is non-negative.
Under these conditions, let $\rho_t(X)$ be the first component of the solution to the BSDE \eqref{eq:bsde} for any $X\in L^{\infty}_{+}(\mathcal{F}_T)$.
Then, $\rho_t$ is a map from $L^{\infty}_{+}(\mathcal{F}_T)$ into $L^{\infty}_{+}(\mathcal{F}_t)$ and is a return risk measure for any $X\in L^{\infty}_{+}(\mathcal{F}_T)$.
Moreover, if $g(\cdot, 1,0)\equiv 0$, then $\rho_t({1})=1$, i.e., the return risk measure is normalized at ${1}$.
\end{lemma}
\begin{proof}
   The driver $g$ being positively homogeneous yields the corresponding property of $\rho_t(X)$ for any $X\in L^{\infty}_{+}(\mathcal{F}_T)$ (cf.\ \cite{RG06,J08}).
   Additionally, by the comparison theorem for BSDEs (cf.\ \cite{EPQ97}), the positivity of the driver $g$ ensures that for any positive random variable $X\in L^{\infty}_+(\mathcal{F}_T)$, $\rho_t(X)\geq 0$, and hence $\rho_t: L^{\infty}_{+}(\mathcal{F}_T) \to L^{\infty}_{+}(\mathcal{F}_T)$.
   Finally, we observe that $(1,0)$ is the unique solution of BSDE \eqref{eq:bsde} with terminal condition $X\equiv 1$ when $g(\cdot,1,0)\equiv0$, which implies that $\rho_t({1})=1$, i.e., the return risk measure is normalized at 1.
\end{proof}
\begin{remark}
    The normalization at $1$ and positive homogeneity of $g$ imply that $g(t,y,0)=0$ for all $y\geq0$.
    This in turn leads to the positive-constancy property of $\rho_t$: for any $c_t\in L^{\infty}_{+}(\mathcal{F}_t)$, we have
    $$\rho_t(c_t)=c_t+\int_t^Tg(s,\rho_s(c_t),Z_s)ds-\int_t^TZ_sdW_s,$$
    which has the unique solution $(\rho_t(c_t),Z_t)\equiv(c_t,0)$. \\
    More generally, if $\rho_t:L^{\infty}(\mathcal{F}_T)\to L^{\infty}(\mathcal{F}_t)$ is normalized at ${1}$ and positively homogeneous, then $\rho_t$ is also positive-constant.
    Indeed, for any $\alpha_t\in L^{\infty}_+(\mathcal{F}_t)$,
    $$\rho_t(\alpha_t)=\alpha_t\rho_t({1})=\alpha_t.$$
    \label{Rconstancy}
    \end{remark}
  \begin{proof}[Proof of Proposition \ref{PRP}]
                For each $X\in L^{\infty}_+(\mathcal{F}_T)$, the solution $\rho_t(X)$ to Equation \eqref{eq:bsde} inherits the properties of positivity, convexity, positive homogeneity and normalization at $1$ from the corresponding properties of the driver $g$.
                Thus, $\rho_t$ satisfies all the assumptions in Corollary~\ref{DynamicReturn1} (continuity from below is a general property for BSDEs solutions), yielding the desired conclusion.
            \end{proof}
            \begin{remark}
                    Note that the assumptions made in Proposition~\ref{PRP} imply that the driver $g$ is non-increasing w.r.t.\ $y$. Indeed, for any $(\omega,t,y,z)\in\Omega\times[0,T]\times\R\times\R^n$ and $c\geq 0$, positive homogeneity, normalization and convexity yield:
                $$g(t,y+c,z)=g(t,y+c,z+0)\leq g(t,y,z)+g(t,c,0)=g(t,y,z),$$
              where the first inequality follows from sublinearity and the second equality follows from normalization and positive homogeneity (see Remark~\ref{Rconstancy}).
              Therefore, we can conclude that $g$ is monotonic with respect to $y$. It is worth noting that this result is consistent with the thesis of Proposition~\ref{IFFCS}.
              \end{remark}

\setcounter{theorem}{0}

\subsection{Proofs of Section~\ref{sec:star}}

\begin{lemma}
    Let $f:\mathcal{X}\to\mathbb{R}\cup\{+\infty\}$ be a functional such that $f(0)<+\infty$. For any $Z\in dom(f)$, we define the functional $f_Z:\mathcal{X}\to\mathbb{R}\cup\{+\infty\}$ as:
    \begin{equation}
        f_Z(X):=
        \begin{cases}
            \alpha f(Z) +(1-\alpha)f(0) &\mbox{ if there exists } \alpha\in[0,1] \mbox{ s.t.\ } X=\alpha Z, \\
            +\infty &\mbox{ otherwise.}
        \end{cases}
        \label{convexf}
    \end{equation}
    Then $f_Z$ is a proper convex functional.
    Moreover, $f_Z$ is lower semicontinuous w.r.t.\ the topology $\tau$ on $\mathcal{X}$ and it is also order lower semicontinuous.
    \label{Lconvexf}
\end{lemma}

 \begin{proof}
Convexity can easily be verified by direct inspection.
The lower semicontinuity w.r.t.\ the topology $\tau$ of $\mathcal{X}$ follows from Remark~1 of \cite{C00}.
It remains to prove the order lower semicontinuity.
To this end, it is sufficient to consider sequences instead of nets thanks to the order separability of $\mathcal{X}$ (see Lemma~3 in \cite{BF10}).
Fix $Z\in dom(f)$ with $Z\neq 0$ (if $Z=0$ then $f_Z(X)=+\infty \ \forall X\in\mathcal{X}\setminus\{0\}$ and the thesis is obvious) and let $X_n\xrightarrow{o} X$.
We want to prove that $\liminf_n f_Z(X_n)\geq f_Z(X)$.
Let us suppose $f_Z(X)<+\infty$; if $\liminf_n f_Z(X_n)=+\infty$ there is nothing to prove.
Conversely, if $\liminf_nf_Z(X_n)$ is finite we can consider only (sub-)sequences $X_n$ of the form $X_n=\alpha_n Z$ for some $\alpha_n\in[0,1]$.
We first prove that $\alpha_n$ is a Cauchy sequence in $[0,1]$, hence it converges to some $\bar{\alpha}\in[0,1]$.
Indeed, fixing $\varepsilon>0$ and $\omega\in\Omega$  s.t.\ $Z(\omega)\neq0$  we have: $$|\alpha_n-\alpha_m||Z(\omega)|=|(\alpha_n-\alpha_m)Z(\omega)|\leq |\alpha_nZ(\omega)-X(\omega)|+|\alpha_mZ(\omega)-X(\omega)|\leq \varepsilon'$$ for $n,m$ large enough, where the last inequality follows from the order convergence w.r.t.\ the pointwise order relation between random variables.
Hence, for $n,m$ large enough $|\alpha_n-\alpha_m|\leq \varepsilon$ (we can choose $\varepsilon'$ s.t.\ $\varepsilon'/|Z(\omega)|\leq \varepsilon$).
We now prove that $X_n\xrightarrow{o}\bar{\alpha}Z$.
Indeed, $|\alpha_nZ-\bar{\alpha}Z|=|\alpha_n-\bar{\alpha}||Z|\leq c_n|Z|=:Y_n$ with $c_n\downarrow 0$ (take $c_n=\sup_{k\geq n}|\alpha_k-\bar{\alpha}|$), which yields $Y_n\downarrow 0$ and by uniqueness of the order limit we can conclude that $X=\bar{\alpha}Z$.
So, any order convergent sequence of the form $X_n=\alpha_n Z$ order converges to $\bar{\alpha}Z$, where $\bar{\alpha}=\lim_{n\to\infty}\alpha_n$.
In conclusion, when $\liminf_n f_Z(X_n)<+\infty$, we have:
     $$\liminf_n f_Z(X_n)=\lim_n(\alpha_n f(Z)+(1-\alpha_n)f(0))=\bar{\alpha}f(Z)+(1-\bar{\alpha})f(0)=f_Z(X).$$
     If $f_Z(X)=+\infty$, we can suppose by contradiction that $\liminf_nf_Z(X_n)<+\infty$.
     As before we can consider $X_n=\alpha_n Z$ with $\alpha_n\in[0,1] \ \forall n\in\mathbb{N}$, concluding that $X_n\xrightarrow{o} X=\bar{\alpha}Z$ with $\bar{\alpha}=\lim_n\alpha_n$, so $\liminf_nf_Z(X_n)=\alpha_nf(Z)+(1-\alpha_n)f(0)\to\bar{\alpha}f(Z)+(1-\bar{\alpha})=f_Z(X)<+\infty$, which is a contradiction.
    \end{proof}
\begin{lemma}
    A functional $f:\mathcal{X}\to\mathbb{R}\cup\{+\infty\}$ with $f(0)<+\infty$ is star-shaped w.r.t.\ its value at $0$ if and only if, for any $\lambda\geq 1$, it holds that $f(\lambda X)\geq \lambda f(X)+(1-\lambda)f(0)$.
    \label{Lemma:SS}
\end{lemma}
    \begin{proof}
Let us assume that $f$ is star-shaped and fix $\lambda\geq 1$.
Thus, we have:
\begin{align*}
f(X) = f\Big(\frac{1}{\lambda}\lambda X\Big)\leq \frac{1}{\lambda}f(\lambda X) + \Big(1-\frac{1}{\lambda}\Big)f(0),
\end{align*}
where the inequality follows from the star-shapedness of $f$.
Rearranging terms, we obtain the desired thesis.
The converse implication can be proven through similar considerations.
\end{proof}
\begin{proof}[Proof of Lemma~\ref{Lemma:infSS}]
The fact that $f(0)=c$ is evident upon direct inspection.
For any $\lambda\geq1$, star-shapedness follows by the following chain of inequalities:
\begin{align*}
f(\lambda X)&=\inf_{\gamma\in\Gamma}f_{\gamma}(\lambda X)\geq \inf_{\gamma\in\Gamma}\left\{\lambda f_{\gamma}(X)+(1-\lambda)f(0)\right\} \\
&=\lambda\left(\inf_{\gamma\in\Gamma}f_{\gamma}(X)\right)+(1-\lambda)f(0)=\lambda f(X)+(1-\lambda)f(0),
\end{align*}
where the first inequality is due to star-shapdness of the family $(f_{\gamma})_{\gamma\in\Gamma}$ and Lemma~\ref{Lemma:SS}.
\end{proof}
     \begin{proof}[Proof of Proposition~\ref{staticSS}]
        By Lemma~\ref{Lconvexf}, $f_Z$ is a proper, convex and lower semicontinuous functional.
        We only need to verify the minimum condition.
        Fixing $Z\in dom(f)$, by star-shapedness of $f$, we have that for any fixed $X\in\mathcal{X}$ s.t.\ $X=\alpha Z$ with $\alpha\in[0,1]$,  $f(X)=f(\alpha Z)\leq\alpha f(Z)+(1-\alpha)f(0)=f_Z(X)$.
        Conversely, if $X\in\mathcal{X}$ verifies $X\neq\alpha Z$ for any $\alpha\in[0,1]$, then $f(X)\leq f_Z(X)=+\infty$.
        Hence, for any $Z$ in the proper domain of $f$ we have $f(X)\leq f_Z(X) \ \forall X\in\mathcal{X}$.
        Moreover, for each fixed $X\in dom(f)$ we can choose $Z=X$.
        Then $f_X(X)=f(X)$ (if $X\not\in dom(f)$ it is obvious by the previous point that $f_Z(X)\geq f(X)=+\infty$, which yields $f(X)=f_Z(X)=+\infty$ for any $Z\in dom(f)$).
        Thus, we have $f(X)=\ds\min_{Z\in dom(f)}f_Z(X) \ \forall X\in\mathcal{X}$.
        The min-max representation \eqref{minmaxstatic} follows from Proposition~1 of \cite{BF10}.
        Indeed, given that $\sigma(\mathcal{X},\mathcal{X}^*)$ has the C-property and for any $Z\in dom(f)$ $f_Z$ is a proper, convex and order lower semicontinuous function, it can be represented as: $$f_Z(X)=\sup_{q\in\mathcal{X}^*}\{\langle q,X\rangle-f^*_Z(q)\} \ \forall X\in\mathcal{X}.$$

        Now let us consider a monotone and star-shaped functional $f:\mathcal{X}\to\mathbb{R}$.
        We can build the family $(\tilde{f}_Z)_{Z\in dom(f)}$ as follows: for each $Z\in dom(f)$ we define $\tilde{f}_Z(X):=\inf\left\{f_Z(Y):Y\geq X\right\}$.
        Clearly, $\tilde{f}_Z$ is monotone by definition.
        We need to prove that it is also convex, proper, $\tilde{f}_Z(X)\geq f(X) \ \forall X\in\mathcal{X}$ and $\tilde{f}_Z(0)=f(0) \ \forall Z\in dom(f)$.
        We first observe that $\tilde{f}_Z(X)\leq f_Z(X)$ for any $X\in\mathcal{X}$ by definition, thus $\tilde{f}_Z\not\equiv +\infty$.
        Moreover, by monotonicity of $f$, for any $Y\geq X$ we have that $f_Z(Y)\geq f(Y) \geq f(X)$.
        Taking the infimum over $Y\geq X$ we get $\tilde{f}_Z(X)\geq f(X)$ and properness follows.
        Moreover, given that $f\leq \tilde{f}_Z \leq f_Z$ it follows that $f(X)\leq \tilde{f}_X(X) \leq f_X(X)=f(X)$ for any $X\in dom(f)$, so that the minimum is attained at $Z=X$, and analogously we can prove that $f_Z(0)=f(0) \ \forall Z\in dom(f)$.
        Now we prove that $\tilde{f}_Z$ is convex for any $Z\in dom(f)$.
        Let us fix $Z\in dom(f)$.
        For any $X_1,X_2\in\mathcal{X},\lambda\in(0,1)$ and choosing $Y_1\geq X_1, Y_2\geq X_2$ with $Y_1,Y_2\in\mathcal{X}$, we have that $\lambda X_1+(1-\lambda)X_2\leq \lambda Y_1+(1-\lambda)Y_2$, hence:
        \begin{eqnarray*}
            \tilde{f}_Z(\lambda X_1+(1-\lambda)X_2)&=&\inf\left\{f_Z(Y) : Y\geq \lambda X_1+(1-\lambda)X_2\right\} \\ &\leq& f_Z(\lambda Y_1+(1-\lambda)Y_2)\leq \lambda f_Z(Y_1)+(1-\lambda)f_Z(Y_2),
    \end{eqnarray*}
        where the last inequality is due to convexity of $f_Z$.
        Taking the infimum over $Y\geq X_1$ and then over $Y\geq X_2$ we obtain:
        $$\tilde{f}_Z(\lambda X_1+(1-\lambda)X_2)\leq \lambda\tilde{f}_Z(X_1)+(1-\lambda)\tilde{f}_Z(X_2).$$
        Summing up, we have proved that if $f$ is star-shaped and monotone it can be represented as: $$f(X)=\min_{Z\in dom(f)}\tilde{f}_Z(X) \ \ \forall X\in\mathcal{X},$$ where $\tilde{f}_Z$ is a proper and convex functional such that $\tilde{f}_Z(0)=f(0).$

        In order to obtain the representation in Equation~\eqref{minmaxstaticmon} we need to show that $\tilde{f}_Z$ is order lower semicontinuous w.r.t.\ the usual partial order relation between real random variables.
        According to Lemma~3 of \cite{BF10} it is enough to prove that $\tilde{f}_Z$ is continuous from below (i.e., if $X_n\uparrow X$ then $\tilde{f}_Z(X_n)\to \tilde{f}_Z(X)$ and sequences can be used instead of nets due to the order separability of $\mathcal{X}$).
        We can equivalently write $\tilde{f}_Z$ as follows, imposing w.l.o.g.\ $f(0)=0$:
        \begin{equation}
            \tilde{f}_Z(X)=
            \begin{cases}
            \ds\inf_{\alpha\in A} \alpha f(Z)  \ \ &\mbox{ if } A\neq\emptyset, A:=\{\alpha\in[0,1]:\exists Y \mbox{ s.t. } Y\geq X \mbox{ and } Y=\alpha Z\}, \\
            +\infty &\mbox{ otherwise. }
                \end{cases}
                \label{eqcfb}
        \end{equation}
        Let $X_n\uparrow X$.
        By monotonicity we have $\tilde{f}_Z(X_n)\leq \tilde{f}_Z(X) \ \forall n\in\mathbb{N}$, which implies $\limsup_{n}\tilde{f}_Z(X_n)\leq \tilde{f}_Z(X)$.
        So we only need to prove that $\liminf_{n}\tilde{f}_Z(X_n)\geq \tilde{f}_Z(X)$ to establish the thesis.
        Let us suppose that $\tilde{f}_Z(X)<+\infty$; if $\liminf_{n}\tilde{f}_Z(X_n)=+\infty$ the thesis is obvious.
        Conversely, if $\liminf_{n}\tilde{f}_Z(X_n)<+\infty$ there exist infinitely many $n\in\mathbb{N}$ such that $\tilde{f}_Z(X_n)<+\infty$.
        In particular, we can consider a sequence $(X_{n_k})_{k\in\mathbb{N}}$ such that $\tilde{f}_Z(X_{n_k})=\alpha_{n_k} f(Z)$ for some $\alpha_{n_k}\in[0,1]$, which verifies $\lim_k\tilde{f}_Z(X_{n_k})=\liminf_n\tilde{f}_Z(X_n)$ similarly as done in the proof of Lemma~\ref{Lconvexf}.
        By Equation~\eqref{eqcfb}, when $\tilde{f}_Z(X_{n_k})=\alpha_{n_k} f(Z)$ with $\alpha_{n_k}\in[0,1]$, there exists $Y_{n_k}\geq X_{n_k}$ s.t.\ $Y_{n_k}=\alpha_{n_k}Z$.
        Thus, we can write:
        $$X=\lim_nX_n=\liminf_kX_{n_k}\leq\liminf_kY_{n_k}=(\liminf_k\alpha_{n_k})Z=:\bar{\alpha}Z,$$
        where $\liminf_k\alpha_{n_k}:=\bar{\alpha}\in[0,1]$ given that $\alpha_{n_k}\in[0,1] \ \forall k\in\mathbb{N}.$ These observations yield:
        \begin{align*}
            &\tilde{f}_Z(X)\leq \tilde{f}_Z(\liminf_kY_{n_k})=\tilde{f}_Z(\bar{\alpha}Z)\leq f_Z(\bar{\alpha}Z)= \bar{\alpha}f(Z)\\ &=\liminf_{k}\alpha_{n_k}f(Z)=\liminf_k\tilde{f}_Z(X_{n_k})=\liminf_n\tilde{f}_Z(X_n).
        \end{align*}
        The first inequality is due to monotonicity of $\tilde{f}_Z$, the second inequality follows from $\tilde{f}_Z\leq f_Z$, whereas the second equality is due to the definition of $f_Z$ and the fact that $\bar{\alpha}\in[0,1]$.
        Summing up, if $\tilde{f}_Z(X)<+\infty$ and $X_n\uparrow X$, then $ \tilde{f}_Z(X_n)\to \tilde{f}_Z(X)$.
        If $\tilde{f}_Z(X)=+\infty$, let us suppose by contradiction that there exists a sequence $X_n\uparrow X$ such that $\liminf_n \tilde{f}_Z(X_n)<+\infty$.
        As before, we can find a (sub-)sequence $Y_n=\alpha_n Z$ with $\alpha_n\in[0,1]$ such that $Y_n\geq X_n$ and $X=\liminf_n X_n\leq \bar{\alpha}Z$, where $\bar{\alpha}:=\liminf_n\alpha_n\in[0,1]$.
        We can infer as above that $\tilde{f}_Z(X)\leq \liminf_n\tilde{f}_Z(X_n)<+\infty$, hence the contradiction.
        We conclude that $\tilde{f}_Z$ is continuous from below. Proposition~1 of \cite{BF10} yields for any $Z\in dom(f)$: $$\tilde{f}_Z(X)=\sup_{q\in\mathcal{X}^*_+}\{\langle q, X \rangle -\tilde{f}_Z^{*}(q)\} \ \ \forall X\in\mathcal{X},$$ thus we establish Equation~\eqref{minmaxstaticmon}.
        \end{proof}
         \begin{proof}[Proof of Corollary \ref{corCA}]
                 We define for any $Z\in dom(\rho)$ the functional:
                    \begin{equation*}
                        \rho_Z(X):=\begin{cases}
                            \alpha \rho(Z) + c \ \ &\mbox{ if } X \mbox{ is non-constant and there exist }\\&\qquad \alpha\in(0,1], c\in\R \mbox{ s.t. } X=\alpha Z + c, \\
                            c \ \ &\mbox{ if there exists } c\in\R \mbox{ s.t. } X=c, \\
                            +\infty &\mbox{ otherwise.}
                        \end{cases}
                    \end{equation*}
                Let us notice that if $Z$ is constant then the first case in the definition of $\rho_Z$ cannot happen.
                In this circumstance, either $X$ is constant and $\rho_Z(X)=c$ or $\rho_Z(X)=+\infty$.
                We have that $\rho_Z$ is cash-additive.
                Indeed, for any $m\in\R$, if $\rho_Z(X)<+\infty$ then $\rho_Z(X+m)=\alpha \rho(Z) + c +m =\rho_Z(X)+m$, and if $\rho_Z(X)=+\infty$ the same holds for $\rho_Z(X+m)$.
                Convexity is a routine verification.
                We need to prove that $\rho_Z(X)\geq \rho(X)$ for any $X\in\mathcal{X}$ and $\rho_X(X)=\rho(X)$ when $X\in dom(\rho)$. Let $X=\alpha Z + c$ for some $\alpha\in(0,1]$ and $c\in\mathbb{R}$ (the case when $X=c$ is obvious by cash-additivity of $\rho$).
                Then $\rho_Z(X)=\alpha \rho(Z) + c = \alpha \rho(Z+c/\alpha)\geq \rho(\alpha Z +c) = \rho(X)$, where we used cash-additivity and star-shapedness of $\rho$. Moreover, if $X\in dom(\rho)$ we can take $Z=X$, so $\rho_X(X)=\rho(X)$. Thus for any $X\in\mathcal{X}$, $\rho(X)=\min_{\gamma\in\Gamma}\rho_{\gamma}(X)$ with $\Gamma:= dom(\rho).$

                Now we prove order lower semicontinuity to obtain the min-max representation.
                Let $X_n\xrightarrow{o} X$.
                We start by considering the case $\rho_Z(X)<+\infty$.
                If $\liminf_n\rho_Z(X_n)=+\infty$ there is nothing to prove.
                So let us suppose that $\liminf_n\rho_Z(X_n)<+\infty$.
                As argued in Lemma~\ref{Lconvexf}, we can consider sequences of the form $X_n=\alpha_nZ+c_n$ where $\alpha_n\in[0,1],c_n\in\R \ \forall n\in\mathbb{N}.$ If we prove that $\alpha_n\to\bar{\alpha}\in[0,1]$ and $c_n\to\bar{c}\in\R$ with $X=\bar{\alpha}Z+\bar{c}$ then we have: $\liminf_n\{\alpha_n\rho(Z)+c_n\}=\bar{\alpha}\rho(Z)+\bar{c}=\rho_Z(X)$, hence the thesis is proved.
                If $Z$ is constant then also $X_n$ must be constant for any $n\in\mathbb{N}$ and by order convergence $X_n=c_n\to X$ a.s., so also $X$ is constant. Thus, cash-additivity of $\rho_Z$ implies $\rho_Z(X_n)=c_n\to X=\rho_Z(X)$. If $Z$ is not constant a.s., there exist $\omega_1,\omega_2\in\Omega$ such that $Z(\omega_1)\neq Z(\omega_2)$ and by order convergence we have $\alpha_nZ(\omega_i)+c_n\to \bar{Z}_{\omega_i}$ (eventually $\alpha_n$ can be equal to $0$ for some $n$) with $i=1,2$ and $\bar{Z}_{\omega_i}$ is a real number depending on $\omega_i$. Taking the difference of the two limits we have $\lim_{n\to\infty}(\alpha_nZ(\omega_1)-\alpha_n Z(\omega_2))=\bar{Z}_{\omega_1}-\bar{Z}_{\omega_2}$, hence:
                \begin{align*}
              &\lim_{n\to\infty}\alpha_n=\frac{\bar{Z}_{\omega_1}-\bar{Z}_{\omega_2}}{Z(\omega_1)-Z(\omega_2)}=:\bar{\alpha}\in[0,1], \\  &\lim_{n\to\infty}c_n=\bar{Z}_{\omega_1}-\bar{\alpha}Z(\omega_1)=:\bar{c}\in\mathbb{R}.
               \end{align*}
                Moreover $X_n\xrightarrow{o} \bar{\alpha}Z+\bar{c}$, indeed we have:
                \begin{eqnarray*}
                |X_n-(\bar{\alpha}Z+\bar{c})|&=&|(\alpha_n-\bar{\alpha})Z+(c_n-\bar{c})| \leq|\alpha_n-\bar{\alpha}||Z|+|c_n-\bar{c}|\\
                &\leq & \sup_{k\geq n}\{|\alpha_k-\bar{\alpha}|\}|Z|+\sup_{k\geq n}\{|c_k-\bar{c}|\}:=Y_n\downarrow 0,
                \end{eqnarray*}
                and by uniqueness of order limits we conclude $X=\bar{\alpha}Z+\bar{c}$, so the thesis follows.
                If $\rho_Z(X)=+\infty$ we can proceed by contradiction as done in Lemma~\ref{Lconvexf} to conclude that also $\liminf_n\rho_Z(X_n)=+\infty$.
                Thus $\rho_Z$ is order lower semicontinuous and the min-max representation holds.

                If in addition $\rho$ is monotone we can obtain monotonicity of $\rho_Z$ for any $Z\in dom(\rho)$ by defining $\tilde{\rho}_Z(X):=\inf\{\rho_Z(Y):Y\geq X\}$.
                In this case we still have cash-additivity, indeed:
                $$\tilde{\rho}_Z(X+m)=\inf\{\rho_Z(Y):Y\geq X+m\}=\inf\{\rho_Z(W+m):W\geq X\}=\tilde{\rho}_Z(X)+m,$$ where in the last equality we used cash-additivity of $\rho_Z$ and the definition of $\tilde{\rho}_Z$. Moreover, analogously as in Proposition~\ref{staticSS}, we have $\rho\leq\tilde{\rho}_Z\leq \rho_Z$ and $\tilde{\rho}_Z$ is still convex, thus the first thesis follows.

                As far as the min-max representation is concerned, we must prove that $\tilde{\rho}_Z$ is continuous from below.
                The proof is similar to the proof of Proposition~\ref{staticSS}.
                Once again we can write an equivalent expression for $\tilde{\rho}_Z(X)$. Defining for each $X\in\mathcal{X}$ and $Z\in dom(\rho)$ the set $\mathcal{A}^Z_{X}:=\{(\alpha,c)\in[0,1]\times\R:\exists Y\in\mathcal{X} \mbox{ s.t. } Y\geq X \mbox{ and } Y=\alpha Z+c\}$, it follows that:
                \begin{equation*}
                    \tilde{\rho}_Z(X)=\begin{cases}
                        \ds\inf_{(\alpha,c)\in\mathcal{A}_X^Z}\{\alpha \rho(Z)+c\} \ \ &\mbox{ if } \mathcal{A}^Z_{X}\mbox{ is non-empty}, \\
            +\infty &\mbox{ otherwise. }
                    \end{cases}
                \end{equation*}
                We observe that when $\mathcal{A}_X^Z$ is non-empty the infimum in the previous equation must be finite given that $\tilde{\rho}_Z(X)\geq \rho(X)>-\infty$, thus when $\mathcal{A}^Z_X$ is non-empty the infimum is in fact a minimum because there exist $\alpha\in[0,1]$ and $c\in\R$ such that $\tilde{\rho}_Z(X)=\alpha\rho(Z)+c$.
                Given $X_n\uparrow X$ we want to prove that $\lim_{n}\tilde{\rho}_Z(X_n)=\tilde{\rho}_Z(X)$. By monotonicity it is enough to prove that $\liminf_n\tilde{\rho}_Z(X_n)\geq \tilde{\rho}_Z(X).$ We suppose $\tilde{\rho}_Z(X)<+\infty$ and $\liminf_n\tilde{\rho}_Z(X_n)<+\infty$ (if $\liminf_n\tilde{\rho}_Z(X_n)=+\infty$ the thesis is clearly verified). Thus, there exists a (sub-)sequence $X_{n_k}$ such that $\tilde{\rho}_Z(X_{n_k})<+\infty \ \ \forall n\in\mathbb{N},$ hence $\tilde{\rho}_Z(X_{n_k})={\rho}_Z(Y_{n_k})=\alpha_{n_k}\rho(Z)+c_{n_k}$ for some $Y_{n_k}=\alpha_{n_k} Z+c_{n_k}$ with $\alpha_{n_k}\in[0,1],c_{n_k}\in\mathbb{R}$ with $Y_{n_k}\geq X_{n_k}$ and $\lim_{k}\tilde{\rho}_Z(X_{n_k})=\liminf_n\tilde{\rho}_Z(X_n).$
             Because $[0,1]$ is a compact set, we can extract a sub-sequence $\alpha_{n_{k_j}}$ such that $\alpha_{n_{k_j}}\to\bar{\alpha}\in[0,1]$. By monotonicity of $\tilde{\rho}_Z$, we infer that $\tilde{\rho}_Z(X_{n_{{k}}})$ must converge to $l\in\mathbb{R}$ (otherwise, we would have $\liminf_{k}\tilde{\rho}_Z(X_{n_k})=+\infty$, which contradicts the hypothesis).
             Thus, also $\tilde{\rho}_Z(X_{n_{k_j}})\to l$ and then $c_{n_{k_j}}=\tilde{\rho}_Z(X_{n_{k_j}})-\alpha_{n_{k_j}}\rho(Z)\to l-\bar{\alpha}\rho(Z):=\bar{c}\in\R$.
               So we have:
                $$X=\lim_{j\to\infty} X_{n_{k_j}}\leq \lim_{j\to\infty} Y_{n_{k_j}} =\lim_{j\to\infty} \{\alpha_{n_{k_j}}Z+c_{n_{k_j}}\}=\bar{\alpha}Z+\bar{c} \ \ d\mathbb{P}\mbox{-a.s.}$$
                This chain of inequalities and monotonicity of $\tilde{\rho}_Z$ yield:
                \begin{eqnarray*}
                    \tilde{\rho}_Z(X)&\stackrel{\mathrm{mon.}}\leq & \tilde{\rho}_Z(\lim_{j\to\infty} Y_{n_{k_j}})\stackrel{\mathrm{\tilde{\rho}_Z\leq\rho_Z}}\leq \rho_Z(\lim_{j\to\infty} Y_{n_{k_j}}) \\ &\stackrel{\mathrm{def.}}=&\bar{\alpha}\rho(Z)+\bar{c}= \lim_{j\to\infty} \{\alpha_{n_{k_j}}\rho(Z)+c_{n_{k_j}}\} \\
&=&\lim_{j\to\infty}\tilde{\rho}_Z(X_{n_{k_j}})=\liminf_{n\to\infty}\tilde{\rho}_Z(X_n).
                \end{eqnarray*}
                If $\tilde{\rho}_Z(X)=+\infty$ we can proceed by contradiction, proving that also $\liminf_n\tilde{\rho}_Z(X_n)=+\infty$, similarly as done in the proof of Proposition \ref{staticSS}.
                Hence, the min-max representation $$                    \rho(X)=\min_{\gamma\in\Gamma}\sup_{q\in\mathcal{C}\cap\mathcal{X}^*_+}\left\{\langle q,X\rangle-\rho^*_{\gamma}(q)\right\} \ \ \forall X\in\mathcal{X}
$$ holds, with $\Gamma:= dom(\rho)$.
                \end{proof}

        As a preparation for the dynamic results we provide a characterization of Lipschitz and star-shaped functionals as pointwise minimum of  Lipschitz and convex functions.
        This result non-trivially extends Theorem~2 in \cite{C00}, where only positively homogeneous functions have been considered.
        \begin{lemma}
            Let $(\mathcal{X},\|\cdot\|)$ be a normed vector space and $f:\mathcal{X}\to\mathbb{R}$ be a Lipschitz and star-shaped function and denote by $k>0$ its  Lipschitz constant.
            Then there exist a set of indexes $\Gamma$ and a set of equi-Lipschitz (with constant $k$) and convex functions labelled by $f_{\gamma}:\mathcal{X}\to\mathbb{R}$ with $\gamma\in\Gamma$ such that the following representation holds:
            \begin{equation*}
                f(X)=\min_{\gamma\in\Gamma}f_{\gamma}(X) \ \ \ \forall X\in\mathcal{X}.
            \end{equation*}
        Moreover, $f_{\gamma}(0)=f(0)$ for any $\gamma\in\Gamma$.
\label{Tminstarshaped}
        \end{lemma}
        \begin{proof}
            Let us fix $Z\in\mathcal{X}$.
            We define:
            $$f^Z(X):=\inf_{\bar{X}\in\mathcal{X}}\{k\|X-\bar{X}\|+f_Z(\bar{X})\}=\inf_{\alpha\in[0,1]}\{k\|X-\alpha Z\|+\alpha f(Z)+(1-\alpha)f(0)\},$$
            where $k$ is the Lipschitz constant of $f$ and $f_Z$ is given by Equation \eqref{convexf}.
            That is, $f^Z$ is the infimal convolution between $f_Z$ (convex and proper) and the Lipschitz and convex function $b_k(X):=k\|X\|$.
            Thus, by the properties of infimal convolution,\footnote{For a more comprehensive understanding of infimal convolution, interested readers may refer to \cite{R70,RW97,BEK05}.
            In general, the infimal convolution between $b_k$ and a function $h$ is commonly known in the literature as the \textit{Pasch-Hausdorff envelope} of $h$.} $f^Z$ is a proper, convex, and Lipschitz function.
            Moreover, we have that $f^Z(X)\geq f(X) \ \forall X\in\mathcal{X}$, indeed by Lipschitzianity and shar-shapedness of $f$ we have:
            $$f(X)-\alpha f(Z)-(1-\alpha)f(X)\leq f(X)-f(\alpha Z)\leq k\|X-\alpha Z\|,$$
            hence $f(X)\leq k\|X-\alpha Z\|+\alpha f(Z)+(1-\alpha)f(X)$ and taking the infimum over $\alpha\in[0,1]$ we obtain $f(X)\leq f^Z(X)$.
            Moreover, for any fixed $X\in\mathcal{X}$ we can choose $Z=X$ and $\alpha=1$, which yields $f^X(X)\leq f(X)$.
            The thesis follows by taking $\Gamma=\mathcal{X}$ and $f_{\gamma}=f^Z$.
            Moreover, by definition of $f_{\gamma}$ we have:
            $$f(0)\leq f_{\gamma}(0)=f^Z(0)=\inf_{\alpha\in[0,1]}\{k\|\alpha Z\|+\alpha\left(f(Z)-f(0)\right)\}+f(0)\leq f(0),$$
            where the first inequality is due to $f^Z(X)\geq f(X)$ for any $X\in\mathcal{X}$ and the last inequality follows by choosing $\alpha=0$, hence $f_{\gamma}(0)=f(0) \ \ \forall \gamma\in\Gamma$.
        \end{proof}
\begin{proof}[Proof of Lemma \ref{asSS}]
We prove the first statement.
If $\rho$ is a normalized star-shaped risk measure we have by monotonicity that $(A^m)_{m\in\R}$ is a monotone and increasing family of sets.
Moreover,  for any $\lambda\in[0,1]$, $Y\in\lambda A^{m}$ can be written as $Y=\lambda X$ with $X\in A^m$ and star-shapedness of $\rho$ yields:
$\rho(\lambda X)\leq \lambda \rho(X)\leq \lambda m$, so that $Y=\lambda X \in A^{\lambda m}$.
Right-continuity is given by the increasing property, which implies $A^m\subseteq \bigcap_{\bar{m}>m}A^{\bar{m}}$.
The converse inclusion follows by observing that if $X\in\bigcap_{\bar{m}>m}A^{\bar{m}}$ then $\rho(X)\leq \bar{m}$ for any $\bar{m}>m$, thus $\rho(X)\leq m$ and $X\in A^m$. Moreover, by normalization we have that $0\in\mathcal{A}_{\rho}^{m}$ if and only if $m\geq0$, thus $\inf\{m\in\R:0\in\mathcal{A}_{\rho}^m\}=0$.

Conversely, if $\mathcal{A}=(A^m)_{m\in\R}$ is a star-shaped family of acceptance sets we have that $\rho_{\mathcal{A}}$ is monotone due to the monotonicity of $(A^m)_{m\in\R}$.
Let $X\in\mathcal{A}^m$ for some $m\in\mathbb{R}$.
The star-shapedness of the acceptance sets yields:
\begin{align*}
\rho_{\mathcal{A}}(\lambda X)&=\inf\left\{m'\in\R:\lambda X\in A^{m'}\right\}\\&=\inf\left\{\lambda m'\in\R:\lambda X\in A^{\lambda m'}\right\} \\ &\leq\inf\left\{\lambda m'\in\R:\lambda X\in \lambda A^{m'}\right\}=\lambda\rho_{\mathcal{A}}(X).
\end{align*}
In addition, $\rho(0)=\inf\{m\in\R: 0\in A^{m}\}=0$, thus normalization follows.
The last thesis follows verbatim as in the proof of Theorem~1 in \cite{DK12}.
\end{proof}
\begin{proof}[Proof of Theorem \ref{th:asSS}]
$i)\implies ii)$.
This implication is proved in Proposition~\ref{staticSS}.  \\
$ii)\implies iii)$.
Let $A^m_{\gamma}=\left\{X\in\mathcal{X}:\rho_{\gamma}(X)\leq m\right\}$.
By convexity of $\rho_{\gamma}$ for any $\gamma\in\Gamma$ and by Proposition~1 in \cite{DK12}, $\mathcal{A}_{\gamma}$ is a convex family of acceptance sets.
Moreover, by item $ii)$ we have:
\begin{align*}
A^m_{\rho}&=\left\{X\in \mathcal{X}:\rho(X)\leq m\right\}=\left\{X\in\mathcal{X}:\rho_{\gamma}(X)\leq m \ \mbox{ for some } \gamma\in\Gamma\right\} \\
&=\{X\in\mathcal{X}:X\in \bigcup_{\gamma\in\Gamma}A^m_{\gamma}\}=\bigcup_{\gamma\in\Gamma}A^m_{\gamma}.
\end{align*}
Thus we can write by Lemma \ref{asSS}:
\begin{align*}
\rho(X)&=\inf\left\{m\in\R:X\in A^m_{\rho}\right\}=\inf\left\{m\in\R:X\in\bigcup_{\gamma\in\Gamma}A^m_{\gamma}\right\} \notag \\&=\inf\left\{m\in\R:X\in A^m_{\gamma} \mbox{ for some } \gamma\in\Gamma\right\}.
\end{align*}
In addition, normalization of $\rho_{\gamma}$ for any $\gamma\in\Gamma$ yields:
$$\rho_{\gamma}(0)=\inf\{m\in\mathbb{R}:0\in A^m_{\gamma}\}=0,$$
thus Equation \eqref{eq: normas} is satisfied by $(A_{\gamma})_{\gamma\in\Gamma}$. \\
$iii)\implies i)$.
We observe that for any $\lambda\in[0,1]$,
\begin{equation}
\bigcup_{\gamma\in\Gamma}\lambda A^m_{\gamma}\subseteq\bigcup_{\gamma\in\Gamma}A^{\lambda m}_{\gamma}.
\label{eq:inclusion}
\end{equation}
Indeed, if $X\in \bigcup_{\gamma\in\Gamma}\lambda A^m_{\gamma}$ then there exists $\bar{\gamma}\in\Gamma$ such that $X\in \lambda A^m_{\bar{\gamma}}$, thus by convexity of $A^m_{\bar{\gamma}}$ we have $X\in A^{\lambda m}_{\bar{\gamma}}$ (given that $0\in A^{0}_{\gamma} \ \forall\gamma\in\Gamma$), hence $X\in\bigcup_{\gamma\in\Gamma}A^{\lambda m}_{\gamma}$.

Now we prove star-shapedness of $\rho$.
Let us fix $\lambda\in[0,1]$ and $X\in\mathcal{X}$.
Then,
\begin{align*}
\rho(\lambda X)&=\inf\left\{m\in\R:\lambda X\in A^m_{\gamma} \mbox{ for some } \gamma\in\Gamma\right\}=\inf\{m\in\R:\lambda X\in \bigcup_{\gamma\in\Gamma}A^m_{\gamma}\}\\
&=\inf\{\lambda m'\in\R:\lambda X\in \bigcup_{\gamma\in\Gamma}A^{\lambda m'}_{\gamma}\}\leq\inf\{\lambda m':\lambda X \in
\bigcup_{\gamma\in\Gamma}\lambda A^{m'}_{\gamma}\}\\
&=\lambda\inf\{m'\in\R:X\in A^{m'}_{\gamma} \mbox{ for some } \gamma\in\Gamma\}=\lambda\rho(X),
\end{align*}
where the inequality is due to Equation~\eqref{eq:inclusion}. Furthermore we have:
$$\rho(0)=\inf\{m\in\R: 0\in A^{m}_{\gamma} \mbox{ for some } \gamma\in\Gamma\}=0,$$ by Equation \eqref{eq: normas}.
\end{proof}

\setcounter{theorem}{0}

\subsection{Proofs of Section~\ref{sec:mr}}
Before proving the main results, we first need a preliminary lemma whose proof is omitted since it can be derived from standard arguments in measure theory.
 \begin{lemma}
            Let $f:\Omega\times[0,T]\times\mathbb{R}\times\mathbb{R}^n\to\mathbb{R}$ be such that \begin{itemize}
                \item $f$ is an adapted stochastic process.
                \item  $d\mathbb{P}\times dt$-a.s.\ the function $(y,z)\mapsto f(\omega,t,y,z)$ is continuous.
                \item $f$ is star-shaped w.r.t.\ $(y,z)$, i.e., for any $(y,z)\in\mathbb{R}\times\mathbb{R}^n$ and $\alpha\in[0,1]$ we have $f(\omega,t,\alpha y,\alpha z)\leq \alpha f(\omega,t,y,z)+(1-\alpha)f(\omega,t,0,0) \ d\mathbb{P}\times dt$-a.s.
                \end{itemize}
                Under these conditions, it holds that: $$d\mathbb{P}\times dt\mbox{-a.s. } f(\omega,t,\alpha y,\alpha z)\leq \alpha f(\omega,t,y,z)+(1-\alpha)f(\omega,t,0,0) \mbox{ for any } (y,z)\in\mathbb{R}\times\mathbb{R}^n \mbox{ and }\alpha\in[0,1].$$
                 \label{LFMS}
        \end{lemma}
        \begin{proof}[Proof of Proposition \ref{POSS}]
    For brevity, we assume $\rho_t(0)\equiv 0$ and $g(t,0,0)\equiv0$ in the proof.
    However, it is important to note that the proof can be straightforwardly adapted to the general case.
    Let us fix $X\in L^2(\mathcal{F}_T)$ and suppose that $g(t,\alpha y,\alpha z)\geq \alpha g(t,y,z)$ for any $\alpha\geq1$ and $(y,z)\in\mathbb{R}\times\mathbb{R}^n \ d\mathbb{P}\times dt$-a.s.
    We define $\tilde{g}(t,y,z):=\alpha g(t,\frac{y}{\alpha},\frac{z}{\alpha})$ and let $(\tilde{\rho}_t(\alpha X),\tilde{Z}_t)_{t\in[0,T]}$ be the corresponding solution to the BSDE with driver $\tilde{g}$ and terminal condition $\alpha X$.
    Clearly, because $\tilde{g}$ satisfies the same hypotheses as $g$ (\textit{mutatis mutandis}), we have existence and uniqueness of this solution.
    By uniqueness of the solution we find $(\tilde{\rho}_t(\alpha X),\tilde{Z}_t)_{t\in[0,T]}=\alpha(\rho_t(X),Z_t)_{t\in[0,T]}$.
    Moreover, the star-shapedness of $g$ yields $g(t,y,z)\geq \tilde{g}(t,y,z) \ d\mathbb{P}\times dt$-a.s.\ and by the comparison theorem for BSDEs it holds that $\rho_t(\alpha X)\geq\tilde{\rho}_t(\alpha X)=\alpha\rho_t(X)$, so $\rho_t$ is star-shaped.

    Conversely, let us assume $\rho_t(\alpha X)\geq \alpha \rho_t(X) \ \forall \alpha\geq1,X\in L^2(\mathcal{F}_T)$.
    With the same notation as in Theorem~3.4 in \cite{J08}, let us define
    $$\mathcal{S}^{z}_{y}(g):=\left\{t\in[0,T]:g(t,y,z)=L^1-\lim_{\varepsilon\to0^+}\frac{1}{\varepsilon}[\rho_{t,t+\varepsilon}(y+z(W_{t+\varepsilon}-W_t))-y]\right\},$$
    where $\rho_{u,s}$ is the solution to BSDE \eqref{eq: star-shapedBSDE} on the interval $[u,s]$ with $u\leq s$. We
    fix $\alpha\geq 1$, $(y,z)\in\mathbb{R}\times\mathbb{R}^n$ and $t\in \mathcal{S}^y_z(g)\cap \mathcal{S}^{\alpha y}_{\alpha z}(g)$.
    Hence, star-shapedness of $\rho_t$ implies for any sufficiently small $\varepsilon>0$ the inequality $\alpha g_{\varepsilon}(t,y,z)\leq g_{\varepsilon}(t,\alpha y,\alpha z) \ d\mathbb{P}$-a.s., where
    \begin{align*}
    & g_{\varepsilon}(t,\alpha y,\alpha z)=\frac{1}{\varepsilon}\left(\rho_{t,t+\varepsilon}(\alpha(y+z(W_{t+\varepsilon}-W_t)))-\alpha y\right), \\
    &\alpha g_{\varepsilon}(t,y,z)=\alpha\frac{1}{\varepsilon}\left({\rho}_{t,t+\varepsilon}((y+z(W_{t+\varepsilon}-W_t)))-y\right).
    \end{align*}
    Lemma~2.1 in \cite{J08} ensures that (extracting a subsequence) $g_{\frac{1}{n}}(t,\alpha y,\alpha z)\to g(t,\alpha y,\alpha z)$ $ d\mathbb{P}$-a.s., as $n\to\infty$ and analogously for $\alpha g_{\frac{1}{n}}$.
    By combining the last convergences and the previous inequality we obtain $\alpha g(t,y,z)\leq g(t,\alpha y,\alpha z) \ d\mathbb{P}$-a.s.;
    Lemma 2.1 in \cite{J08} asserts also that $\mathcal{L}([0,T]\setminus \mathcal{S}^y_z(g)\cap \mathcal{S}^{\alpha y}_{\alpha z}(g))=0$, thus
    $\alpha g (t,y,z)\leq g(t,\alpha y,\alpha z) \ d\mathbb{P}\times dt$-a.s.; Lipschitzianity of $g$ w.r.t.\ $(y,z)$ implies continuity of $g$ w.r.t.\ $(y,z)$ and Lemma~\ref{LFMS} yields:
    $$d\mathbb{P}\times dt\mbox{-a.s. } \  g(t,\alpha y,\alpha z)\geq \alpha g(t,y,z) \ \ \ \forall (y,z)\in\mathbb{R}\times\mathbb{R}^n, \forall \alpha\geq1.$$
    The case of positive homogeneity follows similarly.
\end{proof}
  \begin{proof}[Proof of Theorem \ref{SuCO}]
             We divide the proof into several steps to make it simpler to comprehend.
             We start by proving the statements for star-shaped drivers.
             The positively homogeneous case is similar; we nevertheless provide a sketch of the proof in this case for the sake of clearness.
             Let us recall that when the driver $g$ is star-shaped (resp.\ positively homogeneous) then also $\rho_t$ inherits this property, according to Proposition~\ref{POSS}.
             {\begin{center}
\textit{The star-shaped case}:
\end{center}}
\noindent\textit{Step 1: approximation of the driver $g$.} \\ Given $\Omega'\subseteq\Omega$ with $\mathbb{P}(\Omega')=1$  and $t\in I\subseteq[0,T]$ with $\mathcal{L}(I)=1$, where $\mathcal{L}$ is the Lebesgue measure on the measurable space $([0,T],\mathcal{B}([0,T])$, s.t.\ $g$ is star-shaped and Lipschitz for all $\omega\in\Omega'$ and $t\in I$, we can define for each $(\omega,t)\in\Omega'\times I$:
\begin{equation*}
g_{\beta,\mu}(\omega,t,y,z):=
\begin{cases}
    m g(\omega,t,\beta,\mu)+(1-m)g(\omega,t,0,0) \ \ &\mbox{ if } \exists m\in[0,1] \mbox{ s.t. } (y,z)=m (\beta,\mu), \\
    +\infty \ &\mbox{ otherwise}.
\end{cases}
\end{equation*}
By Lemma~\ref{Lconvexf} we know that this is a lower semicontinuous, convex and proper function for any $(\omega,t)\in\Omega'\times I$.
Furthermore, we define:
\begin{equation*}
    g^{\beta,\mu}(\omega,t,y,z)=\inf_{\alpha\in[0,1]}\{k(|y-m\beta|+|z-m\mu|)+mg(\omega,t,\beta,\mu)+(1-m)g(\omega,t,0,0)\},
\end{equation*}
where $k$ is the Lipschitz constant of the driver $g$. By Lemma~\ref{Tminstarshaped} $g^{\beta,\mu}$ is a proper convex and Lipschitz function w.r.t.\ $(y,z)$, for any $(\omega,t)\in\Omega'\times I$, given that it is the Pasch-Hausdorff\footnote{Let us recall that according to Example 9.11 in \cite{RW97}, if the Pasch-Hausdorff envelope $g^{\beta,\mu}$ of $g_{\beta,\mu}$ is not equal to $-\infty$, then it is a Lipschitz continuous function. In our particular case, we have $g^{\beta,\mu} \not\equiv -\infty$.
To illustrate this, consider any $(\beta,\mu) \in \mathbb{R} \times \mathbb{R}^n$. We can choose $(y,z) = (\beta,\mu)$, thus $m = 1$ attains the infimum in the definition of $g^{\beta,\mu}(\omega,t,\beta,\mu)$. As a result, we obtain $g^{\beta,\mu}(\omega,t,\beta,\mu) = g(\omega,t,\beta,\mu) > -\infty$, indicating that $g^{\beta,\mu}$ is uniformly Lipschitz continuous.} envelope of a convex function (convexity follows from the fact that the infimal convolution between two convex functions is convex, provided that it is still proper).
Moreover, $g^{\beta,\mu}(\omega,t,0,0)=g(\omega,t,0,0)$ for any $(\beta,\mu)\in\mathbb{R}\times\mathbb{R}^n$ and $(\omega,t)\in\Omega'\times I$, according to the second thesis of Lemma~\ref{Tminstarshaped}.
Hence, $g^{\beta,\mu}(\cdot,\cdot,0,0)$ coincides $d\mathbb{P}\times dt$-a.s.\ with $g(\cdot,\cdot,0,0)$, which yields $g^{\beta,\mu}(\cdot,\cdot,0,0)\in L^2_{\mathcal{F}}(T,\mathbb{R})$ for any $(\beta,\mu)\in\mathbb{R}\times\mathbb{R}^n$.
Regarding the measurability of $g^{\beta,\mu}$, we can observe that the map
 $$m\mapsto k(|y-m\beta|+|z-m\mu|)+mg(\omega,t,\beta,\mu)+(1-m)g(\omega,t,0,0)$$ is continuous.
 Hence, it is sufficient to consider the infimum over a countable set $(m_n)_{n \in \mathbb{N}}$. Specifically, the infimum over a countable set of $\mathcal{P} \times \mathcal{B}(\mathbb{R}) \times \mathcal{B}(\mathbb{R}^n)$-measurable processes defined as $$g^{\alpha,\beta}_n(\omega,t,y,z):= k(|y-{m_n}\beta|+|z-m_n\mu|)+m_ng(\omega,t,\beta,\mu)+(1-m_n)g(\omega,t,0,0)$$ will preserve the $\mathcal{P} \times \mathcal{B}(\mathbb{R}) \times \mathcal{B}(\mathbb{R}^n)$-measurability of the resulting process.
 Furthermore, by Lemma~\ref{Tminstarshaped}, the Lipschitz and star-shapedness conditions on $g$ guarantee that for any $(\beta,\mu)$ we have $g^{\beta,\mu}(\omega,t,y,z)\geq g(\omega,t,y,z)$ $\forall (\omega,t)\in\Omega'\times I$ and $(y,z)\in\mathbb{R}\times\mathbb{R}^n$. Finally, if we take $(\beta,\mu)\equiv(y,z)$ and $m=1$, equality holds.
 So, for any $(\omega,t)\in\Omega'\times I$, we have the representation:
$$g(\omega,t,y,z)=\ds\min_{(\beta,\mu)\in\mathbb{R}\times\mathbb{R}^n}g^{\beta,\mu}(\omega,t,y,z) \ \ \ \forall (y,z)\in\mathbb{R}\times\mathbb{R}^n.$$
   We stress that also in this case the family of functions $\{g^{\beta,\mu}:(\beta,\mu)\in\mathbb{R}\times\mathbb{R}^n\}$ is equi-Lipschitz.
   Indeed for any $(\omega,t)\in\Omega'\times I$ the function $g^{\beta,\mu}$ is  Lipschitz with constant $k$, and $k$ does not depend on the choice of $(\beta,\mu)\in\mathbb{R}\times\mathbb{R}^n$, being the Lipschitz constant of the driver $g$.\medskip\\
\textit{Step 2: building drivers for convex risk measures.}\\ We define $$\Gamma:=\{(\alpha,\delta)_{t\in[0,T]} \mbox{ square integrable predictable stochastic processes valued in } \mathbb{R}\times\mathbb{R}^n \}$$ and for any $(\alpha,\delta)\in\Gamma$ and any $(\omega,t)\in\Omega'\times I$:
            \begin{equation*}
                g^{\alpha,\delta}(\omega,t,y,z):=g^{\alpha_t(\omega),\delta_t(\omega)}(\omega,t,y,z).
            \end{equation*}
           In particular, $g^{\alpha,\delta}$ is a predictable process, being the infimum over a countable set of predictable processes: $${g}^{\alpha,\delta}_n(\omega,t,y,z):=k(|y-m_n\alpha_t(\omega)|+|z-m_n\delta_t(\omega)|)+m_ng(\omega,t,\alpha_t(\omega),\delta_t(\omega))+(1-m_n)g(\omega,t,0,0).$$
           Moreover, $g^{\alpha,\delta}(\omega,t,0,0)=g(\omega,t,0,0)$ $\forall (\omega,t)\in\Omega'\times I$, which follows from the second thesis of Lemma~\ref{Tminstarshaped}, hence $g^{\alpha,\delta}(\cdot,\cdot,0,0)\in L^2_{\mathcal{F}}(T,\mathbb{R})$.
           In addition, it is Lipschitz w.r.t.\ $(y,z)$ thanks to the equi-Lipschitz property of the family $$\{g^{\beta,\mu}:(\beta,\mu)\in\mathbb{R}\times\mathbb{R}^n\},$$ which means we can choose the Lipschitz constant of $g^{\alpha,\delta}$ equal to $k$ $\forall (\omega,t)\in\Omega'\times I$.
           Finally, it is convex, given that for any $(\omega,t)\in\Omega'\times I$ we have that $g^{\alpha_t(\omega),\delta_t(\omega)}(\omega,t,y,z)$ is convex w.r.t $(y,z)$, according to $\textit{Step 1}$. Then the BSDE:
            \begin{equation*}
                \rho^{\alpha,\delta}_t (X)=X+\int_t^Tg^{\alpha,\delta}(s,\rho^{\alpha,\delta}_s,Z^{\alpha,\delta}_s)ds-\int_t^TZ_s^{\alpha,\delta}dW_s
            \end{equation*}
            admits a unique, monotone and convex solution (the convexity of $\rho^{\alpha,\delta}_t$ follows from the convexity of the driver, cf.\ \cite{RG06}). \medskip\\
\textit{Step 3: representing $\rho$ as a minimum of $\rho^{\alpha,\delta}$ over $\Gamma$.} \\
           Also in this case for any $(\alpha,\delta)\in\Gamma$ we have $g^{\alpha,\delta}(t,\rho_t,Z_t)\geq g(t,\rho_t,Z_t) \ d\mathbb{P}\times dt$-a.s.\ and by the comparison theorem (cf.\ \cite{EPQ97}) we conclude $\rho^{\alpha,\delta}_t(X)\geq\rho_t(X)$.
           Now let us consider the process $(\bar{\alpha},\bar{\delta})=(\rho_t,Z_t)_{t\in[0.T]}$.
           We have that $(\bar{\alpha},\bar{\delta})\in\Gamma$ (given that $\rho_t$ and $Z_t$ are predictable processes valued in $\mathbb{R}\times\mathbb{R}^n$) and it holds that
$$g^{\bar{\alpha},\bar{\delta}}(t,\rho_t,Z_t)=g(t,\rho_t,Z_t) \ d\mathbb{P}\times dt \mbox{-a.s.},$$
by the definition of $g^{\bar{\alpha},\bar{\delta}}$ and the results in \textit{Step 1}.\footnote{Indeed, for any fixed $(\omega,t)\in\Omega'\times I$ we can write: $g^{\bar{\alpha},\bar{\delta}}(\omega,t,\rho_t(\omega),Z_t(\omega))=g^{\rho_t(\omega),Z_t(\omega)}(\omega,t,\rho_t(\omega),Z_t(\omega))=g(\omega,t,\rho_t(\omega),Z_t(\omega))$, as proved in \textit{Step 1} or Lemma~\ref{Tminstarshaped}.
Given that $\mathbb{P}\otimes \mathcal{L}(\Omega'\times I)=0$, we have $g^{\bar{\alpha},\bar{\delta}}(t,\rho_t,Z_t)\equiv g(t,\rho_t,Z_t) \ \ d\mathbb{P}\times dt$-a.s.}
Then we conclude that $d\mathbb{P}$-a.s.\ and $\forall t\in[0,T]$ it holds that $\rho^{\bar{\alpha},\bar{\delta}}_t(X)=\rho_t(X)$ by uniqueness of the solution.
In other words, we have obtained the representation given in Equation~\eqref{minsub} for any fixed $X\in L^{\infty}(\mathcal{F}_T):$
\begin{equation*}
    \rho_t(X)=\essmin_{(\alpha,\delta)\in\Gamma}\rho_t^{\alpha,\delta}(X)=\rho_t^{\bar{\alpha},\bar{\delta}}(X) \ \ \forall t\in[0,T].
\end{equation*}
             {\begin{center} \textit{Sketch of the proof for the positively homogeneous case:}
             \end{center}}

Let us consider $(y,z),(\beta,\mu)\in\mathbb{S}_1\times\mathbb{S}_n\cup (0,0)$ where $\mathbb{S}_j$ is the unit sphere in $\mathbb{R}^j$.  Given $\Omega''\subseteq\Omega$ with $\mathbb{P}(\Omega'')=1$  and $t\in J\subseteq[0,T]$ with $\mathcal{L}(J)=1$, s.t.\ $g$ is positively homogeneous and Lipschitz for all $\omega\in\Omega''$ and $t\in J$, we can define for each $(\omega,t)\in\Omega''\times J$ the function
              \begin{equation*}
                  g_{\beta,\mu}(\omega,t,y,z):=
                  \begin{cases}
                      mg(\omega,t,\beta,\mu) \ &\mbox{if there exists } m\geq0 \mbox{ s.t. } (y,z)=m(\beta,\mu) \\
                      +\infty \ &\mbox{otherwise}.
                  \end{cases}
              \end{equation*}
             Keeping $(\beta,\mu)$ fixed, we evaluate the infimal convolution of $g_{\beta,\mu}$ and the sublinear and Lipschitz function $b_k(y,z):=k(|y|+|z|)$ where $k>0$ is the Lipschitz constant of the driver $g$.
             For all $(\omega,t)\in\Omega''\times J$ it holds that:
              \begin{align*}
                  g^{\beta,\mu}(\omega,t,y,z)&=\inf_{(\bar{y},\bar{z})\in\mathbb{R}\times\mathbb{R}^n}\{b_k(y-\bar{y},z-\bar{z})+g_{\beta,\mu}(\omega,t,\bar{y},\bar{z})\} \\
                  &=\inf_{m\geq0}\{k(|y-m\beta|+|z-m\mu|)+mg(\omega,t,\beta,\mu)\}.
              \end{align*}
              We observe that for any $(\omega,t)\in\Omega'\times I$ the function $g^{\beta,\mu}(\omega,t,\cdot,\cdot)$ is a sublinear---being the infimal convolution of two sublinear functions---and uniformly Lipschitz continuous function in $(y,z)$, hence $g^{\beta,\mu}(\omega,t,\cdot,\cdot)$ is real-valued.
                            Due to positive homogeneity we can extend the same argument for all $(y,z)\in\mathbb{R}\times\mathbb{R}^n$ (see Theorem~2 in \cite{C00} for a proof of this statement). In particular, according to Theorem~2 in \cite{C00}, for any $(y,z)\in\mathbb{R}\times\mathbb{R}^n$ the infimum is realized when $(\beta,\mu)=\frac{(y,z)}{|y|+|z|}$, with the convention that if $(y,z)=(0,0)$ the previous quantity is equal to $(0,0)$.
            Considering the set $$\Gamma:=\{(\alpha,\delta)_{t\in[0,T]} \mbox{ predictable stochastic process valued in } \mathbb{S}_1\times\mathbb{S}_n \cup (0,0) \ \  d\mathbb{P}\times dt \mbox{-a.s.}\},$$
            and the driver
            \begin{equation*}
             g^{\alpha,\delta}(\omega,t,y,z):=g^{\alpha_t(\omega),\delta_t(\omega)}(\omega,t,y,z),
           \end{equation*}
            the remaining part of proof can be derived similarly as in the star-shaped case, verifying that

           \begin{equation*}
                \rho^{\alpha,\delta}_t(X)=X+\int_t^Tg^{\alpha,\delta}(s,\rho^{\alpha,\delta}_s,Z^{\alpha,\delta}_s)ds-\int_t^TZ_s^{\alpha,\delta}dW_s
            \end{equation*}
            admits a unique, monotone and sublinear solution, satisfying
                      $\rho^{\alpha,\delta}_t(X)\geq\rho_t(X)$.            The equality holds when $(\bar{\alpha}_t,\bar{\delta}_t)_{t\in[0,T]}=\left(\frac{1}{|\rho_t|+|Z_t|}(\rho_t,Z_t)\right)_{t\in[0.T]}$, with the convention that if $(\rho_t(\omega),Z_t(\omega))=(0,0)$ then $(\bar{\alpha}_t(\omega),\bar{\delta}_t(\omega))=(0,0)$. \\
           \begin{center}
           {\textit{The min-max representations:}}
           \end{center}
The representation given in Equation \eqref{minmax} for star-shaped risk measures follows from Equation \eqref{eq:conrep}, observing that $g^{\gamma}$ satisfies standard assumptions and convexity.
Additionally, when $g$ is positively homogeneous, the family $(g^{\gamma})_{\gamma\in\Gamma}$ inherits the same property. As a result, $G^{\gamma}$ vanishes on its proper domain.
\end{proof}
\begin{remark}
    Let us note that the theses of Theorem~\ref{SuCO} still hold if we only suppose the process $g$ to be $\mathcal{F}_t$-adapted or $\mathcal{F}_t$-progressively measurable.
    In these cases, the solution to Equation~\eqref{eq: star-shapedBSDE} will be an adapted or progressively measurable process, respectively.
    Therefore, the proof of Theorem~\ref{SuCO} remains unchanged except for the choice of $\Gamma$, which becomes as follows for the positively homogeneous case (analogously for the star-shaped case): $$\Gamma:=\left\{(\alpha,\delta) \  \mathcal{F}_t\mbox {-adapted (resp.\ p.m.) square integrable processes valued in } \mathbb{S}_1\times\mathbb{S}_n\cup (0,0)\right\}.$$
   Furthermore, it is worth noting that it is not necessary to restrict attention to the space $L^{2}(\mathcal{F}_T)$ for the terminal condition $X$.
   In fact, it is possible to consider $X\in L^{\infty}(\mathcal{F}_T)$ and $g(\cdot,0,0)\in \mbox{BMO}(\mathbb{P})$. In this case, it follows that $g^{\gamma}(\cdot,0,0)\in\mbox{BMO}(\mathbb{P})$ for any $\gamma\in\Gamma$, and $\rho^{\gamma}_t\in L^{\infty}(\mathcal{F}_t)$ for any $t\in[0,T]$. Moreover, the min-max representation in Equation~\eqref{minmax} becomes:
    $$\rho_t(X)=\essmin_{\gamma\in\Gamma}\essmax_{(\beta,q)\in\mathcal{G}\times\mathcal{Q}^{\infty}}\mathbb{E}_{\mathbb{Q}^{q}_t}\left[D^{\beta}_{t}X-\int_t^TD^{\beta}_{t,s}G^{\gamma}(s,\beta_s,q_s)ds\Big|\mathcal{F}_t\right],$$
    with $(\beta_t,q_t)_{t\in[0,T]}$ bounded processes.
\end{remark}

Before presenting the proof of Corollary \ref{corCSA}, we provide a preliminary lemma on convex functions.
Although the proof of this lemma is not difficult, we were unable to find a reference for this result.

\begin{lemma}
Let $f:\mathbb{R}\to\mathbb{R}$ be convex and not monotone.
Then $f$ is coercive: $$\lim_{|x|\to+\infty}f(x)=+\infty,$$ in particular $f$ admits a global minimum. Moreover, the set $M:=\argmin_{y\in\mathbb{R}}{f(y)}$ is compact, thus $M$ admits a maximum and minimum.
\label{LemmaConvex}
    \end{lemma}
    \begin{proof}
     Let $x_1<x_2$ with $f(x_1)>f(x_2)$.
     Convexity implies that the function $\gamma(x)=\frac{f(x)-f(x_1)}{x-x_1}$ is non-decreasing.
     Given that $f(x_1)>f(x_2)$ we have \mbox{$\gamma(x_2)<0$} so $\lim_{x\to-\infty}\gamma(x)=l\in[-\infty,0)$. The case where $l=-\infty$ is straightforward.
     Assume now that $l>-\infty$.
     This yields: $\frac{f(x_1)-f(x)}{x_1-x}=l+o(1)$ as $x\to-\infty$, thus $f(x)=lx+o(x)$ as $x\to-\infty$ and letting $x\to-\infty$ we obtain $\lim_{x\to-\infty}f(x)=+\infty$.
     Considering $x_3<x_4$ with $f(x_3)<f(x_4)$ we can prove as above that $\lim_{x\to+\infty}f(x)=+\infty$.
     So we establish coercivity of $f$.
     Now we know that $f$ is continuous on the entire $\mathbb{R}$ because it is convex, so coercivity ensures that $f$ admits a minimum.
     Finally, by standard theory we know the set $M:=\argmin_{y\in\mathbb{R}}{f(y)}\neq\emptyset$ is convex, hence it is a point or an interval.
     In particular, it must be bounded given that $f$ is coercive.
     Indeed, if $M=(-\infty,c]$ or $[c',+\infty)$ for some $c,c'\in\mathbb{R}$ then $f$ could not satisfy $\lim_{|x|\to+\infty}f(x)=+\infty.$
     We prove the closure. Let $(x_n)_{n\in\mathbb{N}}\subseteq M$ s.t.\ $x_n\to x$.
     Then $f(x_n)=m:=\ds\min_{y\in\mathbb{R}}f(y) \ \ \forall n\in\mathbb{N}$ thus $\lim_{n\to\infty}f(x_n)=m$ and continuity of $f$ yields $x\in M$, hence the thesis.
    \end{proof}
    \begin{proof}[Proof of Corollary~\ref{corCSA}]
    The first thesis is a simple adaptation of the previous results. If $g$ does not depend on $y$,  with the same notation used in the proof of Theorem \ref{SuCO}, for the star-shaped case (and similarly when $g$ is positively homogeneous) we can choose the generator $g^{\mu}$ by defining:
    \begin{align*}
                  g^{\mu}(\omega,t,y,z)&=\inf_{\bar{z}\in\mathbb{R}^n}\{b_k(z-\bar{z})+g_{\mu}(\omega,t,\bar{z})\} \\
                  &=\inf_{m\geq0}\{k(|z-m\mu|)+mg(\omega,t,\mu)\},
              \end{align*}
with $b_k=k|z|.$ Clearly, $g^{\mu}$ does not depend on $y$.
Choosing $$\Gamma:=\{\delta \mbox{ predictable square integrable stochastic process valued in } \mathbb{R}^n\},$$
and proceeding analogously as in the proof of Theorem~\ref{SuCO} we can conclude that $\rho_t^{\delta}$ is a dynamic convex and cash-additive risk measure (because its driver is convex and does not depend on $y$) and the minimum is achieved when $\delta\equiv Z$, where $Z$ is the second component of the solution to Equation~\eqref{eq: star-shapedBSDE} with driver $g$.
Hence, we have proved the representation in Equation~\eqref{minmax1} (of which the right-hand side directly follows from Equation~\eqref{eq:dualSA}).

As far as cash-subadditivity is concerned, we know from Theorem~\ref{SuCO} that for any $(\alpha,\delta)\in\Gamma$ the drivers $g^{\alpha,\delta}(t,y,z)$ are convex (if $g$ is star-shaped, otherwise they are sublinear, but the proof remains the same) and satisfy the standard assumptions.
Moreover, we have proved that $g^{\alpha,\delta}(t,\rho_t,Z_t)\geq g(t,\rho_t,Z_t)$ and $g^{\rho,Z}(t,\rho_t,Z_t)=g(t,\rho_t,Z_t)$.
We want to prove that for any $(\alpha,\delta)\in\Gamma$ there exists $\tilde{g}^{\alpha,\delta}$ that verifies the standard assumptions, convexity, decreasing monotonicity w.r.t.\ $y$, and, for any $(\alpha,\delta)\in\Gamma$,  $\tilde{g}^{\alpha,\delta}\geq g \
 d\mathbb{P}\times dt$-a.s., where equality holds if $(\alpha,\delta)=(\rho,Z)$.
 By the comparison theorem, if these assertions are true for any $(\alpha,\delta)\in\Gamma$, the first component of the solution $(\tilde{\rho}^{\alpha,\delta})_{t\in[0,T]}$ to the BSDE generated by the driver $\tilde{g}^{\alpha,\delta}$ is cash-subadditive and Equation~\eqref{minmax2} is verified.
 Let us define for each fixed $(\alpha,\delta)\in\Gamma$, $(\omega,t,y,z)\in\Omega\times[0,T]\times\R\times\R^n$:
\begin{equation*}
    \tilde{g}^{\alpha,\delta}(\omega,t,y,z):=\inf\left\{g^{\alpha,\delta}(\omega,t,\bar{y},z):\bar{y}\leq y\right\}.
\end{equation*}
Then $\tilde{g}^{\alpha,\delta}$ is a predictable process, given that continuity of $g^{\alpha,\delta}$ w.r.t.\ $y$ ensures that the infimum can be taken over a countable set $(y_n)_{n\in\mathbb{N}}$, thus the infimum over a countable set of predictable processes is still predictable.
By definition, $\tilde{g}^{\alpha,\delta}$ is decreasing in $y$ $d\mathbb{P}\times dt$-a.s., while convexity w.r.t.\ $(y,z)$ can be verified as in Proposition~\ref{staticSS}.
Moreover, by definition $\tilde{g}^{\alpha,\delta}(\omega,t,y,z)\leq g^{\alpha,\delta}(\omega,t,y,z)$ and in particular $\tilde{g}^{\alpha,\delta}(t,\rho_t,Z_t)\leq g^{\alpha,\delta}(t,\rho_t,Z_t) \ d\mathbb{P}\times dt$-a.s.
Furthermore, fixing $(\alpha,\delta)\in\Gamma$,$(\omega,t,y,z)\in\Omega\times[0,T]\times\R\times\R^n$, we have for any $\bar{y}\leq y$:
$$g^{\alpha,\delta}(\omega,t,\bar{y},z)\geq g(\omega,t,\bar{y},z)\geq g(\omega,t,y,z),$$
where first inequality follows from the proof of Theorem~\ref{SuCO} and the second inequality is due to decreasing monotonicity of $g$.
Taking the infimum over $\bar{y}\leq y$ we obtain $\tilde{g}^{\alpha,\delta}(\omega,t,y,z)\geq g(\omega,t,y,z)$.
Summing up, we have for any $(\alpha,\delta)\in\Gamma$ $g\leq \tilde{g}^{\alpha,\delta}\leq g^{\alpha,\delta} \ \ d\mathbb{P}\times dt$-a.s., hence $\tilde{g}^{\alpha,\delta}(t,0,0)\in L^2_{\F}(T,\R)$ and $g(t,\rho_t,Z_t)\leq \tilde{g}^{\rho,Z}(t,\rho_t,Z_t)\leq g^{\rho,Z}(t,\rho_t,Z_t)=g(t,\rho_t,Z_t)$, so $g(t,\rho_t,Z_t)=\tilde{g}^{\rho,Z}(t,\rho_t,Z_t)$.

We need to check that $\tilde{g}^{\alpha,\delta}$ is Lipschitz with constant $k.$
    It is enough to verify that for almost all $(\omega,t)\in\Omega\times [0,T]$ we have $|\tilde{g}^{\alpha,\delta}(\omega,t,y_1,z_2)-\tilde{g}^{\alpha,\delta}(\omega,t,y_2,z_2)|\leq k|y_1-y_2|$ and $|\tilde{g}^{\alpha,\delta}(\omega,t,y_1,z_1)-\tilde{g}^{\alpha,\delta}(\omega,t,y_1,z_2)|\leq k|z_1-z_2|$ $\forall y_1,y_2\in\R$, $z_1,z_2\in\R^n$.
If ${g}^{\alpha,\delta}$ is already decreasing w.r.t.\ $y$ then by definition $\tilde{g}^{\alpha,\delta}\equiv{g}^{\alpha,\delta}$ and Lipschitzianity follows from the same property of $g^{\alpha,\delta}$.
If ${g}^{\alpha,\delta}$ is increasing w.r.t.\ $y$ then for almost all $(\omega,t)\in\Omega\times[0,T]$ and $(y,z)\in\R\times\R^n$ $\tilde{g}^{\alpha,\delta}(\omega,t,y,z)=\lim_{\gamma\to-\infty}g^{\alpha,\delta}(\omega,t,\gamma,z)>+\infty$, so $\tilde{g}^{\alpha,\delta}$ does not depend on $y$ and Lipschitzianity w.r.t.\ $z$ follows from the same property of $g^{\alpha,\delta}$.
So, it only remains to check the case of $g$ non-monotone w.r.t.\ $y$.
We want to prove that  for almost all $(\omega,t)\in\Omega\times[0,T]$ the Lipschitz constant $k$ of $g$ is also a Lipschitz constant for $\tilde{g}^{\alpha,\delta}$.
We can fix $y_1,y_2\in\R$ and $z\in\R^n$.
We know by Lemma~\ref{LemmaConvex} that $g^{\alpha,\delta}(\omega,t,\cdot,z)$ admits a minimum and the set $M$ of minimum points is compact.
Let us call $x^*:=\min M$.
If $y_1,y_2\geq x^*$ then $\tilde{g}^{\alpha,\delta}(\omega,t,y_1,z)=\tilde{g}^{\alpha,\delta}(\omega,t,y_2,z)=g^{\alpha,\delta}(\omega,t,x^*,z)$ hence: $$|\tilde{g}^{\alpha,\delta}(\omega,t,y_1,z)-\tilde{g}^{\alpha,\delta}(\omega,t,y_2,z)|=0\leq k|y_1-y_2|.$$
If $y_1,y_2\leq x^*$ then $\tilde{g}^{\alpha,\delta}(\omega,t,y_1,z)={g}^{\alpha,\delta}(\omega,t,y_1,z)$ and $\tilde{g}^{\alpha,\delta}(\omega,t,y_2,z)={g}^{\alpha,\delta}(\omega,t,y_2,z)$ and Lipschitzianity follows from Lipschitzianity of $g^{\alpha,\delta}$.
Finally, if $y_1\leq x^*\leq y_2$ (or the converse) we have $\tilde{g}^{\alpha,\delta}(\omega,t,y_1,z)={g}^{\alpha,\delta}(\omega,t,y_1,z)$ and $\tilde{g}^{\alpha,\delta}(\omega,t,y_2,z)={g}^{\alpha,\delta}(\omega,t,x^*,z)$, so that
\begin{align*}
|\tilde{g}^{\alpha,\delta}(\omega,t,y_1,z)-\tilde{g}^{\alpha,\delta}(\omega,t,y_2,z)|&=|{g}^{\alpha,\delta}(\omega,t,y_1,z)-{g}^{\alpha,\delta}(\omega,t,x^*,z)|\\&
\leq k|y_1-x^*|\leq k|y_1-y_2|.
\end{align*}
Summing up, for almost all $(\omega,t)\in\Omega\times[0,T]$ the Lipschitz constant $k$ of $g$ also verifies for all $y_1,y_2\in\R$ and $z\in\R^n$: $$|\tilde{g}^{\alpha,\delta}(\omega,t,y_1,z)-\tilde{g}^{\alpha,\delta}(\omega,t,y_2,z)|\leq k|y_1-y_2|.$$
Analogously we can prove that for almost all  $(\omega,t)\in\Omega\times[0,T]$ we have: $$|\tilde{g}^{\alpha,\delta}(\omega,t,y,z_1)-\tilde{g}^{\alpha,\delta}(\omega,t,y,z_1)|\leq k|z_1-z_2|, \ \ \forall y\in\R,z_1,z_2\in\mathbb{R}^n.$$ \end{proof}

\begin{proof}[Proof of Proposition \ref{PM}]
   The proofs of monotonicity, cash-(sub)additivity and continuity from above are readily verified.
   The proof of star-shapedness (resp.\ positive homogeneity) follows verbatim from the proof provided in Lemma~\ref{Lemma:infSS}.
   Here we only prove regularity and time-consistency. \medskip\\
   \textit{Regularity:} Let us consider $A\in\mathcal{F}_t$ and $X\in L^{\infty}(\mathcal{F}_T)$.
    We want to verify $\mathbb{I}_A\rho_t(X)=\mathbb{I}_A\rho_t(\mathbb{I}_AX).$
    There exist $\gamma_1,\gamma_2\in\Gamma$ such that $\rho_t(X)=\rho_t^{\gamma_1}(X)$ and $\rho_t(\mathbb{I}_AX)=\rho_t^{\gamma_2}(\mathbb{I}_AX)$.
    We have:
    $$\mathbb{I}_A\rho_t(X)=\mathbb{I}_A\rho_t^{\gamma_1}(X)=\mathbb{I}_A\rho_t^{\gamma_1}(\mathbb{I}_AX)\geq \mathbb{I}_A\rho_t(\mathbb{I}_AX),$$
    where the second equality follows from regularity\footnote{We recall that monotonicity and cash-subadditivity (resp.\ cash-additivity) of the family $(\rho^{\gamma}_t)_{\gamma\in\Gamma}$ implies regularity of $\rho^{\gamma}_t$ for each $\gamma\in\Gamma$, as we have proved in Lemma~\ref{REGCS}.} of $\rho_t^{\gamma_1}(X)$ and the inequality is due to Equation~\eqref{minC}.
   Analogously, we obtain:
   $$\mathbb{I}_A\rho_t(\mathbb{I}_AX)=\mathbb{I}_A\rho_t^{\gamma_2}(\mathbb{I}_AX)=\mathbb{I}_A\rho_t^{\gamma_2}(X)\geq\mathbb{I}_A\rho_t(X),$$
hence $\mathbb{I}_A\rho_t(X)=\mathbb{I}_A\rho_t(\mathbb{I}_AX).$ \medskip\\
\textit{Time-consistency:} Now we assume the further condition given in Equation~\eqref{TC}.
For each $X\in L^{\infty}(\mathcal{F}_T)$, by conditions \eqref{minC} and \eqref{TC}, there exist $\gamma_1,\gamma_2\in\Gamma$ such that $\rho_t(X)=\rho_t^{\gamma_1}(X)$ and $\rho_s(\rho_t(X))=\rho_s^{\gamma_2}(\rho^{\gamma_2}_t(X))$.
These equalities imply:
\begin{equation}
\rho_s(\rho_t(X))=\rho_s^{\gamma_2}(\rho_t^{\gamma_2}(X))=\rho_s^{\gamma_2}(X)\geq \rho_s^{\gamma_1}(X)=\rho_s(X),
\label{WTC1}
\end{equation}
where the second equality is due to the time-consistency of $\rho_t^{\gamma_2}$, while the inequality follows from the minimum condition \eqref{minC}.
Conversely, by the minimum condition there exists $\gamma_3\in\Gamma$ such that $\rho_s(\rho_t(X))=\rho^{\gamma_3}_s(\rho_t(X))$.
This yields:
\begin{equation}
\rho_s(\rho_t(X))=\rho_s^{\gamma_3}(\rho_t(X))=\rho_s^{\gamma_3}(\rho_t^{\gamma_1}(X))\leq\rho_s^{\gamma_1}(\rho_t^{\gamma_1}(X))=\rho_s^{\gamma_1}(X)=\rho_s(X).
\label{WTC2}
\end{equation}
Here, the second and last equalities are due to the relation $\rho_t(X)=\rho_t^{\gamma_1}(X)$, the inequality is due to the minimum condition~ \eqref{minC}, while the third equality holds given the time-consistency of $\rho_t^{\gamma_1}$.
The inequalities in Equations~\eqref{WTC1} and \eqref{WTC2} yield the time-consistency of $\rho_t$.
\end{proof}
\begin{remark}
   It is evident that certain statements in Proposition~\ref{PM} can be proven independently of a subset of the hypotheses outlined in the same proposition.
   Specifically, the properties of monotonicity, regularity, cash-subadditivity (or cash-additivity), star-shapedness, and time-consistency, along with condition~\eqref{TC} of $(\rho^{\gamma})_{\gamma\in\Gamma}$, are sufficient to establish the corresponding properties of $\rho_t$, without the need for additional assumptions.
   Furthermore, assuming monotonicity of $(\rho^{\gamma})_{\gamma\in\Gamma}$, the continuity from above of $(\rho^{\gamma})_{\gamma\in\Gamma}$ also guarantees the same property for $\rho_t$.
   Moreover, we obtain the same theses if we consider unbounded random variables in the definition of $\rho_t$, i.e., \mbox{$\rho_t:L^{p}(\mathcal{F}_T)\to L^{p}(\mathcal{F}_t)$} with \mbox{$p\in[1,+\infty).$}
    \end{remark}
\begin{proof}[Proof of Proposition~\ref{POBSDEs}]
    By Proposition~\ref{PM} we know that $\rho_t$ is monotone, regular, time-consistent and $\rho_T(X)=X$, inheriting the same properties from the family $(\rho^{\gamma})_{\gamma\in\Gamma}$.
    Then $\rho_t$ is a nonlinear evaluation w.r.t.\ the Brownian filtration $(\mathcal{F}_t)_{t\in[0,T]}$, in the sense of Definition~2.1 in \cite{P05}.
    Moreover, given the convexity (resp.\ sublinearity) and normalization of the family $(\rho^{\gamma})_{\gamma\in\Gamma}$, $\rho_t$ is star-shaped (resp.\ positively homogeneous) and $\rho_t(0)=0 \ \forall t\in[0,T]$.
    We want to apply Theorem~3.1 of \cite{P05}.
    We only need to verify the dominance condition:
    \begin{align}
        &\rho_t(X)-\rho_t(Y)\leq \rho_t^{b_k}(X-Y), \ \forall X,Y\in L^{2}(\mathcal{F}_T),
        \label{D2}
    \end{align}
    where $k$ is the equi-Lipschitz constant for the family $(\rho^{\gamma})_{\gamma\in\Gamma}$, $b_k$ is the function \mbox{$b_k(y,z):=k(|y|+|z|)$} and $\rho^{b_k}_t(X)$ is the first component of the solution of the BSDE with terminal condition $X$ and driver $b_k$.
    Fixing $X,Y\in L^{2}(\mathcal{F}_T)$ by condition \eqref{minC1} there exist  $\gamma_1,\gamma_2\in\Gamma$ such that $\rho_t(X)=\rho_t^{\gamma_1}(X)$ and $\rho_t(Y)=\rho_t^{\gamma_2}(Y)$.
    We are ready to prove condition \eqref{D2}: $$\rho_t(X)-\rho_t(Y)=\rho_t^{\gamma_1}(X)-\rho_t^{\gamma_2}(Y)\leq \rho_t^{\gamma_2}(X)-\rho_t^{\gamma_2}(X)\leq \rho_t^{b_k}(X-Y).$$
The first inequality follows from the minimum condition \eqref{minC1}, while the second inequality is due to Corollary~4.4 in \cite{P05}, observing that $\rho^{\gamma}$ is induced via BSDEs for each $\gamma\in\Gamma$ and the drivers $g^{\gamma}$ are equi-Lipschitz with constant $k$.
The thesis of Theorem~3.1 in \cite{P05} ensures there exists a driver $g$ that satisfies standard assumptions for existence and uniqueness of the solution, with $g(\cdot,0,0)\equiv 0$, such that:
$$\rho_t(X)=X+\int_t^Tg(s,\rho_s,Z_s)ds-\int_t^TZ_sdW_s.$$
The star-shapedness (resp.\ positive homogeneity) of the driver $g$ follows from Proposition~\ref{POSS} and star-shapedness (resp.\ positive homogeneity) of $\rho_t$.
    \end{proof}
      \begin{remark}
        We stress that, when the family of drivers is convex, the hypothesis of normalization on $g^{\gamma}(\cdot,0,0)$ is not necessary (while if the family $(g^{\gamma})_{\gamma\in\Gamma}$ is sublinear this hypothesis is automatically satisfied).
        Indeed, if $g^{\gamma}(\cdot,0,0)\not\equiv0$ we can apply Corollary~5.12 of  \cite{P05}.
        Let us suppose \mbox{$g^{\gamma}(\cdot,0,0):=\tilde{g}(t)\in L^{2}_{\mathcal{F}}(T;\mathbb{R})$} for any $\gamma\in\Gamma$. By the minimum condition~\eqref{minC1} we know that $\rho_t(0)=\rho_t^{\gamma}(0)$ for any $t\in[0,T]$ and $\gamma\in\Gamma$.
        Then one can easily verify the constraint in Equation (5.13) of \cite{P05}:
        $$\rho_t^{-b_k+\tilde{g}}(0)\leq\rho_t(0)\leq\rho_t^{b_k+\tilde{g}}(0),$$ where $b_k:=k(|y|+|z|)$ with $k$ being the equi-Lipschitz constant of the family $(g_{\gamma})_{\gamma\in\Gamma}$.
        The constraint follows from the Lipschitz condition:
        $$|{g}^{{\gamma}}(t,y,z)-\tilde{g}(t)|\leq k(|y|+|z|) \ \forall \gamma\in\Gamma,$$
        which yields $$-b_k+\tilde{g}\leq g^{{\gamma}} \leq b_k+\tilde{g} \ \forall \gamma\in\Gamma,$$ thus by the comparison theorem we have the bounds:
        $$\rho_t^{-b_k+\tilde{g}}(0)\leq\rho^{\tilde{\gamma}}_t(0)=\rho_t(0)\leq\rho_t^{b_k+\tilde{g}}(0) \ \ \forall t\in[0,T].$$
        By Corollary~5.12 of \cite{P05} the nonlinear evaluation $\rho_t$ is induced by a BSDE with a driver $g$ such that $g(t,0,0)=\tilde{g}(t) \ \forall t\in[0,T]$. \\
       Furthermore, assuming the additional constraints: $X\in L^{\infty}(\mathcal{F}_T)$ and $g(\cdot,0,0)\in\text{BMO}(\mathbb{P})$, the theses outlined in Proposition \ref{POBSDEs} remain valid. Specifically, under these assumptions, we have \mbox{$\rho_t:L^{\infty}(\mathcal{F}_T)\to L^{\infty}(\mathcal{F}_t)$.}  This observation applies equally to Corollary \ref{cor:CSA}.
    \end{remark}
    \begin{proof}[Proof of Proposition \ref{IFFCS}]
    If $g$ is decreasing w.r.t.\ $y$ then $\rho_t$ is cash-subadditive (cf.\ Proposition~7.3 in \cite{ELKR09}).
        Let us prove the converse statement.
        Fixing $(y,z)\in\mathbb{R}\times\mathbb{R}^n$, $h\in\mathbb{R}$ and $t\in[0,T)$, we define for any sufficiently small $\varepsilon>0$:
        \begin{align*}
        g_{\varepsilon}(t,y+h,z)=\frac{1}{\varepsilon}\left({\rho}_{t,t+\varepsilon}(y+h+z(W_{t+\varepsilon}-W_t))-(y+h)\right).
        \end{align*}
        By cash-subadditivity of $\rho_t$ we infer $g^{h}_{\varepsilon}(t,y,z):=g_{\varepsilon}(t,y+h,z)$ is decreasing in $h$, namely for $h_1\geq h_2$ we have that $g^{h_1}_{\varepsilon}(t,y,z)\leq g^{h_2}_{\varepsilon}(t,y,z) \ d\mathbb{P}\mbox{-a.s.}$
        Using the same arguments as in the proof of Proposition~\ref{POSS} we infer $$d\mathbb{P}\times dt\mbox{-a.s. } g(t,y+h_1,z)\leq g(t,y+h_2,z), \ \forall (y,z)\in\mathbb{R}\times\mathbb{R}^n, \forall h_1\geq h_2.$$
        Hence, $d\mathbb{P}\times dt$-a.s.\ $\forall (y,z)\in\mathbb{R}\times\mathbb{R}^n$ $g(t,y+h,z)$ is decreasing in $h\in\mathbb{R}$, which implies that $d\mathbb{P}\times dt$-a.s. $g$ is decreasing in $y$ for any $z\in\mathbb{R}^n$.
    \end{proof}
  \begin{proof}[Proof of Corollary \ref{cor:CSA}]
       Given that $\rho_t$ is cash-additive, the hypotheses of Theorem~4.4 in \cite{P04} are satisfied.
       Then the driver $g$ does not depend on $y$ and $g(t,0)\equiv0$ for all $t\in[0,T].$
       The second part of the thesis follows by observing that $\rho_t$ is cash-subadditive.
       Then by Proposition~\ref{IFFCS} the driver $g$ is decreasing w.r.t.\ $y$.
    \end{proof}

    \setcounter{theorem}{0}

    \subsection{Proofs of Section~\ref{sec:supsolution}}
       \begin{lemma}
       If $g$ satisfies SA plus star-shapedness and $X\in L^0(\mathcal{F}_T)$ then the corresponding $\mathcal{E}^{g}_t$ is star-shaped.
   \end{lemma}
   \begin{proof}
       Fixing $X\in L^0(\mathcal{F}_T)$, if $\mathcal{A}(X,g)=\emptyset$ there is nothing to prove.
       Suppose now that $\mathcal{A}(X,g)\neq\emptyset$. For any $\lambda\in[0,1]$ the star-shapedness of $g$ yields $\lambda(Y,Z)\in\mathcal{A}(\lambda X,g)$ when $(Y,Z)\in\mathcal{A}(X,g)$, thus $\lambda\mathcal{A}(X,g)\subseteq\mathcal{A}(\lambda X,g)$, which yields the thesis.
   \end{proof}
\begin{proof}[Proof of Theorem \ref{superapresentation}]
Let us fix $X\in L^{\infty}(\mathcal{F}_T)$.
If $\mathcal{A}(X,g)=\emptyset$ there is nothing to prove.
If $\mathcal{A}(X,g)\neq\emptyset$ we define, similarly as done in Theorem~\ref{SuCO}, for all $(\beta,\mu)\in\mathbb{R}\times\mathbb{R}^n$ the function

              \begin{equation*}
                  g_{\beta,\mu}(\omega,t,y,z):=
                  \begin{cases}
                      mg(\omega,t,\beta,\mu) \ &\mbox{if there exists } m\in[0,1] \mbox{ s.t. } (y,z)=m(\beta,\mu) \\
                      +\infty \ &\mbox{otherwise}.
                  \end{cases}
              \end{equation*}
              We notice that for any $(\beta,\mu)$ the function $g_{\beta,\mu}$ is a proper, lower semicontinuous and convex function w.r.t.\ $(y,z)$ $d\mathbb{P}\times dt$-a.s.
              Moreover, according to Proposition~\ref{staticSS}, $g_{\beta,\mu}(t,y,z)\geq g(t,y,z)\geq 0$, thus $g_{\beta,\mu}$ satisfies positivity, while normalization follows from  $g_{\beta,\mu}(\cdot,0,0)\equiv g(\cdot,0,0)\equiv0$.
              Let us define the space $\Gamma:=\left\{(\alpha,\delta)\in\mathcal{S}\times\mathcal{L}: \delta \mbox{ is admissible}\right\}$.
              For each $\gamma:=(\alpha,\delta)\in\Gamma$ the driver is defined by:
              \begin{equation*}
                  g^{\gamma}(\omega,t,y,z):=g_{\alpha,\delta}(\omega,t,y,z):=g_{\alpha_t(\omega),\delta_t(\omega)}(\omega,t,y,z),
              \end{equation*}
              which is a $\mathcal{P}\otimes\B(\R)\otimes\B(\R^n)$-measurable function (given that $g$ is $\mathcal{P}\otimes\B(\R)\otimes\B(\R^n)$-measurable and $(\alpha,\delta)$ is predictable), that also inherits the properties of convexity, lower semicontinuity, positivity and normalization from $g_{\beta,\mu}$.
              Fixing $\gamma\in\Gamma$, if $\mathcal{A}(X,g^{\gamma})\neq\emptyset$, by Theorem~\ref{EUsuper}, there exists a unique minimal supersolution $(\mathcal{E}^{g^{\gamma}}(X),Z^{\gamma})$ to Equation \eqref{Defsupersolution} with driver $g^{\gamma}$ and Proposition 2.2 in \cite{DKRT14} ensures that $\mathcal{E}^{g^{\gamma}}$ is convex.
              In addition, we have $g^{\gamma}(\omega,t,y,z)\geq g(\omega,t,y,z)$, thus by the comparison principle for minimal supersolutions (see Proposition~3.3 in \cite{DHK13}) it results that
              \begin{equation}
                  \mathcal{E}^{g^{\gamma}}_t(X)\geq \mathcal{E}^{g}_t(X) \ \ \mbox{ for any } t\in[0,T].
                  \label{convin}
                    \end{equation}
                    Clearly, inequality \eqref{convin} is true also when $\mathcal{A}(X,g^{\gamma})=\emptyset$.
                    Indeed by definition $\mathcal{E}^{g^{\gamma}}_t(X)=+\infty$.
                    Furthermore, considering $\bar{\gamma}:=(\mathcal{E}^g(X),\bar{Z})\in\Gamma$ it holds that $$g^{\bar{\gamma}}(t,\mathcal{E}^g_t(X),\bar{Z_t})=g_{\mathcal{E}^g(X),\bar{Z}}(t,\mathcal{E}^g_t(X),\bar{Z}_t)= g(t,\mathcal{E}^g_t(X),\bar{Z}_t).$$
                    In particular, if $\mathcal{A}(X,g)\neq\emptyset$, then $\mathcal{A}(X,g^{\bar{\gamma}})\neq\emptyset$.
                    Indeed, the equality $g^{\bar{\gamma}}(t,\mathcal{E}^g_t(X),\bar{Z}_t)\equiv g(t,\mathcal{E}^g_t(X),\bar{Z})$ implies that the relation:
              $$\begin{cases}
              Y^{\bar{\gamma}}_s-\int_s^tg^{\bar{\gamma}}(u,Y^{\bar{\gamma}}_u,Z^{\bar{\gamma}}_u)du+\int_s^tZ^{\bar{\gamma}}_udW_u\geq Y^{\bar{\gamma}}_t \ \ \mbox{ for every } 0\leq s \leq t \leq T, \\
              Y^{\bar{\gamma}}_T=X
              \end{cases}$$
              is satisfied by $(Y^{\bar{\gamma}},Z^{\bar{\gamma}})=(\mathcal{E}^g(X),\bar{Z})$.
              This means that $(\mathcal{E}^g(X),\bar{Z})\in\mathcal{A}(X,g^{\bar{\gamma}})$, so $\mathcal{A}(X,g^{\bar{\gamma}})\neq\emptyset$ and, by definition of minimal supersolutions, $\mathcal{E}_t^g(X)\geq \mathcal{E}_t^{g^{\bar{\gamma}}}(X)$ for any $t\in[0,T]$.
              The last relation together with the inequality \eqref{convin} implies \mbox{$\mathcal{E}_t^{g}(X)=\mathcal{E}^{g^{\bar{\gamma}}}_t(X)$} for any $t\in[0,T]$, hence Equation~\eqref{minsuper} follows.
              Equation~\eqref{minmaxsuper} is due to Theorem~3.4 of \cite{DKRT14}.

              The last statements for positively homogeneous $g$ are a straightforward adaptation of the proof given above, making use of Corollary~3.11 of \cite{DKRT14} for the min-max representation.
              \end{proof}
              \begin{remark}
                  We stress that fixing $X\in L^{\infty}(\mathcal{F}_T)$ and $\gamma\in\Gamma$ we can not infer \textit{a priori} that $\mathcal{A}(X,g^{\gamma})\neq\emptyset$, even when $\mathcal{A}(X,g)\neq\emptyset$.
                  Nevertheless, it is enough that for each $X\in L^{\infty}(\mathcal{F}_T)$ there exists $\bar{\gamma}\in\Gamma$ such that  $\mathcal{A}(X,g^{\bar{\gamma}})\neq\emptyset$ and $\mathcal{E}^{g^{\bar{\gamma}}}_t(X)\leq\mathcal{E}^{g}_t(X)$.
              \end{remark}

\setcounter{theorem}{0}

\subsection{Proofs of Section~\ref{sec:app}}
\begin{proof}[Proof of Proposition~\ref{prop:AS}]
 The proof follows a similar line of reasoning as Proposition~10 in \cite{MGRO22}.
 Monotonicity, normalization and full allocation can be verified by direct inspection.
 We define a function $F:[0,1]\to L^{\infty}(\mathcal{F}_t)$ as $F(m)=\rho^{\gamma^X}_t(m X)$, where $t\in[0,T]$ and $X\in L^{\infty}(\mathcal{F}_T)$, with $\gamma^X\in\Gamma$ satisfying $\rho_t(X)=\min_{\gamma\in\Gamma}\rho_t^{\gamma}(X)=\rho^{\gamma^X}_t (X)$ (cf.\ Equation~\eqref{minsub}). Here, $\rho_t^{\gamma^X}$ is a convex risk measure.
 The convexity of $\rho_t^{\gamma^X}$ implies the convexity of $F$.
 Similar to \cite{RZ23}, it can be shown that
$$\int_0^1 F'_{-}(m)dm=\Lambda_t^{AS}(X,X) = \int_0^1 F'_{+}(m)dm=F(1)-F(0)=\rho^{\gamma^X}_t(X),$$
thus $\Lambda^{AS}_t$ is a CAR.
 Furthermore we have that the measure $\mu^{\cdot}_t\triangleq D_t^{\beta^{\cdot}}\mathbb{Q}^{q^{\cdot}}_{t}$ satisfies $$\frac{d\mu^{mY}_t}{d\mathbb{P}}=\exp\left(-\frac{1}{2}\int_t^T |q_s^{mY}|^2ds-\int_t^T\beta_s^{mY} ds -\int_t^Tq_s^{mY}dW_s\right)= L^{mY}(T,t).$$
 Hence, using Fubini's Theorem and Equation~\eqref{ASCAR}, we obtain  $\Lambda^{AS}_t(X,Y)=\mathbb{E}_{\mathbb{P}}\left[\left.\hat{L}^Y(T,t)X\right|\mathcal{F}_t\right]$.
Now let us consider the cash-subadditive case. By fixing $Y\in L^{\infty}(\mathcal{F}_T)$ and utilizing Corollary \ref{corCSA}, we observe that $\rho_t^{\gamma^Y}$ is a cash-subadditive convex risk measure.
Therefore, $\beta^{m Y}_t\geq 0$ for any $m \in[0,1]$, based on the representation~\eqref{minmax2}.
The proof can then follow verbatim from Theorem 15 in \cite{RZ23}.
Similarly, we can prove the cash-additivity of the Aumann-Shapley CAR when the underlying risk measure is cash-additive.
The results regarding the penalized Aumann-Shapley CAR follow through routine verification.
We only prove sub-allocation.
Let $(X_i)_{i=1,\dots,n},X\in L^{\infty}(\mathcal{F}_T)$ such that $\ds\sum_{i=1}^nX_i=X$.
Then:
\begin{align*}\Lambda_t^{p-AS}(\sum_{i=1}^nX_i,Y)&=\mathbb{E}_{\mathbb{Q}^{q^{Y}}_t}[D_t^{\beta^Y}\sum_{i=1}^nX_i|\mathcal{F}_t]-c^{\gamma^Y}_t(D_t^{\beta^Y}\mathbb{Q}^{q^{Y}}_t) =\sum_{i=1}^n\mathbb{E}_{\mathbb{Q}^{q^{Y}}_t}[D_t^{\beta^Y}X_i|\mathcal{F}_t]-c^{\gamma^Y}_t(D_t^{\beta^Y}\mathbb{Q}^{q^{Y}}_t) \\
&\geq\sum_{i=1}^n\mathbb{E}_{\mathbb{Q}^{q^{Y}}_t}[D_t^{\beta^Y}X_i|\mathcal{F}_t]-nc^{\gamma^Y}_t(D_t^{\beta^Y}\mathbb{Q}^{q^{Y}}_t)
=\sum_{i=1}^n\left(\mathbb{E}_{\mathbb{Q}^{q^{Y}}_t}[D_t^{\beta^Y}X_i|\mathcal{F}_t]-c^{\gamma^Y}_t(D_t^{\beta^Y}\mathbb{Q}^{q^{Y}}_t)\right) \\
&=\sum_{i=1}^n\Lambda_t^{p-AS}(X_i,Y).
\end{align*}
The inequality follows from positivity of the penalty term $c_t^{\gamma^Y}$ as defined in Equation~\eqref{penfun}.
 \end{proof}
    \begin{proof}[Proof of Proposition~\ref{prop:linear}]
    First, we note that the existence of a minimizer $(\bar{\gamma}, \bar{\pi}) \in \Gamma \times \mathcal{K}$ has already been established at the beginning of Section~\ref{sec:portchoice}.
    Furthermore, for each $\gamma \in \Gamma$, the convex optimization problem
    $$V^{\gamma}_t(F) = \essmin_{\pi \in \mathcal{K}} \rho^{\gamma}_t(X^{\pi}_T + F)$$
can be solved using methods similar to those employed in the case of linear utility (refer to Theorem~4.3 in \cite{LS14}). This result is possible due to the cash-additivity of $\rho^{\gamma}_t$, which is ensured by Corollary~\ref{corCSA}. Specifically, the optimal portfolio for each $\gamma \in \Gamma$ is given by the first component of the solution to Equation (16) in \cite{LS14}, without considering jumps and employing the sign conventions for disutilities.
More precisely, we have:
\begin{equation}
V_t^{\gamma}(F) = X^{\bar{\pi}^{\gamma}}_t + Y_t^{\gamma},
\label{eq:convopt1}
\end{equation}
where $\bar{\pi}^{\gamma}$ is the optimal strategy corresponding to the fixed $\gamma \in \Gamma$, and $Y^{\gamma}_t$ follows the dynamics:
$$Y^{\gamma}_t = F + \int_t^T \tilde{g}^{\gamma}(s, Z^{\gamma}_s) ds - \int_t^T Z_s^{\gamma} dW_s,$$
with $\tilde{g}^{\gamma}(t, Z^{\gamma}_t) = \bar{\pi}^{\gamma}_t b_t + g^{\gamma}(t, Z^{{\gamma}}_t + \bar{\pi}^{\gamma}_s \sigma_s)$.
Here, $g^{\gamma}$ is the driver that induces the dynamics of the risk measure $\rho_t^{{\gamma}}$.
Indeed, by considering the dual representation of $\rho^{\gamma}_t$ (cf.\ Equation \eqref{eq:dualSA}), we find that the convex conjugate of $G^{\gamma}$ is $g^{\gamma}$.
Therefore, we can conclude that $g^{\gamma}$ satisfies Equation (12) in \cite{LS14}.
Moreover, $g^{\gamma}$ also satisfies the hypotheses (H1), (H2), and (H3) stated in \cite{LS14}, as $g^{\gamma}$ fulfills the conditions $g^{\gamma} \geq g \geq 0$, $g^{\gamma}(t,0)\equiv0$, and $g^{\gamma}$ is Lipschitz.
Hence, the solution to the convex optimization problem is given by Equation \eqref{eq:convopt1}, as asserted in Theorem 4.3 in~\cite{LS14}.
Taking the minimum over $\gamma \in \Gamma$, we obtain:
$$V_t(F) = \essmin_{\gamma \in \Gamma} \{X_t^{\bar{\pi}^{\gamma}} + Y^{\gamma}_t\} = X^{\bar{\pi}}_t + Y^{\bar{\gamma}}_t,$$
thus proving the desired result.
\end{proof}

\end{document}